\documentclass{theoretics}

\title{Determinantal Sieving}

\ThCSauthor[London]{Eduard Eiben}{Eduard.Eiben@rhul.ac.uk}[0000-0003-2628-3435]
\ThCSauthor[Berlin]{Tomohiro Koana}{tomohiro.koana@gmail.com}[0000-0002-8684-0611]
\ThCSauthor[London]{Magnus Wahlström}{Magnus.Wahlstrom@rhul.ac.uk}[0000-0002-0933-4504]

\ThCSaffil[London]{Dept. of Computer Science, Royal
    Holloway, University of London, United Kingdom}
\ThCSaffil[Berlin]{Algorithmics and Computational Complexity, Technische Universität Berlin, Germany}

\ThCSthanks{A preliminary version of this work was presented at SODA 2024 \cite{EibenKW24}.}
\ThCSshortnames{E.\ Eiben, T.\ Koana, M.\ Wahlström}  
\ThCSshorttitle{Determinantal Sieving}
\ThCSyear{2025}
\ThCSarticlenum{21}
\ThCSreceived{Aug 5, 2024}
\ThCSrevised{Apr 6, 2025}
\ThCSaccepted{Jun 29, 2025}
\ThCSpublished{Oct 2, 2025}
\ThCSdoicreatedtrue
\ThCSkeywords{Parameterized complexity, FPT algorithms, Algebraic algorithms, Matroid theory}

\usepackage{thm-restate}

\usepackage{booktabs}
\usepackage{xspace}

\newcommand{\N}{\ensuremath{\mathbb N}\xspace}
\newcommand{\R}{\ensuremath{\mathbb R}\xspace}
\newcommand{\F}{\ensuremath{\mathbb F}\xspace}
\newcommand{\Q}{\ensuremath{\mathbb Q}\xspace}
\newcommand{\cI}{\ensuremath{\mathcal I}\xspace}

\newcommand{\cC}{\ensuremath{\mathcal C}\xspace}
\newcommand{\cE}{\ensuremath{\mathcal E}\xspace}

\newcommand{\cP}{\ensuremath{\mathcal P}\xspace}

\newcommand{\cF}{\ensuremath{\mathcal F}\xspace}
\newcommand{\cT}{\ensuremath{\mathcal T}\xspace}

\newcommand{\Oh}{O}

\DeclareMathOperator{\Pf}{Pf}

\DeclareMathOperator{\supp}{supp}
\DeclareMathOperator{\osupp}{osupp}

\declaretheorem[style=standard,sibling=theorem,name=redrule]{Reduction Rule}

\makeatletter
\newcommand\IfRestateTF{%
  \ifx\label\thmt@gobble@label 
    \expandafter\@firstoftwo
  \else
    \expandafter\@secondoftwo
  \fi
}
\makeatother
\newcommand{\RestateRemark}{\IfRestateTF{{\normalfont\bfseries (Restated) }}{}}

\begin{document}
\maketitle

\begin{abstract}
  We introduce a new, remarkably powerful tool to the toolbox of
  algebraic FPT algorithms, \emph{determinantal sieving}.
  Given evaluation access to a polynomial $P(x_1,\ldots,x_n)$ over a field $\F$ of characteristic 2, defined on the set of variables $X={x_1,\ldots,x_n}$, and a linear matroid $M=(X,\cI)$ over $\F$ of rank $k$, one can determine -- with $O^*(2^k)$ evaluations of $P$ (where $O^*$ suppresses factors polynomial in the input size) -- whether there exists a multilinear term in the monomial expansion of $P$ whose support forms a basis for $M$.
   The known tools of \emph{multilinear
    detection} and \emph{constrained multilinear detection} then
  correspond to the case where $M$ is a uniform matroid and 
 the truncation of a disjoint union of uniform matroids, respectively. 
  More generally, let the \emph{odd support} of a monomial $m$ be the
  set of variables which have odd degree in $m$. Using $O^*(2^k)$
  evaluations of $P$, we can sieve for those terms $m$ whose odd
  support spans $M$. Applying this framework to well-known efficiently
  computable polynomial families allows us to simplify, generalize and
  improve on a range of algebraic FPT algorithms, such as:
  \begin{itemize}
  \item Solving \textsc{$q$-Matroid Intersection} in time $O^*(2^{(q-2)k})$
    and \textsc{$q$-Matroid Parity} in time $O^*(2^{qk})$,
    improving on $O^*(4^{qk})$ for matroids represented over general
    fields (Brand and Pratt, ICALP 2021) 
  \item \textsc{Long $(s, t)$-Path} in $O^*(1.66^k)$ time,
    improving on $O^*(2^k)$ (Fomin et al., SODA 2023),
    as well as further results on paths and linkages in so-called \emph{frameworks},
    including \textsc{Rank $k$ $(S,T)$-Linkage}
    in $O^*(2^k)$ time (improving on $O^*(2^{|S|+\Oh(k^2 \log (k+|\F|))})$
    over general fields by Fomin et al.)
  \item Many instances of the \textsc{Diverse X} paradigm, finding
    a collection of $r$ solutions to a problem with a minimum
    mutual distance of $d$ in time $O^*(2^{r^2 d/2})$, improving solutions for
    \textsc{$k$-Distinct Branchings} 
    from time $2^{O(k \log k)}$ to $O^*(2^k)$ (Bang-Jensen et al., ESA 2021),
    and for \textsc{Diverse Perfect Matchings} from $O^*(2^{2^{O(rd)}})$
    to $O^*(2^{r^2d/2})$ (Fomin et al., STACS 2021)
  \end{itemize}
  For several other problems, such as \textsc{Set Cover}, 
  \textsc{Steiner Tree}, \textsc{Graph Motif} and
  \textsc{Subgraph Isomorphism}, where the current algorithms are
  either believed to be optimal or are proving exceedingly difficult
  to improve, we show matroid-based generalisations
  at no increased cost to the running time.
  In the above, all matroids are assumed to be represented over fields of
  characteristic~2 and all algorithms use polynomial space.
  Over general fields, we achieve similar results when the polynomial is given as an arithmetic circuit. However, this comes at the cost of exponential space, as the approach relies on computations within the \emph{exterior algebra}. For a class of arithmetic circuits
  we call \emph{strongly monotone}, this is even achieved without any loss
  of running time. However, the \emph{odd support} sieving result appears to
  be specific to working over characteristic~2. 
\end{abstract}


\section{Introduction} \label{sec:intro}

\emph{Algebraic algorithms} is an algorithmic paradigm with 
remarkably powerful, yet non-obvious applications, especially for the design of  exact (exponential-time) and
parameterized algorithms. To narrow the scope a bit, let us consider
more specifically what may be called the \emph{enumerating polynomial}
method. Consider a problem of looking for a particular substructure in
an object; for example, given a graph $G$, we may ask if $G$ has a
perfect matching, or a path on at least $k$ vertices, etc. (We
  focus on the decision problem. Given the ability to solve the
  decision problem, an explicit solution can be found with limited
  overhead; see Bj\"orklund et al.~\cite{BjorklundKK14witness,BjorklundKKL15alenex} 
  for deeper investigations into this.)
For surprisingly many applications, this problem can be reduced to
polynomial identity testing, by constructing a multivariate polynomial
$P(X)=P(x_1,\ldots,x_n)$ (occasionally referred to as \emph{multivariate generating polynomial}~\cite{Brand19thesis})
such that the monomials of $P$ enumerate all
instances of the desired substructure, and then testing
whether $P(X)$ contains at least one monomial or not. The latter is
the \emph{polynomial identity testing} (PIT) problem, which can be
solved efficiently in randomized polynomial time via the Schwartz-Zippel
(a.k.a.\ DeMillo-Lipton-Schwartz-Zippel) Lemma, requiring only the ability
to evaluate $P(X)$, possibly over an extension field of the original
field~\cite{Schwartz80,Zippel79}\footnote{We typically assume that we are working over a
  finite field, preferably of characteristic 2, and of a size small enough that the time for field operations does not overwhelm the running time. 
  See Section~\ref{subsec:conventions} for a discussion on this.}.
Therefore the challenge lies in
constructing an enumerating polynomial that can be
efficiently evaluated. In particular, it is a priori non-obvious why
it would be easier to construct an enumerating polynomial for a
problem than it is to simply solve the problem directly. 

In our experience, the ability to do so has two sources. First, there are
well-known families of polynomials that can be efficiently evaluated
(despite having exponentially many monomial terms) and which can be
usefully interpreted combinatorially as enumerating polynomials for
certain objects. For example, if $G$ is a bipartite graph with vertex
partition $U \cup V$, the \emph{Edmonds matrix} of $G$ over some field $\F$ 
is a matrix $A \in \F^{U \times V}$ constructed by replacing
the non-zero entries of the bipartite adjacency matrix of $G$ by
distinct new variables -- i.e., for every edge $e \in E(G)$ we define
a variable $x_e$, and we let $A(u,v)=x_{uv}$ if $uv \in E(G)$ and
$A(u,v)=0$ otherwise. If $|U|=|V|$, then $P(X)=\det A$ is a polynomial
over the variables $X=\{x_e \mid e \in E(G)\}$, and can easily be seen
to be an enumerating polynomial for perfect matchings in $G$.
(Note that we pay no attention to the coefficients of the monomials,
which are here either $1$ or $-1$ depending on the matching; in
particular, we are not concerned with \emph{counting} the objects.)
For our second example, let $G$ be a digraph, and let $A$ be its
standard adjacency matrix, modified as above so that
non-zero entries $A[u,v]$ are replaced by $x_{(u,v)}$
for distinct new variables $x_{(u,v)}$, $(u,v) \in E(G)$.
Then $A^k[u,v]$ enumerates $k$-edge walks from $u$ to $v$.
Further examples include the Tutte matrix, which provides a
way to enumerate perfect matchings in non-bipartite graphs \cite{Tutte47};
\emph{branching walks} (due to Nederlof~\cite{Nederlof13algorithmica}), which are
a relaxation of subtrees of a graph similar to how walks are a
relaxation of paths; and any number of applications of basic linear
algebra, especially in the context of \emph{linear matroids} (see below). 

Second, there is a rich toolbox of \emph{transformations} of polynomials,
by which a given enumerating polynomial can be modified into a more
relevant form. We are particularly concerned with what can be called
\emph{sieving} operations: transformations applied to a given
polynomial $P(X)$ such that every monomial $m$ in the monomial expansion
of $P$ either survives (possibly multiplied by some new factor)
or is cancelled, depending on the properties of $m$. 
For example, consider a graph $G$ with edges partitioned as $E(G)=E_R \cup E_B$ 
into \emph{red} and \emph{blue} edges. Does $G$ have a perfect
matching where precisely half the edges are red (or more generally,
with precisely $w$ red edges)? 
This is known as the \textsc{Exact Matching} problem, and is
not known to have a deterministic polynomial-time algorithm. However,
there is a simple randomized polynomial-time algorithm using the
enumerating polynomial approach (cf.~Mulmuley et al.\ in 1987~\cite{MulmuleyVV87matching}).
Assume for simplicity that $G$ is
bipartite, and let $A$ be the Edmonds matrix of $G$ as above (if $G$
is not bipartite, a similar construction works over the Tutte matrix of $G$).
Introduce a new variable $z$, and for every edge $uv \in E_R$ multiply
$A[u,v]$ by~$z$. Now, a perfect matching $M$ of $G$ with $w$ red edges
will correspond to a monomial where the degree of $z$ is $w$, and we
are left asking for monomials in $P(X,z)=\det A$ where the $z$-component
is $z^w$. Via the standard method of \emph{interpolation}, we can define
a second polynomial $P_2(X)$ which enumerates precisely these monomials,
and $P_2(X)$ can be evaluated using $O(n)$ evaluations of $P(X)$ (i.e.,
$P_2$ \emph{sieves} for monomials in $P(X,z)$ where $z$ has degree $w$).
Thus, applying polynomial identity testing to $P_2$ gives a randomized
polynomial-time algorithm for \textsc{Exact Matching}. 

For applications to exact and parameterized algorithms, more powerful
transformations are available. The most well known is
\emph{multilinear detection}: Given a polynomial $P(X)$, does the
monomial expansion of $P$ contain a monomial of degree $k$ which is
\emph{multilinear}, i.e., where every variable has degree at most one?
Slightly more generally, to avoid undesired cancellations, we consider
the following. Let $P(X,Y)$ be a polynomial in two sets of variables
$X$ and $Y$. Say that a monomial $m$ is \emph{$k$-multilinear in $Y$}
if the total degree of $m$ in $Y$ is $k$ (not counting any
contributions from $X$) and every variable in $Y$ has degree at most
one in $m$. Then the following is known. 

\begin{lemma}[Multilinear detection~\cite{Bjorklund14detsum,BjorklundHKK17narrow}] \label{lemma:multlinear}
  Let $P(X,Y)$ be a polynomial 
  over a field of characteristic~2. There is a polynomial $Q(X,Y)$,
  that can be computed using $O^*(2^k)$ evaluations of $P$, such that $Q$
  is not identically zero if and only if $P$ contains a monomial that
  is $k$-multilinear in $Y$.
\end{lemma}

For example, consider again the case where $A$ is the modified
adjacency matrix of a graph $G$, and scale every entry $A[u,v]$ by
a new variable $y_v$. Then the terms of $A^k[u,v]$ that are
multilinear in $Y=\{y_v \mid v \in V\}$ enumerate $k$-edge
\emph{paths} from $u$ to $v$, i.e., multilinear detection and a PIT
algorithm solve the \textsc{$k$-Path} problem.
This idea was pioneered by Koutis~\cite{Koutis08ICALP} and
improved by Williams and Koutis~\cite{KoutisW09limits,Williams09IPL},
using a different approach based on group algebra; the above
polynomial sieving result is by Björklund et al.~\cite{Bjorklund14detsum,BjorklundHKK17narrow}.
Multilinear detection and other algebraic sieving has had many applications, including Bj\"orklund's
celebrated algorithm for finding undirected Hamiltonian cycles in time
$O^*(1.66^n)$~\cite{Bjorklund14detsum} and an algorithm solving
\textsc{$k$-Path} in time $O^*(1.66^k)$~\cite{BjorklundHKK17narrow}.
See Koutis and Williams~\cite{KoutisW16CACM} for an overview; other
related work is surveyed in Section~\ref{ssec:related}.

Some variations are also known. One arguably simpler variant 
is when $|Y|=k$ and we wish to sieve for monomials in $P(X,Y)$ where
every variable of $Y$ occurs (regardless of their degree).
This can be handled over any field in $2^k$ evaluations of $P$
using \emph{inclusion-exclusion} (and this is a ``clean'' sieve, which
does not change the coefficient of any monomial).
This has been used, e.g., in parameterized algorithms for
\textsc{List Colouring}~\cite{GutinMOW21listcol} and 
\textsc{Rural Postman}~\cite{GutinWY17rpostman}.
Another variant, which is a generalisation of multilinear detection, 
is \emph{constrained multilinear detection}.
Let $P(X,Y)$ be a polynomial. Let $C$ be a set of colours, and for
every $q \in C$ let $d_q \in \N$ be the \emph{capacity} of colour~$q$.
Let a colouring $c \colon Y \to C$ be given.
A monomial $m$ is \emph{properly coloured} if, for every $q \in C$,
the total degree of all variables in $m$ with colour $q$ is at most $d_q$. 
Bj\"orklund et al.~\cite{BjorklundKK16} show the following (again, we
allow an additional set of variables $X$ to avoid undesired
cancellations).

\begin{lemma}[Constrained multilinear detection~\cite{BjorklundKK16}]
  Let $P(X,Y)$ be a polynomial 
  over a field of characteristic 2. Let a colouring $c \colon Y \to C$
  and a list of colour capacities $(d_q)_{q \in C}$ be given.  
  There is a polynomial $Q(X,Y)$,
  that can be computed using $O^*(2^k)$ evaluations of $P$, such that $Q$
  is not identically zero if and only if $P$ contains a monomial that
  is $k$-multilinear in $Y$ and properly coloured. 
\end{lemma}

Using this method, Bj\"orklund et al.~\cite{BjorklundKK16} solve
\textsc{Graph Motif} and associated optimization variants
in time $O^*(2^k)$, which is optimal under the Set Cover Conjecture
(SeCoCo)~\cite{BjorklundKK16,CyganDLMNOPSW16seth}.

Although many other variations of algebraically styled FPT algorithms
are known~\cite{BjorklundHKK07convolution,BrandDH18,Brand19thesis,BrandP21,CyganNPPRW22twdconn,FominLPS16JACM},
the above methods (degree-extraction and multilinear detection)
are remarkable in terms of their power and simplicity in their applications.
In this paper, we show an extension of this toolbox. 

\subsection{Determinantal sieving}

We introduce \emph{determinantal sieving}, a powerful new sieving
operation that drastically extends the power of the tools of
multilinear detection and constrained multilinear detection.
Let $P(x_1,\ldots,x_n)$ be a polynomial over a field $\F$ of
characteristic 2, and let $M \in \F^{k \times n}$ be a matrix (e.g., a
linear matroid on the ground set $X=\{x_1,\ldots,x_n\}$). 
For a monomial $m$ in $P$, let $\supp(m)$ be the set of variables $x_i$
of non-zero degree in $m$.
We show a sieving method that, using $O^*(2^k)$ evaluations of $P$,
sieves for those monomials $m$ in $P$ that are multilinear of degree $k$ 
and such that the matrix $M[\cdot,\supp(m)]$ indexed by the support 
is non-singular. More precisely, we show the following.

\begin{theorem}[Basis sieving] \label{ithm:simple}
  Let $P(X)$
  be a polynomial of degree $d$ over a field $\F$ of characteristic~2, 
  and let $A \in \F^{k \times X}$ be a
  matrix. There is a randomized algorithm,
  using $O(d2^k)$ evaluations of $P$ and using $O^*(2^k)$ arithmetic operations, to test if there is a multilinear
  term $m$ of degree $k$ in the monomial expansion of $P(X)$ such that the matrix
  $A[\cdot,\supp(m)]$ is non-singular.
  The algorithm uses polynomial space, needs only evaluation access to $P$,
  has no false positives and produces false negatives with probability
  at most $2k/|\F|$. 
\end{theorem}

By applying Theorem~\ref{ithm:simple} with a Vandermonde matrix $A$ (see Section~\ref{sec:prel-matroids}), we can recover the multilinear detection result stated in Lemma~\ref{lemma:multlinear}.
The proof is remarkably simple, consisting of merely inspecting the
result of an application of inclusion-exclusion sieving; see Section~\ref{sec:sieving}. 
While similar results are known, the above theorem is new in its
generality and time and space efficiency, which implies that
determinantal sieving can be applied as a plug-in replacement 
for multilinear detection at no increased cost (see related work
and Section~\ref{ssec:mld-compare}).
In all our applications, the failure rate can be made
arbitrarily small with negligible overhead by moving to an extension field of $\F$.

We give a useful variant, where we sieve for a basis among the variables
of odd degree in each monomial $m$  -- the \emph{odd support} of $m$, denoted by $\osupp(m)$.
This has applications on its own; see the
\emph{Diverse X} and \emph{paths and linkages} examples below.
For further variants,
see Section~\ref{sec:sieving}.

\begin{theorem}[Odd sieving] \label{ithm:odd-sieve}
  Let $P(X)$
  be a polynomial of degree $d$ over a field $\F$ of characteristic~2, 
  and let $A \in \F^{k \times X}$ be a
  matrix. There is a randomized algorithm,
  using $O(d2^k)$ evaluations of $P$ and using $O^*(2^k)$ arithmetic operations, to test if there is a 
  term $m$ in the monomial expansion of $P(X)$ such that the 
  matrix $A[\cdot, \osupp(m)]$ has full row rank.
  The algorithm uses polynomial space, needs only evaluation access to $P$,
  has no false positives and produces false negatives with probability
  at most $(k+d)/|\F|$.
\end{theorem}

\subsubsection{Over general fields}

The aforementioned sieving algorithms only work over fields of characteristic 2.
By utilizing the \emph{exterior algebra}, we can effectively sieve over arbitrary fields.
We will follow the work of Brand et al.~\cite{BrandDH18}, who exhibited the power of the exterior algebra in parameterized algorithms.
Assume that a polynomial $P(X)$ over $\F$ is represented by an arithmetic circuit $C$.
Following the idea of Brand et al.~\cite{BrandDH18}, we attempt to evaluate $C$ over the exterior algebra $\Lambda(\F^k)$.
The exterior algebra is an algebra over $\F$ of dimension $2^k$, where
the addition is commutative but the multiplication (called \emph{wedge product}) is not (see Section~\ref{sec:extensor} for the definition).
Thus, naively evaluating over $C$ will not preserve the coefficients of $P(X)$.
We present two ways to circumvent this issue.

The first one concerns the restriction on the circuit.
We consider \emph{strongly monotone circuits}, which are basically circuit without any ``cancellation'' whatsoever.
An arithmetic circuit $C$ is \emph{skew} if at least one input of every product gate is an input gate.
We show that the result of evaluating a strongly monotone circuit $C$ over $\Lambda(\F^k)$ turns out non-zero only if $P(X)$ contains a monomial $m$ such that $A[\cdot, \supp(m)]$ is non-singular.

\begin{theorem} \label{ithm:smonotone}
  Let $P(X)$ be a polynomial of degree $d$ over a field $\F$, represented by a strongly monotone arithmetic circuit $C$,
  and let $A \in \F^{k \times X}$ be a
  matrix. There is a randomized algorithm
  in $O^*(2^{\omega k / 2})$ arithmetic operations (where $\omega < 2.373$ is
  the matrix multiplication exponent) or $O^*(2^k)$ arithmetic operations if $C$ is skew that tests if there is a multilinear
  term $m$ in the monomial expansion of $P(X)$ such that the matrix
  $A[\cdot,\supp(m)]$ is non-singular.
  The algorithm uses $O^*(2^k)$ space, 
  has no false positives and produces false negatives with probability
  at most $2k/|\F|$.
\end{theorem}

We also provide a way to sieve over arbitrary arithmetic circuits inspired by the \emph{lift mapping} of Brand et al.~\cite{BrandDH18}, which maps every  element in $\Lambda(\F^k)$ to $\Lambda(\F^{2k})$, an algebra of dimension $4^k$. 
Although the lift mapping costs extra time and space usage, it brings commutativity to the algebra, allowing us to evaluate the circuit over the exterior algebra.

\begin{theorem} \label{ithm:general}
  Let $P(X)$ be a polynomial of degree $d$ over a field $\F$, represented by a skew arithmetic circuit $C$,
  and let $A \in \F^{k \times X}$ be a
  matrix. There is a randomized algorithm
  in $O^*(2^{\omega k})$ arithmetic operations or $O^*(4^k)$ arithmetic operations if $C$ is skew that tests if there is a multilinear
  term $m$ in the monomial expansion of $P(X)$ such that the matrix
  $A[\cdot,\supp(m)]$ is non-singular.
  The algorithm uses $O^*(4^k)$ space, 
  has no false positives and produces false negatives with probability
  at most $2k/|\F|$.
\end{theorem}

\subsubsection{Linear matroids} \label{ssec:matroids}
For applications of determinantal sieving, we view the 
labelling matrix $M$ as representing a \emph{linear matroid} over the
variable set. A \emph{matroid} is a pair $M=(V,\cI)$ where $V$ is the
ground set and $\cI \subseteq 2^V$ a set of \emph{independent sets} in
$M$, subject to the following axioms:
(1) $\emptyset \in \cI$;
(2) If $B \in \cI$ and $A \subset B$ then $A \in \cI$; and 
(3)  For any $A, B \in \cI$ such that $|A|<|B|$ there exists an
  element $x \in B \setminus A$ such that $A+x \in \cI$.
  A \emph{linear matroid} is a matroid $M$ represented by a matrix $A$ with
column set $V$, such that a set $S \subseteq V$ is independent in $M$
if and only if $A[\cdot,S]$ is non-singular. 

A more complete overview of matroid theory concepts is given in
Section~\ref{sec:prel}, but let us review two particularly
relevant matroid constructions. 
A \emph{uniform matroid} $U_{n,k}$ is the matroid $M=(V,\cI)$
  where $\cI=\binom{V}{\leq k}$ (for $|V|=n$), i.e., a set is
  independent if and only if it has cardinality at most $k$.
  Letting $M$ be a uniform matroid in determinantal sieving
  corresponds to traditional multilinear detection.
  More generally, 
  a \emph{partition matroid} $M=(V,\cI)$ is defined by a partition
  $V=V_1 \cup \ldots \cup V_d$ of the ground set and a list of
  capacities $(c_i)_{i=1}^d$; note that we allow $c_i>1$~\cite{OxleyBook}. 
A set $S \subseteq V$ is independent if
  and only if $|S \cap V_i| \leq c_i$ for every $i \in [d]$.
  Constrained multilinear detection corresponds roughly to the case of
  $M$ being a partition matroid (or more precisely, the truncation of
  a partition matroid to rank $k$).
Both of these classes can be represented over fields of characteristic~2.

There also exists a range of transformations that can be applied
to matroids, with preserved representation; see
Section~\ref{sec:prel-matroids}. Here, we only note the 
operation of \emph{truncation}: Given a matroid $M=(V,\cI)$,
represented over a field $\F$ (either a finite field or the rationals), and an integer $k$,
we can in polynomial time \emph{truncate} $M$ to have rank $k$ while preserving the representation,
at the cost of moving to an extension field~\cite{LokshtanovMPS18,Marx09-matroid}.
Thereby, whenever we are looking for a solution of rank $k$,
we may assume that every matroid $M=(V,\cI)$ in our input
is represented by a full-rank matrix of dimension $k \times |V|$.

We find it particularly interesting that the fastest known
method for multilinear detection, which sieves over a random bijective
labelling~\cite{Bjorklund14detsum,BjorklundHKK17narrow}, 
can be seen as a direct application of Theorem~\ref{ithm:simple}
using a randomized representation of a uniform matroid;
see Section~\ref{ssec:mld-compare}.
In this sense, the results of this paper come without any extra computational cost -- 
they rely on the same sieving steps that existing algorithms already perform, only computed on a more carefully chosen set of evaluation points.

\subsubsection{Comparison to related work} \label{ssec:related}
Let us now compare the determinantal sieving method to other
approaches from the literature. While the text so far (excepting the
material on the exterior algebra) has been written to be 
digestible for a reader of general background,
this comparison is inevitably more technical.
We also note that the algebras underpinning the
exterior algebra, apolar algebra and (over fields of characteristic 2)
group algebra approaches are isomorphic~\cite{Brand22note,BrandDH18},
hence either of these methods is capable of recovering some variant 
of the determinantal sieving result from the right perspective, but we focus on what is
present in the literature. 

The earliest work to identify multilinear detection as a useful
primitive for FPT algorithms is Koutis~\cite{Koutis05IPL}. Koutis and
Williams~\cite{Koutis08ICALP,KoutisW16limits,Williams09IPL} refined
the approach, giving a randomized $O^*(2^k)$-time, polynomial-space
algorithm for multilinear detection using a group algebra.
The method implicitly solves determinantal sieving for monotone
arithmetic circuits (i.e., circuits over $\mathbb{Z}_+$) using
matroids over GF(2), but independence over larger base fields was not considered~\cite{KoutisW16limits}.
Koutis~\cite{Koutis12motif} also proposed the task of constrained
multilinear detection and provided an $O^*(2.54^k)$-time algorithm.
By comparison, the polynomial sieving methods solve multilinear
detection~\cite{Bjorklund14detsum,BjorklundHKK17narrow}
and constrained multilinear detection~\cite{BjorklundKK16}
in time $O^*(2^k)$ and polynomial space
for arbitrary polynomials over fields of characteristic 2,
but the more general task of sieving for matroid bases was (again) not considered.  

Fomin et al.~\cite[Section~5.1]{FominLPS17} use the representative
sets lemma to solve what is effectively determinantal sieving over
arbitrary fields in deterministic time $O^*(7.77^k)$ and exponential
space, again for monotone arithmetic circuits.
They also consider a weighted optimization version.

Finally, more recent research has employed algebraic approaches of
exterior algebra~\cite{BrandDH18} and apolar algebra~\cite{BrandP21}
for derandomizations and generalisations of the above. 
These methods intrinsically solve the determinantal sieving problem,
but pay exponential overhead in both running time and space usage due
to the complexity of the underlying algebraic operations
(as seen in the contrast between Theorem~\ref{ithm:simple}
and Theorem~\ref{ithm:general}). 
For instance, Brand et al.~\cite{BrandDH18} give a randomized $O^*(4.32^k)$-time algorithm for multilinear 
detection for arbitrary polynomials represented by an arithmetic circuit.

In summary, there are several approaches in the literature which can 
provide some form of a determinantal sieving procedure, but
the results are all restricted either in the structure
of the arithmetic circuit encoding the polynomial (such as only
applying to monotone circuits) or in requiring significant 
overhead in time and space usage. By contrast, Theorem~\ref{ithm:simple}
has the advantage of simultaneously (1) providing the (presumably best
possible) running time of $O^*(2^k)$ and polynomial space;
(2) applying to any polynomial over a field of characteristic 2,
regardless of encoding\footnote{While working only over fields of
  characteristic 2 is of course a restriction, in practice we have not
  found it to be an obstacle (excepting issues of derandomization
  or counting algorithms). In particular, we are not
  aware of any ``combinatorially important'' matroid that is
  representable but not representable over fields of characteristic 2.
};
and (3) being simple to work with and apply,
requiring no algebraic techniques deeper than basic linear algebra and
matroid theory. We argue that this qualitatively and significantly
increases the applicability of the result, as hopefully evidenced from
the applications we provide.

\subsection{Applications}

Given Theorems~\ref{ithm:simple}--\ref{ithm:general},
a large collection of applications can be achieved by combining a
suitable enumerating polynomial for a problem with a suitable matroid labelling. 
Before we undertake a survey, let us more carefully define our terms. 
Let $V$ be a ground set and $\cF \subseteq 2^V$ a set system over $V$.
An \emph{enumerating polynomial} for $\cF$ is a polynomial $P(X,Y)$
over a set of variables $X \cup Y$, where $X=\{x_v \mid v \in V\}$,
such that the following holds:
(i) $P(X,Y)$ is multilinear in $X$
and (ii) for every $S \subseteq V$, there is a monomial $m$ in $P(X,Y)$
  whose support in $X$ is exactly $S$ if and only if $S \in \cF$.
Similarly, to capture applications of Theorem~\ref{ithm:odd-sieve} (odd sieving),
define a \emph{parity-enumerating polynomial} for $\cF$ as a
polynomial $P(X,Y)$ where for every $S \subseteq V$, there exists a monomial $m$
whose odd support in $X$ is $S$ if and only if $S \in \cF$.
The definition can be generalized further -- for example, if we want to
refer to an ``enumerating polynomial for walks'' we could treat walks
as \emph{multisets} of vertices or edges, and adjust the definition
accordingly. However, the above suffices for almost all applications. 

We next survey results covered by our approach. Our results cover
multiple areas, and include
both significant speedups of previous results (see Table~\ref{table:speedups})
and generalisations where a previous running time for a problem can
be reproduced in a broader setting; e.g., generalized to so-called
\emph{frameworks}~\cite{FominGKSS24,Lovasz1977,LovaszGeomBook2019}.
Furthermore, in general, both
the proofs and the algorithms are short and simple,
given existing families of enumerating polynomials and linear matroids.

\subsubsection{Matroid Covering, Packing and Intersection Problems}
\label{sec:intro-appl-matroids}
We begin with a straightforward application to the
\textsc{Set Cover} and \textsc{Set Packing} problems. 
Let $V$ be a ground set and $\cE \subseteq 2^V$ a collection of sets.
Let $M=(V,\cI)$ be a matroid of rank $k$, and let $t$ be an integer.
In \textsc{Rank~$k$ Set Cover} we ask: is there a subcollection
$S \subseteq \cE$ with $|S| \leq t$ such that $\bigcup S = \bigcup_{E \in S} E$ spans $M$?
In \textsc{Rank~$k$ Set Packing} we ask if there is
a collection $S \subseteq \cE$
of pairwise disjoint sets
with $|S|=t$ such that
$\bigcup S$ is a basis of $M$. 
(The variant of \textsc{Rank~$k$ Set Packing} where $\bigcup S$
is only required to be independent in $M$, not a basis,
reduces to the above via truncation of~$M$.)

\begin{theorem} \label{ithm:setcover}
  \textsc{Rank~$k$ Set Cover} and \textsc{Rank~$k$ Set Packing} for
  matroids represented over a field of characteristic 2 can be solved
  in randomized time $O^*(2^k)$ and polynomial space,
  and in time $O^*(2^{\omega k / 2})$ and $O^*(2^k)$ space over general fields.
\end{theorem}

To achieve this result, we use a simple subset-enumerating polynomial.
Assume an input $(V,\cE,M,t,k)$ is given, and define a set of
variables $X_{v,E}$, $v \in V$, $E \in \cE$, as well as a set of
fingerprinting variables $Y=\{y_{i,E} \mid i \in [t], E \in \cE\}$
to prevent cancellations. Define
\[
  P(X,Y) = \prod_{i=1}^t \sum_{E \in \cE} y_{i,E} \prod_{v \in E} x_{v,E}
  = \sum_{f \colon [t] \to \cE} \left( \prod_{i \in [t]} y_{i,f(i)} \prod_{v\in E_i} x_{v,E_i} \right).
\]
We consider the polynomial in $X$ obtained from $P(X,Y)$ by substituting each variable $y_{i,E}$ with a random element from $\mathbb{F}$.
By the Schwartz-Zippel lemma, if there exists a family $S \subseteq \mathcal{E}$ of $t$ sets with $U = \bigcup S$, then with high probability the resulting polynomial contains a monomial of the form $\prod_{v \in U} x_{v,\iota(v)}$, where $\iota \colon U \to S$ is a mapping that assigns to every $v \in U$ a set $\iota(v) \in S$.
Hence to solve \textsc{Rank~$k$ Set Packing}
we associate each variable $x_{v,E}$ with the vector representing $v$
in $M$, and invoke Theorem~\ref{ithm:simple} or Theorem~\ref{ithm:smonotone}
depending on the representation of $M$. 
For \textsc{Rank~$k$ Set Cover} we simply evaluate $P$ at a point
$x_{v,E} \gets 1+x_{v,E}$ for every $x_{v,E} \in X$ for the same result.

Note that Theorem~\ref{ithm:setcover} is tight for matroids
represented over a field of characteristic 2 under \emph{Set Cover
  Conjecture} (SeCoCo).
SeCoCo asserts that there is
no algorithm that solves \textsc{Set Cover} in time $O^*(2^{(1-\varepsilon)n})$
for any $\varepsilon > 0$~\cite{CyganDLMNOPSW16seth}, and since
\textsc{Set Cover} corresponds to the simple case where each element
$v_i \in V$ is associated with the $n$-dimensional unit vector $e_i$,
tightness follows.

Theorem~\ref{ithm:setcover} improves on state of the art even for very
simple settings. In \textsc{Matroid $q$-Parity}, the input is a
matroid $M=(V,\cI)$, a partition of $V$ into sets of size $q$, 
and an integer $k$, and the question is whether there is a packing
of $k$ sets that is independent in $M$. This problem can be solved in
polynomial time if $q=2$ and $M$ is linear, but is hard even for
linear matroids if $q \geq 3$. The fastest known algorithm for
\textsc{Matroid $q$-Parity} by Brand and Pratt
(for matroids over $\mathbb{R}$)
runs in deterministic $O^*(4^{qk})$ time
with exponential space~\cite{BrandP21},
improving on a previous result of Fomin et al.~\cite{FominLPS16JACM}
with running time $O^*(2^{\omega q k})$
over general fields.
We get the following. 

\begin{corollary}
  \textsc{Matroid $q$-Parity} for a linear matroid over a field of
  characteristic 2 can be solved in randomized time $O^*(2^{qk})$ and
  polynomial space.  
\end{corollary}

For a related problem, we get a greater speedup. 
In \textsc{$q$-Matroid Intersection}, the input is $q$ matroids $M_1,
\ldots, M_q$ of rank $k$, and the question is if they have a common
basis. Again, this is tractable if $q=2$, but NP-hard if $q \geq 3$
even for linear matroids. Assume that the matroids are represented by
matrices $A_1$, \ldots, $A_q$ over a common field $\F$ and a common
ground set $V$, where w.l.o.g.\ every matrix $A_i$ has $k$ rows and
has rank $k$ over $\F$.
We can use the Cauchy-Binet formula to sieve for solutions more efficiently.
Let $X=\{x_v \mid v \in V\}$ be a set of variables and let $A_1'$ be
the result of scaling every column $v$ of $A_1$ by $x_v$.
By the Cauchy-Binet formula,
\[
  P(X) := \det (A_1' A_2^T) = \sum_{B \in \binom{V}{k}} \det
  A_1[\cdot,B] \det A_2[\cdot,B] \prod_{v \in B} x_v.
\]
Thus $P(X)$ enumerates monomials $\prod_{v \in B} x_v$ for
common bases $B$ of $A_1$ and $A_2$, and we only have to sieve for
terms that in addition are bases of the remaining $q-2$ matroids. 
For $q=3$, this is immediate; for $q>3$, we can replace matroids $M_3, \ldots, M_q$
by their direct sum, and each variable $x_v$ by a product $\prod_{i=3}^q x_{v,i}$
over variables $x_{v,i}$ representing the copies of $x_v$ in the matroid $M_i$.
We get the following. 

\begin{theorem} \label{ithm:matroid-intersection}
  \textsc{$q$-Matroid Intersection} for linear matroids represented
  over a common field $\F$ of characteristic 2 can be solved in
  randomized time $O^*(2^{(q-2)k})$ and polynomial space.
\end{theorem}

The previous best result (again, over general fields) is Brand and Pratt~\cite{BrandP21}, with
running time $O^*(4^{qk})$. In particular, for $q=3$ this improves on
the state of the art from $O^*(4^{3k})$ to $O^*(2^k)$.
Theorem~\ref{ithm:matroid-intersection} matches the fastest algorithm
by Bj\"orklund et al.~\cite{BjorklundHKK17narrow}
for the much simpler \textsc{$q$-Dimensional Matching} problem.

As a particular special case, Theorem~\ref{ithm:matroid-intersection} with $q=3$ 
implies a polynomial-space, $O^*(2^n)$-time algorithm for \textsc{Directed Hamiltonian Path},
which despite intense efforts at improvement remains the state of the
art for the general case~\cite{Bellman62,BjorklundKK17directed,CyganKN18hamilton}.

\subsubsection{Fair and Diverse Solutions}

Fairness and diversity are important concepts in many areas of research, including artificial intelligence and optimization,
and have also seen increased focus in theoretical computer science.
We discuss two related problems: finding a balanced-fair solution and a diverse collection of solutions.

The problem of finding a balanced-fair solution arises in many contexts \cite{BandyapadhyayFIS23,BentertKN23,Chierichetti0LV17,Chierichetti0LV19}, including \textsc{Matroid Intersection}, \textsc{$k$-matching}, and \textsc{$k$-path}.
We define a general problem \textsc{Balanced Solution}:
Given a set $E$ with coloured elements, a collection $\mathcal{F}$ of subsets of $E$, the goal is to find a set $S \in \mathcal{F}$ of size $k$ such that the number of elements of $S$ with each colour is within certain bounds.
We show that this problem can be solved in $O^*(2^k)$ time using basis sieving:

\begin{theorem}
  \textsc{Balanced Solution} can be solved in $O^*(2^k)$ time if there is an enumerating polynomial for $\cF$ that can be evaluated in polynomial time over a field of characteristic 2.
\end{theorem}

The problem of finding a diverse collection of solutions is another important optimization problem.
Here, the goal is to find not just a single optimal solution, but a collection of solutions that are diverse in some sense.
We measure diversity in terms of Hamming distance, i.e., diverse solutions should have a large Hamming distance between them.
This problem has received significant attention in the parameterized complexity literature \cite{BasteFJMOPR22,BasteJMPR19,FominGJPS24,FominGPPS24,HanakaKKO21}.
We discuss a general method based on the odd sieving technique that can be used to find a diverse collection of solutions for a wide range of optimization problems.
We define the \textsc{Diverse Collection} problem as follows.
The input is a set $E$, collections of subsets $\cF_1, \dots, \cF_k$, and $d_{i,j} \in \N$ for $1 \leq i < j \leq k$, and the goal is to find subsets $S_i \in \cF_i$ for each $i \in [k]$ such that $|S_i \Delta S_j| = |(S_i \setminus S_j) \cup (S_j \setminus S_i)| \ge d_{i,j}$ for every $i, j$.
Let $D =  \sum_{i < j \in [k]} d_{i,j}$.
We use the odd sieving algorithm to obtain an $O^*(2^D)$-time algorithm.
The key here is to use a distinct set of variables for every pair $i,j$. 
Thereby, those elements in the intersection of two solutions, having contribution two, can be excluded in the odd sieving.

\begin{theorem}
  \textsc{Diverse Collection} can be solved in $O^*(2^{D})$ time if all collections $\cF_i$ admit enumerating polynomials that can be evaluated in polynomial time over fields of characteristic 2.
\end{theorem}

This leads to significant speed-ups compared to existing algorithms, one for \textsc{Diverse Matchings} and another for \textsc{$d$-Distinct Branchings}.
The \textsc{Diverse Matchings} problem ask whether a given graph contains $k$ perfect matchings $M_1, \dots, M_k$ whose pairwise Hamming distances are all at least $d$.
Fomin et al.~\cite{FominGJPS24} give an algorithm with running time $2^{2^{O(kd)}}$.
We obtain a faster algorithm running in time $O^*(2^{d \binom{k}{2}})$.
In \textsc{$d$-Distinct Branchings}, we are given a directed graph $G$, two vertices $s, t$, and an integer $d$, and we search for an in-branching rooted at $s$ and out-branching rooted $t$ whose Hamming distance is at least $d$.
This problem can be solved in $O^*(2^d)$ time.
In particular, this answers the question of Bang-Jensen et al.~\cite{Bang-JensenK021} whether there exists an $O^*(2^{O(d)})$-time algorithm.
Previously known algorithms run in time $O^*(2^{d^2 \log^2 d})$~\cite{GutinRW18} and $O^*(d^{O(d)})$ \cite{Bang-JensenK021}.

\subsubsection{Undirected paths and linkages}
\label{ssec:linkage}
As noted above, among the earliest and most powerful
applications of algebraic FPT algorithms are path and cycle
problems. In fact, all the current fastest FPT algorithms for
\textsc{$k$-Path} -- randomized time $O^*(1.66^k)$ for undirected
graphs~\cite{BjorklundHKK17narrow} and $O^*(2^k)$ for
digraphs~\cite{Williams09IPL}; deterministic $O^*(2.55^k)$ time for
both variants \cite{Tsur19b} -- ultimately have algebraic underpinnings.

Another highly surprising result was for the \textsc{$T$-Cycle} problem
(we use the name from Fomin et al.~\cite{FominGKSS24} to
distinguish more clearly from \textsc{$k$-Cycle}).
Here, the input is an undirected graph $G$ and a set of terminals $T
\subseteq V(G)$, and the question is whether $G$ contains a simple
cycle $C$ that passes through all vertices in $T$. 
This problem was known to be FPT parameterized by $k=|T|$,
using graph structural methods, but the running time was 
impractical~\cite{Kawarabayashi08}. Bj\"orklund, Husfeldt and
Taslaman~\cite{BjorklundHT12soda} showed an $O^*(2^k)$-time algorithm,
based on cancellations in the evaluation of a large polynomial. 
Wahlstr\"om~\cite{Wahlstrom13STACS} showed that the problem even has a
\emph{polynomial compression} in $k$, based on a reinterpretation of
the previous algorithm in terms of the determinant of a modified Tutte
matrix (similar to Bj\"orklund's celebrated $O^*(1.66^n)$-time
algorithm for \textsc{Hamiltonicity}~\cite{Bjorklund14detsum}). 
It is this latter determinant approach that we build upon in the
algorithms for path and linkage problems in this paper. 

Let $G$ be an undirected graph and $S, T \subseteq V(G)$ be disjoint
vertex sets. An \emph{$(S,T)$-linkage} in $G$ is a collection of
$|S|=|T|$ pairwise vertex-disjoint paths from $S$ to $T$ -- i.e., a
vertex-disjoint $(S,T)$-flow assuming that vertices of $S \cup T$ 
have capacity 1. Let $\cP$ be an $(S,T)$-linkage for some $(G, S, T)$.
A \emph{padding} of $\cP$ is a collection of oriented cycles that
covers $G-V(\cP)$, where every cycle has length at most 2 (i.e., every
cycle is either a 2-cycle $uvu$ over some edge $uv \in E(G)$
or a loop $v$ on some vertex $v \in V(G)$). 
A \emph{padded $(S,T)$-linkage} is an $(S,T)$-linkage $\cP$ together with a padding of $\cP$.
We show that there is an enumerating polynomial for padded $(S,T)$-linkages. 

\begin{lemma}
  Let $G$ be an undirected graph, possibly with loops,
  and let $S, T \subseteq V(G)$ be disjoint.
  In polynomial time, we can construct a matrix $A$ with entries from the variable set $X= \{x_e \mid e \in E(G) \}$ such that the polynomial $P(X)=\det A$, when evaluated over a field of characteristic 2, enumerates padded $(S,T)$-linkages of $G$.
  In other words, $P(X)$ is a parity-enumerating polynomial for $(S,T)$-linkages: it contains a monomial whose odd support is exactly $E(L)$ if and only if an $(S,T)$-linkage $L$ exists.
\end{lemma}

This result is interesting even when $|S|=|T|=1$, in which case $P(X)$
enumerates padded $(s,t)$-paths. We find this remarkable, as normally,
a polynomial that is efficiently computable would only be expected to
enumerate walks, as opposed to paths or cycles. It is not \emph{too}
powerful, since the padding terms from 2-cycles prevent us from 
using it to solve, e.g., \textsc{Hamiltonian Path} in polynomial time.
But it is highly useful for FPT purposes, since Theorem~\ref{ithm:odd-sieve}
allows us to sieve for terms that span a linear matroid $M$ while ignoring the
padding-part of each padded linkage. Thus we get the following.

\begin{theorem} \label{ithm:linkage}
  Let $G=(V,E)$ be an undirected graph and
  let $M=(V,\cI)$ be a matroid represented over a field of characteristic 2.
  Let $S, T \subseteq V(G)$ be disjoint vertex sets and $k \in \N$.
  In randomized time $O^*(2^k)$ and polynomial space we can find a
  minimum-length $(S,T)$-linkage in $G$ that has rank at least $k$ in
  $M$ (or determine that none exist).  
\end{theorem}

This result improves and generalizes a number of results.
Fomin et al.~\cite{FominGKSS24} gave randomized algorithms in time $O^*(2^{k+p})$
for finding a minimum-length colourful $(S,T)$-linkage,
and in time $O^*(2^{p+O(k^2 \log (q+k))})$
for finding a minimum-length $(S,T)$-linkage of rank at least $k$ in $M$,
where $M$ is represented by a matrix over a finite field of order $q$
and $p=|S|=|T|$.\footnote{The formulation of Fomin et al.~\cite{FominGKSS24}
  is slightly different, but equivalent under simple transformations.}
Theorem~\ref{ithm:linkage} directly generalizes the first result,
removing the dependency on $p$, and improves the running time for the second in the case
that $M$ can be represented over a field of characteristic~2.
It also significantly simplifies the correctness proof, which in~\cite{FominGKSS24}
runs to over 20 pages.

As they observe, even the problem \textsc{Colourful $(s,t)$-Path}
captures a number of problems, including \textsc{$T$-Cycle},
\textsc{Long $(s,t)$-Path} and \textsc{Long Cycle} (i.e.,
finding an $(s,t)$-path, respectively cycle, of length at least $k$).
Finding an $(s,t)$-path of rank at least $k$ also generalizes the
variant \textsc{List $T$-Cycle}, previously shown to be FPT by
Panolan, Saurabh and Zehavi~\cite{DBLP:journals/algorithmica/PanolanSZ19}.

We also show an improvement to \textsc{Long $(s,t)$-Path} and
\textsc{Long Cycle}. Fomin et al.~\cite{FominGKSS24} ask as an
open problem whether these can be solved in time $O^*((2-\varepsilon)^n)$
for some $\varepsilon > 0$, given that \textsc{$k$-Path} and \textsc{$k$-Cycle} 
have $O^*(1.66^k)$-time algorithms due to Björklund et al.~\cite{BjorklundHKK17narrow}. 
We confirm this.

\begin{theorem} \label{ithm:longpath}
  \textsc{Long $(s,t)$-Path} and \textsc{Long Cycle} can be solved in
  randomized time $O^*(1.66^k)$ and polynomial space. 
\end{theorem}

\subsubsection{Subgraph problems}
\label{ssec:subgraph}
Another class of problems where algebraic methods have been important
is for the general question of finding subgraphs of a graph $G$ with a given property. 
We give two results in this domain.

First, let $G=(V,E)$ be a graph and $M$ a matroid over $V$.
Let \textsc{Rank $k$ Connected Subgraph} be the following general problem:
Given integers $k$ and $t$, is there a connected subgraph $H$ of $G$
on at most $t$ vertices such that $V(H)$ has rank at least $k$ in $M?$

\begin{theorem} \label{ithm:subtree}
  \textsc{Rank $k$ Connected Subgraph} for a linear matroid $M$
  can be solved in randomized time $O^*(2^k)$ and polynomial space
  if $M$ is represented over a field of characteristic 2,
  and in randomized time $O^*(2^{\omega k})$ and space $O^*(4^k)$
  otherwise.
\end{theorem}

This result is an application of the powerful notion of \emph{branching walks},
introduced by Nederlof~\cite{Nederlof13algorithmica}, which underlie several FPT algorithms.
We rely on Björklund et al.~\cite{BjorklundKK16} who gave an explicit algorithm for
evaluating the branching walk polynomial.
As special cases of Theorem~\ref{ithm:subtree} with various matroids $M$
we recover the $O^*(2^k)$-time algorithms for
\textsc{Steiner Tree}~\cite{Nederlof13algorithmica} and
\textsc{Group Steiner Tree}~\cite{MisraPRSS12} on $k$ terminals,
and for \textsc{Graph Motif} and \textsc{Closest Graph Motif}~\cite{BjorklundKK16}.

More generally, consider \textsc{Subgraph Isomorphism}, the problem of finding a subgraph of $G$
isomorphic to a given graph $H$. This is W[1]-hard in general
(cf.~$k$-clique), but when restricted to a class of graphs~$\mathcal{G}$, it is FPT parameterized by $|V(H)|$, when every graph in $\mathcal{G}$ has bounded treewidth. 
In fact, up to plausible conjectures, the dependency on the treewidth $w$ for known algorithms
is optimal for every $w \geq 3$~\cite{BringmannS21}.
Like previous algorithms, we employ the \emph{homomorphism polynomial}
(see, e.g., Brand~\cite{Brand19thesis}), and show the following.

\begin{theorem} \label{ithm:subgraph-iso}
  Let $G$ and $H$ be undirected graphs, $k=|V(H)|$ and $n=|V(G)|$,
  and let $M$ be a linear matroid over $V(G)$. Let a tree decomposition
  of $H$ of width $w$ be given.
  We can find a subgraph $H'$ of $G$ isomorphic to $H$
  such that $V(H')$ is independent in $M$ 
  in randomized time $k^{O(1)} \cdot 2^k \cdot n^{w+1}$ and polynomial space
  if $M$ is represented over a field of characteristic 2,
  and in time $k^{O(1)} \cdot 2^{\omega k} \cdot n^{w+1}$ and space $O^*(4^k)$
  otherwise. 
\end{theorem}

\subsubsection{Speeding up dynamic programming}
\label{sec:intro-noDP}

The \emph{representative sets lemma}~\cite{Lovasz1977,Marx09-matroid}
is a statement from matroid theory that has seen a multitude of
applications in parameterized complexity, both in kernelization~\cite{KratschW20JACM}
and in FPT algorithms~\cite{FominLPS16JACM,Marx09-matroid}. The latter
class of application typically consists of a \emph{sped-up dynamic programming}
algorithm; e.g., a dynamic programming algorithm over a state space
that could potentially contain $n^{O(k)}$ different partial solutions,
but where the representative sets lemma is used to prove that it
suffices to maintain a set of $2^{O(k)}$ \emph{representative}
solutions. This includes
algorithms for paths and cycles~\cite{FominLPS16JACM}
as well as many more complex questions.
We refer to this as a \emph{rep-set DP}.

For many of these applications, faster randomized algorithms are known, even in
polynomial space, and the main contribution
of the representative sets lemma becomes to enable an almost competitive 
deterministic FPT algorithm~\cite{FominLPS16JACM,Tsur19b}.
However, for other applications this is not so clear, and there are
many applications of the representative sets lemma where no faster
method is known.
With the more powerful algebraic sieving methods of this paper, 
we can revisit some of these applications and show a speed-up of the
algorithm, while at the same time reducing the space usage to
polynomial space. We give three examples.

In \textsc{Minimum Equivalent Graph} (MEG), the input is a
digraph $G$, and the task is to find a subgraph $G'$ of $G$ with a
minimum number of edges such that $G$ and $G'$ have the same
reachability relation. Fomin et al.~\cite{FominLPS16JACM} give the
first single-exponential algorithm for MEG.
They show that MEG ultimately reduces down to finding an in-branching
$B_1$ and an out-branching $B_2$ with a common root sharing at least
$\ell$ edges for $\ell$ as large as possible, which they solve via rep-set DP in time $O^*(2^{4 \omega n})$. 
We reduce this question to an application of
\textsc{4-Matroid Intersection} and get the following.

\begin{theorem} \label{ithm:MEG}
  \textsc{Minimum Equivalent Graph} can be solved in polynomial space
  and randomized time $O^*(4^{n})$. 
\end{theorem}

In \textsc{(Undirected/Directed) Eulerian Edge Deletion}, 
the input is a graph $G$ (undirected respectively directed),
and the question is whether we can remove at most $k$ edges from $G$
so that the resulting graph is Eulerian (i.e., has a closed walk that
visits every edge precisely once). Cai and Yang~\cite{CaiY11} surveyed
related problems, but left the above questions open. Cygan et al.~\cite{CyganMPPS14euler}
gave the first FPT algorithms, with running times of $O^*(2^{O(k \log k)})$,
and Goyal et al.~\cite{GoyalMPPS18even} improved this to $O^*(2^{(2+\omega)k})$
using a rep-set DP approach over the co-graphic matroid.
We combine the co-graphic matroid approach with suitable enumerating
polynomials to get the following.

\begin{theorem} \label{ithm:euler}
  \textsc{Undirected Eulerian Edge Deletion} and \textsc{Directed
    Eulerian Edge Deletion} can be solved in $O^*(2^k)$ randomized time and
  polynomial space.
\end{theorem}

Finally, we consider a more unusual application. Consider a generic
problem where we are searching for a subset $S$ with property $\Pi$ of a ground set $V$.
In the \emph{conflict-free} version, the input additionally contains a graph $H=(V,E)$
and $S$ is required to be an independent set in $H$. Naturally,
this is hard in general (even disregarding the property $\Pi$),
but multiple authors have considered restricted variants. 
In particular, if $H$ is chordal, then 
Agrawal et al.~\cite{AgrawalJKS20} show that \textsc{Conflict-Free Matching}, where we search for a matching in a graph $G$ and the conflict graph is defined on the edge set of $G$,
can be solved in $O^*(2^{(2\omega + 2)k})$ time, and 
Jacob et al.~\cite{JacobMR21} show that \textsc{Conflict-Free Set Cover}
can be solved in $O^*(3^n)$ time.
We note that the \emph{independent set polynomial} (in our
terminology, an enumerating polynomial for independent sets in a graph)
can be efficiently evaluated if $H$ is chordal~\cite{AchlioptasZ21},
allowing us to speed up both results. See Section~\ref{sec:conflict} for details.

\paragraph{Subsequent work.} After the conference version of this article appeared, Akmal and Koana~\cite{AkmalK25} introduced partition sieving, based on the odd sieving technique. As an application, they obtained improved polynomial-space exact algorithms for \textsc{Edge Coloring} and \textsc{List Edge Coloring}.

\newcommand*{\expspace}{\textsuperscript{\textdagger}}
\newcommand*{\charatwo}{\textsuperscript{\S}}
\begin{table}
  \centering
  \caption{A list of speed-ups over previous results. Results marked with {}\expspace{} use exponential space, and those with {}\charatwo{} only work over a field of characteristic 2.
  For the linkage problems, $p$ is the order of the linkage.}
    \label{table:speedups}
  \begin{tabularx}{.95\textwidth}[ht]{Xll}
    \toprule
      Problem & Existing & New \\ 
    \midrule
      \textsc{$q$-Matroid Intersection} & $O^*(4^{qk})$\expspace \cite{BrandP21} & $O^*(2^{(q - 2)k})$\charatwo \\
      & & $O^*(2^{(q-2+(q \bmod 2))k})$\expspace \\
      \textsc{$q$-Matroid Parity} & $O^*(4^{qk})$\expspace \cite{BrandP21} & $O^*(2^{qk})$\expspace, $O^*(2^{qk})$\charatwo \\ [.5ex]
      \textsc{Long $(s, t)$-Path} & $O^*(2^{k})$ \cite{FominGKSS24} & $O^*(1.66^{k})$ \\
      \textsc{Colourful $(S,T)$-Linkage} & $O^*(2^{k+p})$ \cite{FominGKSS24} & $O^*(2^k)$ \\
      \textsc{Rank $k$ $(S,T)$-Linkage} & $O^*(2^{p+O(k^2 \log (k+|\F|))})$ \cite{FominGKSS24} & $O^*(2^k)$\charatwo \\ [.5ex] 
      \textsc{Diverse Perfect Matchings} & $O^*(2^{2^{O(D)}})$\expspace \cite{FominGJPS24} & $O^*(2^D)$ \\
      \textsc{$k$-Distinct Branchings} & $O^*(k^{O(k)})$\expspace \cite{Bang-JensenK021} & $O^*(2^k)$ \\ [.5ex]
      \textsc{Minimum Equivalent Graph} & $O^*(2^{4 \omega n})$\expspace \cite{FominLPS16JACM} & $O^*(2^{2n})$ \\
      \textsc{(Un)directed Eulerian Deletion} & $O^*(2^{(2 + \omega) k})$\expspace \cite{GoyalMPPS18even} & $O^*(2^{k})$ \\
      \textsc{Chordal-Conflict-free Matching} & $O^*(2^{(2 + 2\omega) k})$\expspace \cite{AgrawalJKS20} & $O^*(2^{2k})$ \\
      \bottomrule
  \end{tabularx}
\end{table}

\paragraph{Structure of the paper.} In Section~\ref{sec:prel} we cover
preliminaries, and in Section~\ref{sec:sieving} we prove the
determinantal sieving statements of Theorems~\ref{ithm:simple}--\ref{ithm:general}.
In Sections~\ref{sec:matroid-applications}--\ref{sec:noDP} we cover
the applications mentioned in Sections~\ref{sec:intro-appl-matroids}--\ref{sec:intro-noDP}, respectively.
We conclude in Section~\ref{sec:conc} with discussion and open problems.

\section{Preliminaries}
\label{sec:prel}

We use standard terminology from parameterized complexity, see, e.g.,
the book of Cygan et al.~\cite{CyganFKLMPPS15PCbook}.
For background on graph theory, see Diestel~\cite{DiestelBook} and Bang-Jensen and Gutin~\cite{book-digraphs}.

Let $P(X)$ be a polynomial over a set of variables $X = \{ x_1, \cdots, x_n \}$.
A monomial is a product $m = x_1^{m_1} \cdots x_n^{m_n}$ for non-negative integers $m_1, \cdots, m_n$.
A monomial $m$ is called multilinear if $m_i \le 1$ for each $i \in [n]$.
We say that its \emph{support} is $\{ i \in [n] \mid m_i
> 0 \}$ and that its \emph{odd support} is $\{ i \in [n] \mid m_i \equiv 1 \mod 2 \}$ denoted by $\supp(m)$ and $\osupp(m)$, respectively. 
We sometimes use the notation $X^m$ for the monomial $m=x_1^{m_1} x_2^{m_2} \cdots x_n^{m_n}$,
to clarify that the monomial $m$ does not include a coefficient.
For a set of variables $X'=\{x_1',\ldots,x_n'\}$ we will also
write $(X')^m=\prod_{i=1}^n (x_i')^{m_i}$.
For a monomial $m$ in the monomial expansion of $P(X)$, we let $P(m)$ 
denote the coefficient of $m$ in $P$, i.e., $P(X)=\sum_m P(m) X^m$
where $m$ ranges over all monomials in $P(X)$.
The total degree of $P(X)$ is $\max_{m} \sum_{i \in [n]} m_i$.
A polynomial of total degree $d$ is called \emph{homogeneous} if every monomial has degree $d$.
The Schwartz-Zippel lemma \cite{Schwartz80,Zippel79} states that a polynomial $P(X)$ of total degree at most $d$ over a field $\F$ becomes non-zero with probability at least $1 - d / |\F|$ when evaluated at uniformly chosen elements from $\F$, unless $P(X)$ is identically zero.

Lemmas~\ref{lemma:interpolation0} and~\ref{lemma:inclusion-exclusion} are the foundation of our sieving algorithms.

\begin{lemma}[Interpolation]
  \label{lemma:interpolation0}
  Let $P(z)$ be a univariate polynomial of degree $n - 1$ over a field $\F$.
  Suppose that $P(z_i) = p_i$ for distinct $z_1, \dots, z_n \in \F$.
  By the Lagrange interpolation, 
  \[
    P(z)
    = \sum_{i \in [n]} p_i \prod_{j \in [n] \setminus \{ i \}} \frac{z - z_j}{z_i - z_j}.
  \]
  Thus, given $n$ evaluations $p_1, \dots, p_n$ of $P(z)$, the coefficient of $z^t$ in $P(z)$ for every $t \in [n]$ can be computed in polynomial time.
\end{lemma}

\begin{lemma} \label{lemma:interpolation}
  Let $X = \{ x_1, \dots, x_n \}$ be a set of variables and let $P(X)$ be a polynomial of degree $d$ over a field $\F$.
  Let $P_k(X)$ be the homogeneous part of $P(X)$ of degree $k$, i.e., for every monomial $m$, its coefficient in $P_k$ is $P_k(m) = P(m)$ if $m$ has degree $k$, and $P_k(m) = 0$ otherwise. 
  Given evaluation access to $P$, we can simulate evaluation access of $P_k$ using $d + 1$ evaluations of $P$ and $\tilde{O}(d)$ arithmetic operations.
\end{lemma}
\begin{proof}
  Given $a_1, \dots, a_n \in \F$, we compute $P_k(a_1, \dots, a_n)$ as follows.
  Let $f_X(z):=P(a_1 z, \ldots, a_n z)$ for a new variable $z$.
  Note that $P_k(a_1, \dots, a_n)$ equals the coefficient of $z^k$ in $f(z)$,
  which can be computed in $\tilde{O}(d)$ arithmetic operations using
  fast interpolation (see e.g.,~\cite{vzGathen-MCA}). Alternatively,
  for a simpler version that is sufficient for our purposes, 
  choose distinct $c_1, \dots, c_{d+1} \in \F$, and for each $i \in [d+1]$, compute $f(z) = P(c_i a_1, \dots, c_i a_n)$ using evaluation access to $P$.
  Using these values, we can interpolate $f(z)$ in $O(d^2)$ arithmetic operations by Lemma~\ref{lemma:interpolation0}.
\end{proof}

\begin{lemma}[Inclusion-exclusion \cite{Wahlstrom13STACS}]
  \label{lemma:inclusion-exclusion}
  Let $P(Y)$ be a polynomial over a set of variables $Y = \{ y_1, \dots, y_n \}$ and a field of characteristic 2.
  For $T \subseteq [n]$, let $Q$ be a polynomial identical to $P$ except that the coefficients of monomials not divisible by $\prod_{i \in T} y_i$ are zero.
  Then, $Q = \sum_{I \subseteq T} P_{-I}$, where $P_{-I}(y_1, \dots, y_n) = P(y_1', \dots, y_n')$ for $y_i' = y_i$ if $i \notin I$ and $y_i' = 0$ otherwise. 
\end{lemma}

Let $A$ be a matrix over a field $\F$.
For a set of rows $I$ and columns $J$, we denote by $A[I, J]$ the submatrix containing rows $I$ and columns $J$.
If $I$ contains all rows ($J$ contains all columns), then we use the shorthand $A[\cdot, J]$ ($A[I, \cdot]$, respectively).

For a $k \times n$-matrix $A_1$ and an $n \times k$-matrix $A_2$, the Cauchy-Binet formula states that
\begin{align*}
  \det (A_1 A_2) = \sum_{S \in \binom{[n]}{k}} \det (A_1[\cdot, S]) \det (A_2[S, \cdot]).
\end{align*}

A square matrix $A$ whose diagonal entries are all zero is called \emph{skew-symmetric} if $A = -A^T$.
Suppose that the rows and columns of $A$ are indexed by $V$.
The \emph{Pfaffian} of $A$ is defined by
\begin{align*}
  \Pf A = \sum_{M} \sigma_M \prod_{uv \in M} A[u, v],
\end{align*}
where $M$ is ranges over all perfect matchings of the complete graph $(V, \binom{V}{2})$, and $\sigma_M = \pm 1$ is the sign of $M$, whose definition is not relevant in this work (see e.g., \cite{murota1999matrices}).
It is well-known that $\det A = (\Pf A)^2$.

\subsection{Linear matroids}
\label{sec:prel-matroids}

We review the essentials of matroid theory, with a focus on linear
matroids. For more background, see Oxley~\cite{OxleyBook}
and Marx~\cite{Marx09-matroid}.
A \emph{matroid} is a pair $M=(V,\cI)$ where $V$ is the ground set 
and $\cI \subseteq 2^V$ a set of \emph{independent sets} in
$M$, subject to the following axioms:
\begin{enumerate}
\item $\emptyset \in \cI$
\item If $B \in \cI$ and $A \subset B$ then $A \in \cI$
\item For any $A, B \in \cI$ such that $|A|<|B|$ there exists an
  element $x \in B \setminus A$ such that $(A+x) \in \cI$.
\end{enumerate}
A \emph{basis} of a matroid $M$ is a maximal independent set.
The \emph{rank} $r(M)$ of $M$ is the cardinality of a basis of $M$. 
A \emph{linear matroid} is a matroid $M$ represented by a matrix $A$ with
column set $V$, such that a set $S \subseteq V$ is independent in $M$
if and only if the set of columns of $A$ indexed by $S$ is linearly independent.
We review some useful matroid constructions, expanded from the introduction.
All the matroids below can be represented over fields of characteristic 2,
although in some cases the only known methods for efficiently constructing a
representation are randomized.
\begin{itemize}
\item A \emph{uniform matroid} $U_{n,k}$ is the matroid $M=(V,\cI)$
  where $\cI=\binom{V}{\leq k}$ (for $|V|=n$), i.e., a set is
  independent if and only if it has cardinality at most $k$.
  It is well-known that a Vandermonde matrix $A \in \mathbb{F}^{k \times n}$ defined by $A[i,j] = a_j^{i-1}$, where $a_j$'s are all distinct, gives a representation of a uniform matroid.

\item A \emph{partition matroid} $M=(V,\cI)$ is defined by a partition
  $V=V_1 \cup \ldots \cup V_d$ of the ground set and a list of
  capacities $(c_i)_{i=1}^d$. A set $S \subseteq V$ is independent if
  and only if $|S \cap V_i| \leq c_i$ for every $i \in [d]$.

\item Given an undirected graph $G=(V,E)$, the \emph{graphic matroid}
  of $G$ is a matroid $M=(E,\cI)$ where a set $F \subseteq E$ is
  independent if and only if it is acyclic. The \emph{cographic matroid}
  of $G$ is a matroid $M=(E,\cI)$ where a set $F \subseteq E$ is
  independent if and only if its deletion preserves connectivity (i.e., $G$ and
  $G-F$ have the same connected components). Graphic and co-graphic matroids
  are representable over every field~\cite{OxleyBook}.
  
\item Let $G=(U \cup V,E)$ be a bipartite graph. The \emph{transversal
    matroid} of $G$ is the matroid $M=(V,\cI)$ where a set
  $S \subseteq V$ is independent if and only if it is matchable in
  $G$.
\item Let $G=(V,E)$ be a digraph and $T \subseteq V$ a set of
  terminals. A set $S \subseteq V$ is \emph{linked} to $T$ if there is
  a collection of $|S|$ pairwise vertex-disjoint paths from $S$ to $T$.
  The set of all sets $S \subseteq V$ that are linked to $T$ form a
  matroid called a \emph{gammoid}.
\end{itemize}

If $M=(V,\cI)$ is a matroid, the \emph{dual matroid} of $M$ is
the matroid $M^*=(V,\cI')$ where a set $S \subseteq V$ is independent
in $M^*$ if and only if $V \setminus S$ contains a basis of $M$. 
Given a representation for $M$, a representation for $M^*$ can be
constructed in deterministic polynomial time over the same field (see \cite{OxleyBook}).  
Given two matroids $M_1(V_1,\cI_1)$ and $M_2(V_2,\cI_2)$ on disjoint
sets $V_1$ and $V_2$, the \emph{disjoint union} of $M_1$ and $M_2$ is
the matroid $M=M_1 \lor M_2=(V,\cI)$ where $V=V_1 \cup V_2$
and a set $S \subseteq V$ is independent if and only if
$S \cap V_1 \in \cI_1$ and $S \cap V_2 \in \cI_2$. 
More generally, for any two matroids $M_1=(V_1,\cI_1)$ and $M_2=(V_2,\cI_2)$
the \emph{matroid union} $M=M_1 \lor M_2$ is the matroid $M=(V,\cI)$
where $V=V_1 \cup V_2$ and $\cI=\{I_1 \cup I_2 \mid I_1 \in \cI_1, I_2 \in \cI_2\}$. 
Given representations of $M_1$ and $M_2$, a representation for $M_1 \lor M_2$
can be constructed in randomized polynomial time, possibly
by moving to an extension field (see \cite{OxleyBook}).  
Note that moving to an extension field preserves the characteristic.
The \emph{extension} of a matroid $M=(V,\cI)$ by rank $d$ is the matroid $M \lor M'$
where $M'$ is the uniform matroid of rank $d$ over $V$. 

For a matroid $M=(V,\cI)$, the \emph{$k$-truncation} of $M$
is the matroid $M'=(V,\cI')$ where for $S \subseteq V$,
$S \in \cI'$ if and only if $S \in \cI$ and $|S| \leq k$.
Given a representation of $M$ over a field $\F$, which is either
a finite field or the rationals, a representation of
the $k$-truncation of $M$ over an extension field of $\F$
can be constructed in polynomial time, at the cost of moving
to an extension field of $\F$~\cite{LokshtanovMPS18,Marx09-matroid}.

Given two matroids $M_1=(V,\cI)$ and $M_2=(V,\cI)$, the \emph{matroid
  intersection} problem is to find a common basis $B$ of $M_1$ and $M_2$.
Matroid intersection can be solved in polynomial time, with a variety
of methods~\cite{Schrijver}. In this paper, with a focus on linear
matroids, we note that the Cauchy-Binet formula yields an enumerating
polynomial for matroid intersection, and thereby a randomized
efficient algorithm. More generally, given a matroid $M=(V,\cI)$
and a partition $E$ of $V$ into pairs, the \emph{matroid matching}
(or \emph{matroid parity}) problem is to find a basis $B$ of $M$
which is a union of $|B|/2$ pairs. Matroid matching is infeasible in
general, but efficiently solvable over linear matroids~\cite{Lovasz79,Schrijver}.

\subsection{Enumerating polynomials} \label{sec:enum-poly}

Let $V$ be a ground set and $\cF \subseteq 2^V$ be a set family over $V$.
An \emph{enumerating polynomial} over a set of variables $X = \{ x_v \mid v \in V \}$ and auxiliary variables $Y$ over a field $\F$ is
\begin{align*}
  P(X, Y) = \sum_{S \in \cF} Q_S(Y) \prod_{v \in S} x_v,
\end{align*}
where $Q_S(Y)$ for $S \in \cF$ is a polynomial over $\F$ that is not identically zero.
We give useful examples of enumerating polynomials that can be efficiently evaluated below.

\paragraph*{$k$-walks.}
For a directed graph $G = (V, E)$, two vertices $s, t \in V$, and an integer $k$, an enumerating polynomial for $k$-walks from $s$ to $t$ is defined as follows.
Let $X = \{ x_{s,0} \} \cup \{ x_{v,i} \mid v \in V, i \in [k] \} \cup \{ x_{e,i} \mid e \in E, i \in [k] \}$ be variables.
For every $A_i$ define a $V \times V$-matrix $A_i$ with
\begin{align*}
  A_i[u,v] = \begin{cases}
    x_{v,i} x_{uv,i} & \text{ if } uv \in E \\
    0 & \text{ otherwise.}
  \end{cases}
\end{align*}
Then, the polynomial
\begin{align*}
  P(X) = x_{s,0} \cdot e_s^T A_1 A_2 \cdots A_{k} e_t,
\end{align*}
where $e_s$ and $e_t$ are the unit vectors with $e_s[s] = 1$ and $e_t[t] = 1$, enumerates all (labelled) $k$-walks from $s$ to $t$.
The polynomial can be defined for undirected graphs analogously.

\paragraph*{Matroid Intersections.}
For linear matroids $M_1 = (V, \cI_1), M_2 = (V, \cI_2)$ with the ground set $V$ represented by $A_1, A_2 \in \F^{k \times V}$, let $X = \{ x_v \mid v \in V \}$ be variables for $V$.
Then, by the Cauchy-Binet formula, the polynomial
\begin{align*}
  P(X) = \det A_1 A_X A_2^T = \sum_{B \in \binom{V}{k}} \det A_1[\cdot, B] \det A_2[\cdot, B] \prod_{v \in B} x_v,
\end{align*}
where $A_X$ is a diagonal matrix of dimension $V \times V$ with $A_X[v, v] = x_v$ for every $v \in V$, enumerates all matroid intersection terms.
Particularly, we obtain an effective evaluation of an enumerating polynomial for branchings in directed graphs.
Recall that an out-branching (in-branching) is a rooted tree with every arc oriented away from (towards) the root.
This can be expressed as the intersection of a graphic matroid and a partition matroid, where the partition matroid ensures that every vertex has in-degree (out-degree) at most one.
Hence this is a special case of matroid intersection and an enumerating polynomial for out-branchings and in-branchings can be efficiently evaluated.
Alternatively, one can use the directed matrix-tree theorem (see \cite{BjorklundKK17directed,gessel95algebraic}).

\paragraph*{Perfect matchings.}
For an undirected graph $G = (V, E)$ (with a fixed ordering $<$ on $V$), the \emph{Tutte matrix} is defined by
\begin{align*}
  A[u,v] = \begin{cases}
    x_{uv} & \text{ if } uv \in E \text{ and } u < v \\
    -x_{uv} & \text{ if } uv \in E \text{ and } v < u \\
    0 & \text{ otherwise}.
  \end{cases}
\end{align*}
Then, the Pfaffian $\Pf A$ enumerates all perfect matchings, which can be efficiently evaluated using an elimination procedure.
For an integer $k$, all $k$-matchings (matchings with $k$ edges) also can be enumerated: Introduce $n - 2k$ vertices that are adjacent to all vertices in $G$; the Pfaffian of the resulting graph enumerates all $k$-matchings.

\subsection{Conventions -- fields and representations}
\label{subsec:conventions}
We frequently make the assumption that operations are performed over a
\emph{sufficiently large} field (typically of characteristic 2); e.g.,
that a given polynomial $P(X)$ can be evaluated over a sufficiently
large field and that a linear matroid $M$ is represented over a
sufficiently large field. The precise notion of ``sufficiently large''
depends somewhat on context, but typically we can assume with no
significant impact on complexity that we are working over a finite field $\F$ with
$|\F|=2^{O(n)}$ elements, where $n=|X|$ is the size of the variable
set of $P$ (respectively, the ground set of the matroid). 

To briefly justify this assumption, we make some quick notes. 
For a finite field $\F$, members of $\F$ can be represented in $O(\log |\F|)$ bits, 
and arithmetic operations over $\F$ can be performed in $\tilde O(\log |\F|)$ time,
where $\tilde O$ hides factor of $O(\log \log |\F|)$. Thus,
in a context where our running times are only given up to a polynomial
factor, such an assumption on $|\F|$ does not have a major impact.
Furthermore, if $P(X)$ is already represented over a field $\F$,
then we generally assume that we can choose to carry out the
evaluation over an extension field of $\F$ instead, if needed.
We refer to von zur Gathen and Gerhard~\cite{vzGathen-MCA} for
background on these algorithms. 

The field size is relevant in our algorithms in two ways. First,
given $P(X)$ and $M$, our algorithms reduce down to testing whether a
second, implicitly defined polynomial $Q(X)$ is non-zero, which
is done via the Schwartz-Zippel lemma. For this step, it suffices
that $|\F| \gg d$ where $d$ is the degree of $P(X)$, where
typically $d=n^{O(1)}$. Second, we require that the matroid $M$ can be
represented over $\F$. This gives a bound on $|\F|$ that depends on
the matroid, but for all the matroids considered in this paper,
a linear representation can be efficiently produced (with high probability)
over any field with $|\F|=2^{\Omega(n)}$ elements. See~\cite{OxleyBook,Marx09-matroid}.

In a setting where this overhead is unacceptable, we have two options.
The main option is to assume that the matroid $M$ is representable over
a smaller field than above, e.g., a field of size $2^{(k+\log n)^{O(1)}}$
so that the time per field operation is restricted to $(k+\log n)^{O(1)}$.
Among the matroids listed in Section~\ref{sec:prel-matroids},
this covers everything except gammoids. In particular, 
graphic and co-graphic matroids can be represented over
any field; uniform matroids and partition matroids can be represented over any field of size at least $n$;  
and a transversal matroid of rank $k$ can be represented over a field
of size $2^{\Omega(k \log n)}$ by the use of the Schwartz-Zippel lemma.
It also covers the use of matroid union (disjoint or not)
and (randomized) truncation to rank $k^{O(1)}$.

Another option might be to follow Koana and Wahlström~\cite{KoanaW25stacs}
and consider \emph{approximate representations} -- randomized
representations of matroids $M'$ over a smaller field size, whose
independent sets are a subset of the independent sets of $M$
and where every independent set of $M$ has some probability $1-\varepsilon$
of being independent in $M'$. We refrain from pursuing this in any detail,
in order not to complicate our results needlessly.

In all the above, the randomness is one-sided -- the errors can
consist of false negatives but not false positives.

\section{Determinantal sieving}
\label{sec:sieving} 

\subsection{Over a field of characteristic 2}

We show that, with only evaluation access to a polynomial (over a
field of characteristic 2), we  can decide whether its expansion contains a monomial whose support spans a linear matroid
(equivalently, contains a basis as a subset).
We will give two sieving algorithms, one that sieves for terms that are also independent (basis sieving) and the other that sieves for terms whose odd support sets are spanning (odd sieving).
One could derive basis sieving from odd sieving using polynomial interpolation (Lemma~\ref{lemma:interpolation}).
We will, however, give a direct proof for basis sieving as well because basis sieving itself has applications.
Typically, basis sieving is useful when we are searching for a solution of size exactly $k$ (regardless of whether the objective is maximisation or minimisation).
Odd sieving is particularly powerful when we want to exclude variables in the support set with even (typically 2) contributions.
See Sections~\ref{sec:diverse} and \ref{sec:paths-linkages} for such applications.

We begin with a support statement. This is the central observation
for our sieving algorithms.

\begin{lemma} \label{lemma:sieve-determinant}
  Let $A \in \F^{k \times k}$ be a matrix over a field $\F$ of
  characteristic 2 and define the polynomial
  \[
    P(y_1,\ldots,y_k) = \prod_{i=1}^k \sum_{j=1}^k y_j A[j,i].
  \]
  Then the coefficient of $\prod_{i=1}^k y_i$ in $P$
  is $\det A$. 
\end{lemma}
\begin{proof}
  Expanding the product into monomials, we get precisely
  \[
    \sum_{f \colon [k] \to [k]} \prod_{i=1}^k y_{f(i)} A[f(i),i]
  \]
  where $f$ ranges over all mappings $[k] \to [k]$.
  Considering only those terms of the sum which contain
  all variables $y_i$, $i \in [k]$ we find that the coefficient of
  $\prod_{i=1}^k y_i$ is precisely a sum over all transversals of $A$,
  i.e., $\det A$, in particular, since $\F$ is of characteristic 2
  the sign term of the determinant disappears. 
\end{proof}

If performed over fields of characteristic other than 2, then instead
of $\det A$ the coefficient is the permanent of $A$ (while over fields
of characteristic 2, the permanent and the determinant agree). To
cover applications for fields of other characteristics, we instead use
the \emph{exterior algebra}; see Section~\ref{sec:extensor}. For the
below, we focus on applications over fields of characteristic 2. 

\subsubsection{Sieving for bases} \label{sec:sieve-basis}

The following is the most immediate application of Lemma~\ref{lemma:sieve-determinant}
(proving Theorem~\ref{ithm:simple} from the introduction).
We also add an observation about tighter running time when the
polynomial is already homogeneous. 

\begin{theorem}[Basis sieving] \label{theorem:sieve-basis}
  Let $X=\{x_1,\ldots,x_n\}$ be a set of variables and 
  let $P(X)$ be a polynomial of degree $d$ over a field $\F$ of characteristic 2.
  Let $A \in \F^{k \times n}$ be a matrix representing a matroid $M=(X, \cI)$. 
  In time $O^*(2^k)$ and polynomial space, using evaluation access to $P$, 
  we can test if the monomial expansion of $P$ contains a multilinear monomial
  $m$ whose support is a basis for $M$. Our algorithm is randomized with no false
  positives and failure probability at most $2k/|\F|$.
  If $P$ is a homogeneous polynomial of degree~$k$, then the polynomial overhead
  disappears and the running time is $O(2^k(nk\cdot f +T))$ where
  $f$ is the time for a field operation and $T$ is the time to
  evaluate $P$.
\end{theorem}
\begin{proof}
  Let $P_k(X)$ denote the homogeneous degree $k$ part of $P(X)$, i.e.,
  for every monomial $m$ of total degree $k$ the coefficient of $m$ in $P_k$
  is $P_k(m)=P(m)$, and for every other monomial $m$ we have $P_k(m)=0$. 
  By Lemma~\ref{lemma:interpolation}, one can evaluate $P_k(X)$ in
  $\tilde{O}(d)$ field operations using $O(d)$ evaluations of $P$.
  To simplify notation, we will write $P$ in place of $P_k$.

  Introduce a set of variables $Y=\{y_1,\ldots,y_k\}$ and define a new polynomial
  \[
    P'(X,Y) = P_k\left(x_1 \sum_{i=1}^k y_i A[i,1], \ldots, x_n \sum_{i=1}^k y_i A[i,n]\right).
  \]
  Let $Q(X,Y)$ be the coefficient of $\prod_{i \in [k]} y_i$ in $P'(X,Y)$, which
  can be computed from $2^k$ evaluations of $P'$ (hence of $P$),
  using the method of inclusion-exclusion in Lemma~\ref{lemma:inclusion-exclusion}. Since $Q$ is
  obtained from $P'$ via substitutions, it suffices to consider its effect on a single monomial at a time.
  Let $m=x_1^{m_1} \cdots x_n^{m_n}$ be a monomial in the expansion of $P$.
  Let $(i_1,\ldots,i_k)$ be the sequence of non-zero indices of $m$
  repeated with multiplicity according to degree, in non-decreasing order; e.g.,
  a monomial $x_1^3x_4^2$ corresponds to sequence $(1,1,1,4,4)$.
  In the evaluation of $P'$, each monomial $m$ in $P$ turns into
  \[
    P(m) \cdot X^m \cdot \prod_{p=1}^k \sum_{j=1}^k y_j A[j,i_p]
  \]
  where $P(m)$ is the coefficient of $m$ in $P$. Using Lemma~\ref{lemma:sieve-determinant},
  the contribution of the monomial $m$ to $Q(X,Y)$ is precisely
  \[
    P(m) \cdot X^m \cdot Y_k \cdot \det A[\cdot,(i_1,\ldots,i_k)]
  \]
  where $Y_k=\prod_{i=1}^k y_i$ and $A'=A[\cdot,(i_1,\ldots,i_k)]$ denotes the
  matrix consisting of columns $i_j$ of $A$ included with multiplicity. 
  Now, if $m$ is not multilinear, then the resulting matrix $A'$ has a repeated column
  and is clearly singular, so $m$ does not contribute to $Q$. If $m$ is multilinear,
  then $m$ contributes a non-zero value to $Q$ if and only if the
  support of $m$ spans $M$. Furthermore, since the first part
  $P(m) X^m$ of this expression is precisely the value of the original
  monomial $m$ in $P$, no further algebraic cancellation occurs in $Q$.
  Hence $Q$ enumerates monomials corresponding to multilinear
  monomials in $P$ whose support spans~$M$. The result now follows
  from a random evaluation of $P$ using Schwartz-Zippel. In particular, 
  $Q$ has degree $2k$ since the sieving started from $P_k(X)$.

  For the case that $P$ is homogeneous, we can bypass the phase of
  extracting $P_k(X)$ and use $P(X)$ directly. The polynomial $Q(X,Y)$
  is defined as a sum over $2^k$ evaluations of $P(X)$
  with arguments $x_j \sum_{i=1}^k y_i A[i,j]$, $j \in X$. 
  The precise running time follows with no additional tricks.   
\end{proof}

We note a variant of this. Instead of every variable $x_v$ being associated with
only one column $v$ of~$A$, we may wish for each variable to be associated with
multiple columns. 

\begin{corollary} \label{cor:basis}
  Let $X=\{x_1,\ldots,x_n\}$ be a set of variables and 
  let $P(X)$ be a polynomial of degree $d$ over a field $\F$ of characteristic 2.
  Let $A \in \F^{k \times V}$ be a matrix representing a matroid $M=(V, \cI)$ of rank~$k$. 
  Suppose that each variable $x_i$ is associated with pairwise disjoint subsets $\Gamma_i \subseteq V$ of size $\gamma_i$.
  In time $O^*(2^k)$ and polynomial space, using evaluation access to $P$, 
  we can test if the monomial expansion of $P$ contains a multilinear monomial
  $m$ such that $\bigcup_{i \in \supp(m)} \Gamma_i$ is a basis for $M$. Our algorithm is randomized with no false
  positives and failure probability at most $2k/|\F|$.
\end{corollary}
\begin{proof}
  Define a new set of variables $X' = \{ x_{i,v}' \mid i \in [n], v \in \Gamma_i \}$, and
  applying Theorem~\ref{theorem:sieve-basis} to the polynomial $P'(X')$
  resulting from an evaluation where
  \[
    x_i = \prod_{v \in \Gamma_i} x_{i,v}'
  \]
  for every $x_i \in X$, yields the desired result.
\end{proof}

\subsubsection{Sieving for spanning sets} \label{sec:odd-sieve}

We give the odd sieving algorithm, proving Theorem~\ref{ithm:odd-sieve}.
The proof is similar to that of basis sieving.
In particular, we give a variant as in Corollary~\ref{cor:basis}, where each variable is associated with a subset of elements from the matroid.
Let us illustrate why only the odd support sets pass through the sieve.
To sieve for spanning sets, we basically need to replace each variable $x_i$ with $1 + x_i$.
Then, $(1 + x_i)^{m_i} = 1 + m_i x_i + \binom{m_i}{2} x_i^2 + \cdots$ (for a monomial $m$) becomes $1 + m_i x_i$ because only multilinear terms survive, and further reduces to $1$ if $m_i$ is even.  
So a variable with even contributions effectively diminishes.
We give the formal proof:

\begin{theorem}[Odd sieving] \label{theorem:matroid-sieving}
  Let $P(X)$ be a polynomial over a variable set $X = \{ x_1, \dots, x_n \}$ over a field $\F$ of characteristic 2 with degree~$d$.
  Let $A \in \F^{k \times V}$ be a matrix representing a matroid $M=(V, \cI)$ of rank~$k$. 
  Suppose that each variable $x_i$ is associated with pairwise disjoint subsets $\Gamma_i \subseteq V$ of size $\gamma_i$.
  Given black-box (evaluation) access to a polynomial $P(X)$, we can test in randomized $\Oh^*(2^{k})$ time with failure probability at most $\delta = (d + k) / |\F|$ and in polynomial space, whether $P$ contains a term in the monomial expansion of $P(X)$
  such that $\Gamma_S = \bigcup_{i \in S} \Gamma_i$ is a basis of $M$, where $S \subseteq X$ is a subset of its odd support set with $\sum_{i \in S} \gamma_i = k$.
\end{theorem}
\begin{proof}
  We will define a polynomial $Q$ such that it evaluates to non-zero with probability at least $1 - \delta$ if it contains a monomial as stated in the lemma and to zero otherwise.
  For every $i \in [n]$, we define
  \[
    x_i^* = x_i'' (1+z^{\gamma_i} x_i' \prod_{q \in \Gamma_i} \sum_{p \in [k]} y_{p} A[p,q]),
  \]
  where $x_i'$, $x_i''$ for $i \in [n]$, $y_p$ for $p \in [k]$, and $z$ are new variables. 
  Let $X' = \{ x_1' \ldots, x_n' \}$, $X'' = \{ x_1'', \ldots, x_n'' \}$, and define a polynomial $Q(X', X'')$ that 
  sieves for those terms in the monomial expansion of
  $P^* = P(x_1^*,\ldots,x_n^*)$ that contain precisely $k$ contributions of $z$ and 
  which contain $y_p$ for each $p \in [k]$.
  By Lemmas~\ref{lemma:interpolation} and \ref{lemma:inclusion-exclusion}, $Q(X', X'')$ can be evaluated using $O^*(2^{k})$ evaluations of $P$.

  The expansion in $P^*$ corresponding to $m$ is
  \begin{align*}
    P^*_m
    &= P(m) \cdot (X'')^m \cdot \prod_{i \in \supp(m)} (1+z^{\gamma_i} x_i' \prod_{q \in \Gamma_i} \sum_{p \in [k]} y_{p} A[p, q])^{m_i} \\
    &= P(m) \cdot (X'')^m \cdot  \sum_{m^*} \prod_{i \in \supp(m^*)} \binom{m_i}{m_i^*} (z^{\gamma_i} x_i' \prod_{q \in \Gamma_i} \sum_{p \in [k]} y_p A[p, q])^{m_i^*} \\
  \end{align*}
  where $m^*$ ranges over all monomials that divide $m$.
  The last equality is due to the binomial theorem.
  Simplifying further gives
  \begin{align*}
  P_m^* = P(m) \cdot (X'')^m \cdot \sum_{m^*} z^{\deg(m^*)} \prod_{i \in \supp(m^*)} \left( \binom{m_i}{m_i^*} (x_i)^{m_i^*} \prod_{q \in \Gamma_i} \sum_{p \in [k]} y_p A[p, q]^{m_i^*} \right),
  \end{align*}
  where $\deg(m_i^*) =\sum_{i \in \supp(m^*)} m_i^*$.
  It follows that the coefficient $z^k$ in $P_m^*$ is the sum over all monomials of degree $k$ that divide $m$.
  By Lemma~\ref{lemma:sieve-determinant}, the coefficient of $z^k \prod_{i \in [k]} y_i$ in $P^*(m)$ is thus
  \begin{align*}
    Q_m = P(m) \cdot (X'')^m \cdot \sum_{m'}  \det A_{m'} \prod_{i \in \supp(m')} \binom{m_i}{m_i'} (x_i')^{m_i'} ,
  \end{align*}
  where $m'$ ranges over all monomials of degree $k$ that divide $m$ and $A_{m'}$ is the $k \times k$-matrix that contain $m_i'$ copies of $A[\cdot, \Gamma_i]$ for each $i \in \supp(m')$.
  If $A_{m_i'}$ contains duplicate columns (i.e., $m_i' \ge 2$ for some $i$), then $\det A_{m_i'} = 0$, and thus we may assume that $m'$ is multilinear.
  Hence, we obtain
  \[
    Q_m = P(m) \cdot (X'')^m \cdot \sum_{m'} \det A_{m'} \prod_{i \in \supp(m')} m_i x_i' ,
  \]
  where $m'$ ranges over all multilinear monomials of degree $k$ that divide $m$.
  Since $\F$ has characteristic 2, the summand correspond to $m'$ is non-zero only if $\supp(m')$ is contained in the odd support of $m$.

  On the other hand, for every pair of monomials $m$ and $m'$ such that $m'$ divides $m$ and $A_{m'}$ is non-singular, there is a term
  \[
    P(m) \cdot \left( \prod_{i \in \supp(m')} m_i\right) (X'')^m  \cdot (X')^{m'} \det A_{m'}    .
  \]
  Since the variables $x_i'$ and $x_i''$ are newly added variables,
  this term does not cancel against any other term from the expansion
  of $Q(m)$.
  More specifically, these variables uniquely indicate the
  combination of the monomials $m$ and $m'$. 
  We evaluate $Q$ for variables $x_i', x_i''$ randomly chosen from $\F$.
  Since $Q$ has degree most $d + k$, by the Schwartz-Zippel lemma, the probability that $Q$ evaluates to zero at most $(d + k) / |\F|$.
\end{proof}

\subsubsection{Multilinear and constrained multilinear detection
  as determinantal sieving}
\label{ssec:mld-compare}

We now make explicit the claim from earlier, that the known
polynomial sieving-based algorithms for multilinear detection~\cite{Bjorklund14detsum,BjorklundHKK17narrow}
and constrained multilinear detection~\cite{BjorklundKK16}
due to Björklund et al.\ are equivalent to 
applications of the algorithm of Theorem~\ref{theorem:sieve-basis}.

\paragraph{User's guide -- the less technical view.}
Let us point out that what we are pursuing here is a technical statement
about the algorithms themselves. If we only want to reproduce
the effect of these previous sieving methods, then we can do so
much easier via combinatorial arguments over matroids as follows.
\begin{itemize}
\item \emph{Multilinear detection} of rank $k$ over a ground set $V$
  is equivalent to basis sieving with a matroid $M$ where every set of
  $k$ elements from $V$ is independent. This is precisely the uniform matroid $M=U_{n,k}$.
\item Another application which is frequently seen is to look for a ``colourful''
  term, where elements of $V$ have $k$ different colours, and we are looking for
  a term that contains every colour. In this case, $M$ would be a unit partition matroid. 
  For example, this covers the \textsc{$T$-Cycle} problem by assigning a private
  colour to every terminal in $T$ and looking for a colourful cycle;
  and \textsc{Steiner Tree} and \textsc{Group Steiner Tree}
  are similarly reduced to the \textsc{Rank $k$ Connected Subgraph} problem.
\item Finally, \emph{constrained multilinear detection} is the following setting:
  the ground set $V$ is coloured from a set of colours $C$, where every colour $q \in C$
  has a capacity $d_q \in \N$ restricting how many times it can be used.
  The corresponding matroid is then precisely the $k$-truncation of a non-unit partition matroid.
  Equivalently, it can be constructed directly as the gammoid of a simple digraph,
  with $k$ sources, $d_q$ internal vertices for every colour $q \in C$ connected to all sources,
  and each element $v \in V$ of colour $q$ being connected to all internal vertices representing colour $q$.
\end{itemize}
For further pointers, see material covering matroid theory~\cite{OxleyBook,Marx09-matroid,Schrijver}.
We now proceed with the more technical demonstrations.

\paragraph{Multilinear detection.} 
We recall the procedure of sieving for multilinear detection, as
presented in Björklund et al.~\cite{BjorklundKK16}.
We demonstrate that their procedure is mathematically equivalent
to an application of determinantal sieving with a random linear matroid.
Let us first recall their method. 
Let $P(X)$ be a homogeneous polynomial of degree $k$ on $n$ variables
$X=\{x_1,\ldots,x_n)$ over a field of characteristic~2. For every $i \in [n]$ and $j \in [k]$ define a
variable $z_{i,j}$, and for $J \subseteq [k]$ and $i \in [n]$
define $z_i^J=\sum_{j \in J} z_{i,j}$.
Let $Z=\{z_{i,j} \mid i \in [n], j \in [k]\}$ and define
\begin{equation}
  Q(Z) = \sum_{J \subseteq [k]} P(z_1^J, \ldots, z_n^J).
  \label{eq:mld}
\end{equation}
Then $Q(Z)$ is not identically zero if and only if $P(X)$ contains a
multilinear monomial~\cite{BjorklundKK16}.

We argue that this is precisely the procedure of Theorem~\ref{theorem:sieve-basis}
applied to the matrix $A$ where $A[j,i]=z_{i,j}$.
Indeed, expanding the inclusion-exclusion step Theorem~\ref{theorem:sieve-basis}
applied to $P(X)$ and $A$
computes
\[
  Q(X,Y) = \sum_{I \subseteq [k]} P_{-I}(x_1 \sum_{j=1}^k y_j A[j,1], \ldots, x_n \sum_{j=1}^k y_j A[j,i])
\]
over the evaluation where $y_j=0$ for $j \in I$. Consider an evaluation $Q(X,1)$.
Then for each index $i \in [n]$, the $i$-th argument of the evaluation is
\[
  x_i \sum_{j \in [k] \setminus I}  y_j z_{i,j} = x_i z_i^{[k] \setminus I}.
\]
Hence (\ref{eq:mld}) is an evaluation of $Q(1,1)$,
and the classical polynomial sieving method for multilinear detection 
can be seen as an instance of determinantal sieving for the matrix $A=(z_{j,i})_{(i,j)}$
consisting entirely of random (independent) values.
Assuming the field $\F$ is large enough, this matrix $A$ is with high
probability a representation of the uniform matroid $U_{n,k}$.

Parenthetically, evaluating at $Y =1$ in Theorem~\ref{theorem:sieve-basis}
as above is always safe -- since distinct monomials in $P(X)$ have
distinct contributions in $X$, setting $Y=1$ does not cause any
undesired cancellations. Evaluating at $X=1$ is not normally safe,
but in the precise case of the matrix $A$ all contributed determinants
$\det A[\cdot,U]$ for $|U|=k$ are distinct polynomials over the
variables $Z$, hence it works in equation (\ref{eq:mld}).

\paragraph{Constrained multilinear detection.}
We now consider the somewhat more involved sieving used for
constrained multilinear detection~\cite{BjorklundKK16},
and demonstrate a similar equivalence.
Let $C$ be a set of colours,  and for each $q \in C$
let $d_q \in N$ denote the capacity of colour $q$.
Refer to such a monomial as a \emph{properly coloured} monomial.
We shift the notation slightly. 
Let $P(X)$ be a homogeneous polynomial of degree $k$ over $n$
variables $X=\{x_1,\ldots,x_n\}$ over a field of characteristic 2 and let $c \colon X \to C$
be a colouring and recall that we are sieving for multilinear
monomials $m$ in $P(X)$ where for every colour $q \in C$,
$m$ contains at most $d_q$ variables of colour $q$.
Define two sets of auxiliary variables
\[
  V=\{v_{i,s} \mid i \in [n], s \in [d_{c(i)}]\}
\]
where $s$ ranges over \emph{shades} of colour $c(i) \in C$,
and
\[
  W=\{w_{q,s,j} \mid q \in C, s \in [d_q], j \in [k]\}.
\]
Finally, for $J \subseteq [k]$ and $i \in [n]$ define
\[
  u_i^J = \sum_{j \in J} u_{i,j} \quad \text{where} \quad
  u_{i,j} = \sum_{s \in [d_{c(i)}]} v_{i,s} w_{c(i),s,j}.
\]
Then
\begin{equation}
  Q(V,W) = \sum_{J \subseteq [k]} P(u_1^J, \ldots u_n^J)
  \label{eq:cmld}
\end{equation}
is not identically zero if and only if $P(X)$ has a properly coloured 
multilinear monomial~\cite{BjorklundKK16}. 
Analogously to unconstrained multilinear detection,
we note that this is equivalent to applying Theorem~\ref{theorem:sieve-basis},
evaluated at $X=Y=1$, to the matrix $A \in \F^{k \times n}$
where
\[
  A[j,i] = \sum_{s \in [d_{c(i)}]} v_{i,s} w_{c(i),s,j}.
\]
We note a factorization of~$A$. Let $S=\{(q,s) \mid q \in C, s \in [d_q]\}$
be the set of all shades of all colours used above. 
Define $B \in \F^{k \times S}$ and $C \in \F^{S \times n}$
by
\[
  B[j,(q,s)] = w_{q,s,j} \quad \text{and} \quad
  C[(q,s),i] = 
  \begin{cases}
    v_{i,s} & c(i)=q \\
    0 & \text{otherwise}
  \end{cases}
\]
Then $A=BC$ by direct expansion of the matrix multiplication.
Here, $C$ is a representation of the transversal matroid
mapping each variable index $i \in [n]$
to the set of shades $(c(i),s)$ available to it,
hence a set of variables is independent in $C$
if and only if the corresponding monomial is properly coloured.
Finally, $B$ is again a fully random matrix, hence
multiplication by $B$ represents the truncation of $C$ to rank $k$.
Hence (\ref{eq:cmld}) corresponds to Theorem~\ref{theorem:sieve-basis}
applied to the $k$-truncation of the ``colour assignment'' matroid $C$,
which is precisely the encoding discussed in the introduction.

\subsection{Over general fields}
\label{sec:extensor}

We give two sieving algorithms for general fields.
First, we present an algorithm for what we call \emph{strongly monotone circuits} -- circuits without any cancellation, informally speaking.
Our second algorithm works for arbitrary arithmetic circuits albeit with a worse running time.

To sieve over general fields, we use the exterior algebra.
For a field $\F$, $\Lambda(\F^k)$ is a $2^k$-dimensional vector space where there is a basis $\{ e_I \mid I \subseteq [k] \}$.
Each element $a = \sum_{I \subseteq [k]} {a_I} e_I$ is called an \emph{extensor}.
For $i \in \{ 0, \ldots, k \}$, we denote by $\Lambda^i(\F^k)$ the vector subspace spanned by bases $e_I$ with $|I| = i$.
For instance, $\Lambda^0(\F^k)$ is isomorphic to $\F$ and $\Lambda^1(\F^k)$ is isomorphic to the vector space $\F^k$, so we will use them interchangeably.
The addition in $\Lambda(\F^k)$ is defined in the element-wise manner.
The multiplication in $\Lambda(\F^k)$ is called \emph{wedge product}, and it is defined as follows:
If $I \cap J \ne \emptyset$, then $e_I \wedge e_J = 0$.
If $I$ and $J$ are disjoint, then $e_I \wedge e_J = (-1)^{\sigma(I, J)} e_{I \cup J}$, where $\sigma(I, J) = \pm 1$ is the sign of the permutation that maps the concatenation of $I$ and $J$ each in increasing order into the increasing sequence of $I \cup J$.
Over vectors $v, v' \in \F^k$, we have anti-commutativity, i.e., $v \wedge v' = -v' \wedge v$, and in particular, $v \wedge v = 0$.
The key property of exterior algebra is that for a matrix $A \in \F^{k \times k}$ with $a_i = A[\cdot, i]$, we have $a_1 \wedge \cdots \wedge a_k = \det A \cdot e_{[k]}$, where $e_{[k]}$ is the basis extensor $e_1 \wedge \cdots \wedge e_k$.
For instance, when $k = 2$,
\begin{align*}
  &(a_{11} e_1 + a_{21} e_2) \wedge (a_{12} e_1 + a_{22} e_2) \\
  &= a_{11} a_{12} \cdot e_1 \wedge e_1 + a_{11} a_{22} \cdot e_1 \wedge e_2 + a_{21} a_{12} \cdot e_2 \wedge e_1 + a_{21} a_{22} \cdot e_2 \wedge e_2 \\
  &= 0 + a_{11} a_{22} \cdot e_1 \wedge e_2 - a_{21} a_{12} \cdot e_1 \wedge e_2 + 0 = (a_{11} a_{22} - a_{12} a_{21}) \cdot e_1 \wedge e_2.
\end{align*}
So a matrix is non-singular if and only if the wedge product of its columns are non-zero.

An extensor $a \in \Lambda(\F^k)$ is \emph{decomposable} if there are vectors $v_1, \dots, v_{\ell}$ such that $a = v_1 \wedge \cdots \wedge v_{\ell}$.
A decomposable extensor $a$ is zero if the vectors $v_1, \dots, v_{\ell}$ are linearly dependent.
For two decomposable extensors $a, a'$, it holds that $a \wedge a' = \pm a' \wedge a$ (this is generally not the case, e.g., $e_1 \wedge (e_2 \wedge e_3 +  e_4) = e_1 \wedge e_2 \wedge e_3 + e_1 \wedge e_4$ and $(e_2 \wedge e_3 +  e_4) \wedge e_1 = e_1 \wedge e_2 \wedge e_3 - e_1 \wedge e_4$).

The sum of two extensors can be computed with $2^{k}$ field operations.
The wedge product $a \wedge a'$ of two extensors $a \in \Lambda(\F^k)$ and $b \in \Lambda^i(\F^k)$ can be computed with $2^k \binom{k}{i}$ field operations according to the definition (hence $O^*(2^k)$ time for $i \in O(1)$).
In general, there is an $O(2^{\omega k / 2})$-time algorithm to compute the wedge product, given implicitly by W\l{}odarczyk \cite{Wlodarczyk19} (see the thesis of Brand~\cite{Brand19thesis} for a more explicit exposition).

Suppose that a polynomial $P(X)$ is represented by an arithmetic circuit $C$.
An \emph{arithmetic circuit} is a directed acyclic graph with a single sink (called output gate) in which every source is labelled by a variable $x_i$ or an element of $\F$ (called input gate) and every other node is labelled by addition (called sum gate) or multiplication (called product gate).
We will assume that every sum and product gate has in-degree 2.
An arithmetic circuit is called \emph{skew} ($\delta$-skew) if at least one input of every product gate is an input gate (has polynomial degree at most $\delta$, respectively).
An arithmetic circuit over the field of rationals $\Q$ is called \emph{monotone} if every constant is non-negative.
We say that an arithmetic circuit (over any field) is \emph{strongly monotone} if the following~hold:
\begin{itemize}
  \item For each gate, the corresponding polynomial is multilinear, which implies that each input to the sum or product gate can be represented as a set family $\cF$ over $X$ and coefficients $c \colon \cF \to \F \setminus \{ 0 \}$.
  \item For every sum gate with two inputs $(\cF, c)$ and $(\cF', c')$, $\cF \cap \cF' = \emptyset$.
  \item For every product gate with two inputs $(\cF, c)$ and $(\cF', c')$, $S \cup S'$ is distinct for every $S \in \cF$ and $S' \in \cF'$.
\end{itemize}
At first glance, the condition for strong monotonicity may seem very restrictive.
However, any ``cancellation-free'' circuit can be turned into an equivalent strongly monotone circuit, often without blowing up its size: simply make $d$ copies of sub-circuits for each gate with out-degree $d > 1$.
Note that, for every input gate $g$ for the variable $x$ with out-degree $d$, we will have $d$ input gates each labelled by a new variable, say $x_i$.
By associating the variables $x_i$ with one vector, the resulting circuit is essentially equivalent to the original.
See e.g., $O^*(2^{qk})$-time algorithm for \textsc{$q$-Matroid Parity} (Theorem~\ref{theorem:general-matroid-parity}) in Section~\ref{sec:matroid-applications}.

\begin{theorem} \label{lemma:simple-matroid-sieving}
  Let $C$ be a strongly monotone arithmetic circuit computing a multilinear polynomial $P(X)$ of degree $d$ over a variable set $X=\{x_1,\ldots,x_n\}$ and a field $\F$.
  Let $A \in \F^{k \times V}$ be a matrix representing a matroid $M=(V, \cI)$ of rank~$k$. 
  Suppose that each variable $x_i$ is associated with a subset $\Gamma_i \subseteq V$ of size $\gamma_i$, and that the subsets $\Gamma_i$ are pairwise disjoint.
  We can test in randomized $O^*(2^{\omega k / 2})$ time with failure probability $d / |\F|$ and in $O^*(2^k)$ space, whether there is a term $m$ in the monomial expansion of $P(X)$ such that $\bigcup_{i \in \supp(m)} \Gamma_i$ is a basis of $M$.
  The running time can be improved to $O^*(2^{k})$ if $C$ is $\delta$-skew for $\delta \in O(1)$ and $\gamma_i \in O(1)$ for all $i$.
\end{theorem}
\begin{proof}
  Fixing an arbitrary ordering of the input of each product gate (this is necessary because the wedge product is not commutative), we evaluate the circuit $C$ over $\Lambda(\F^k)$ by plugging in the extensor $x_i = x_i' a_i$ for every $i \in [n]$, where $x_i'$ is a new variable and $a_i = \bigwedge_{q \in \Gamma_i} A[\cdot, q]$ (the order of wedge products is not important here).
  Let $r \in \Lambda(\F^k)$ denote the result.
  Note that with each variable $x_i'$ substituted by a random element from $\F$, the extensor $r$ can be computed in time $O^*(2^{\omega k / 2})$ (and $O^*(2^k)$ if $C$ is skew and $\max_{i \in [n]} \gamma_i \in O(1)$ -- simply by computing the product according to the definition), as noted in the introduction to exterior algebra.
  We will show that the coefficient of $e_{[k]}$ is non-zero with high probability given that there is a monomial constituting a basis of~$M$.

  We show by induction on the number of gates that
  \[
    r = \sum_{m} \pm P(m) \cdot (X')^m \bigwedge_{i \in \supp(m)} a_{i},
  \]
  where $m$ ranges over all monomials $m$ of $P(X)$.
  For every monomial $m$, the sign $\pm$ depends on the ordering on product gates.
  There are two cases depending on whether the last gate $g$ in $C$ is a sum gate or product gate.
  Let $Q(X) = \sum_{m_Q} Q(m_Q) m_Q$ and $Q'(X) = \sum_{m_Q'} Q(m_Q') m_Q'$ denote its inputs. 
  By the induction hypothesis, suppose that the result of evaluating over the exterior algebra $\Lambda(\F^k)$ is
  \begin{align*}
    q = \sum_{m_Q} \pm Q(m_Q) \left(\prod_{i \in \supp(m_Q)} x_i'\right) \bigwedge_{i \in \supp(m_Q)} a_{i} \text{ and } 
    q' = \sum_{m_Q'} \pm Q'(m_Q') \left(\prod_{i \in \supp(m_Q')} x_i'\right) \bigwedge_{i \in \supp(m_Q')} a_{i}.
  \end{align*}

  First, suppose that $g$ is a sum gate.
  By the strong monotonicity of $C$, each term $m$ in $P$ corresponding to $Q$ (and $Q'$) has coefficient $Q(m)$ (and $Q'(m)$, respectively).
  Thus, for every monomial $m$ in $P$, there is a term $\prod_{i \in \supp(m)} x_i' \bigwedge_{i \in \supp(m)} a_i$ (with the coefficient $\pm Q(m)$ or $\pm Q'(m)$) in $r$.
  We stress that the strong monotonicity is crucial here; suppose that $Q$ and $Q'$ share a term with opposite signs.
  This term should cancel out in $P$, but it does not necessarily when evaluated over $\Lambda(\F^k)$ as the sign may be flipped.

  Next, suppose that $g$ is a product gate.
  By the strong monotonicity of $C$, each term $m$ in $P$ corresponds to a pair of monomials, one from $Q(m)$ and the other from $Q'(m)$.
  We need to verify that for every monomial $m$ of $P(X)$ with $\bigwedge_{i \in \supp(m)} a_i \ne 0$, the corresponding terms in $q \wedge q'$ and $q' \wedge q$ have non-zero coefficients.
  Note that
  \begin{align*}
    q \wedge q' = 
    \sum_{m_Q, m_Q'} \pm Q(m_Q) Q(m_Q') \left(\prod_{i \in \supp(m_Q) \cup \supp(m_Q')} x_i'\right) \bigwedge_{i \in \supp(m_Q)} a_i \wedge \bigwedge_{i \in \supp(m_Q')} a_i.
  \end{align*}
  Since $\bigwedge_{i \in \supp(m_Q)} a_i$ and $\bigwedge_{i \in \supp(m_Q')} a_i$ are both decomposable, $\bigwedge_{i \in \supp(m_Q')} a_i \wedge \bigwedge_{i \in \supp(m_Q)} a_i = \pm \bigwedge_{i \in \supp(m_Q)} a_i \wedge \bigwedge_{i \in \supp(m_Q')}a_i$, and consequently, $q \wedge q'$ and $q' \wedge q$ have the same form possibly with opposite signs.
  In particular, this shows that the induction is correct regardless of how two inputs of product gates are ordered. 

  Thus, the circuit evaluates to
  \begin{align*}
    r = \sum_{m} \pm P(m) \left(\prod_{i \in \supp(m)} x_i'\right) \bigwedge_{i \in \supp(m)} a_i
    = \sum_{m} \pm P(m) \cdot (X')^{m} \cdot \det A[\cdot, \supp(m)] \cdot e_{[k]},
  \end{align*}
  where $m$ ranges over all monomials in $P$ with $\sum_{i \in \supp(m)} \gamma_i = k$.
  Observe that there is no further cancellation between any two terms.
  We evaluate the coefficient of $e_{[k]}$ in $r$ at random coordinates.
  By the Schwartz-Zippel lemma, the result follows.
\end{proof}

We also provide a sieving algorithm for general arithmetic circuits.
The idea is again, to evaluate the circuit over the exterior algebra.
In order to deal with the issue of non-commutativity, we make the assumption that every variable is associated with an even number of matroid elements.
If $v = v_1 \wedge \cdots \wedge v_c$ and $v' = v_1' \wedge \cdots \wedge v_{c'}'$ for even integers $c$ and $c'$, then
\begin{align*}
  v \wedge v'
  &= v_1 \wedge \cdots \wedge v_c \wedge v_1' \wedge \cdots \wedge v_{c'}' \\
  &= (-1)^{cc'} v_1' \wedge \cdots \wedge v_{c'}' \wedge v_1 \wedge \cdots \wedge v_{c} = v' \wedge v.
\end{align*}
Here, the second equality holds because the transposition occurs $cc'$ times.
We thus have commutativity.
 
\begin{theorem} \label{lemma:extensor-matroid-sieving}
  Let $C$ be an arithmetic circuit computing a polynomial $P(X)$ over a variable set $X = \{ x_1, \dots, x_n \}$ and a field $\F$.
  Let $A \in \F^{k \times V}$ be a matrix representing a matroid $M=(V, \cI)$ of rank~$k$. 
  Suppose that each variable $x_i$ is associated with a subset $\Gamma_i \subseteq V$ of \emph{even} size $\gamma_i$, and that the subsets $\Gamma_i$ are pairwise disjoint.
  We can test in randomized $O^*(2^{\omega k / 2})$ time with failure probability $k / |\F|$ and in $O^*(2^k)$ space, whether $C$ contains a term $m$ in the monomial expansion of $P(X)$ such that $\bigcup_{i \in \supp(m)} \Gamma_i$ is a basis of $M$.
  The running time can be improved to $O^*(2^{k})$ if $C$ is $\delta$-skew for $\delta \in O(1)$ and $\gamma_i \in O(1)$ for all $i$.
\end{theorem}
\begin{proof}
  We evaluate the circuit over the algebra $\Lambda(\F^{k})$ by plugging in the extensor $x_i = x_i' a_i$, where $x_i'$ is a new variable and $a_i = \bigwedge_{i \in \Gamma_i} A[\cdot, q]$. 
  Let $r \in \Lambda(\F^{k})$ denote the result.
  Note that with each variable $x_i$ substituted with a random element from~$\F$, the extensor $r$ can be computed in time $O^*(2^{\omega k /2})$ (and $O^*(2^k)$ if $C$ is skew and $\max_{i \in [n]} \gamma_i \in O(1)$).
  As in the proof of Theorem~\ref{lemma:simple-matroid-sieving}, we will show that the coefficient of $e_{[k]}$ is non-zero with high probability given that there is a monomial constituting a basis of $M$.

  Since the evaluation is over a commutative algebra, for every monomial $m$ in $P$, there is a ``term'' in the expansion of $r$:
  \[
    P(m) \cdot \prod_{i \in \supp(m)} (x_i')^{m_i} \cdot \bigwedge_{i \in \supp(m)} a_i.
  \]
  It is straightforward to prove by induction as in the proof of Theorem~\ref{lemma:simple-matroid-sieving} that its coefficient is $P(m)$ (without any sign flip).
  The crucial difference (i.e., no sign flip) arises from the commutativity of the underlying algebra.
  The term corresponding to $e_{[k]}$ in $r$ is
  \[
    \sum_{m} P(m) \left(\prod_{i \in \supp(m)} x_i'\right) \bigwedge_{i \in \supp(m)} a_i
    = \sum_{m} P(m) \cdot (X')^{m} \cdot \det A[\cdot, \supp(m)] \cdot e_{[k]},
  \]
  where $m$ ranges over all monomials in $P$ with $\sum_{i \in \supp(m)} \gamma_i = k$.
  Observe that there is no cancellation between any two terms.
  We evaluate the coefficient of $e_{[k]}$ in $r$ at random coordinates.
  By the Schwartz-Zippel lemma, the result follows.
\end{proof}

One can apply Theorem~\ref{lemma:extensor-matroid-sieving} even if the variables are associated with an odd number of elements, albeit with increased running time, essentially by padding every variable with an additional element.
This is similar to the idea of lift mapping \cite{BrandDH18}.
Using Theorem~\ref{lemma:extensor-matroid-sieving}, we show how to solve \textsc{$q$-Matroid Intersection} in $O^*(2^{(q-2+(q \bmod 2))k})$ time in Theorem~\ref{thm:general-intersection}.

\section{Matroid Covering, Packing and Intersection Problems}
\label{sec:matroid-applications}

For our first application section, we review some fairly
straightforward results regarding matroid variants of the
\textsc{Set Cover} and \textsc{Set Packing} problems defined in Section~\ref{sec:intro} and related problems.
We start off by recalling the definitions.

\textsc{Set Cover} and \textsc{Set Packing} are classical NP-hard problems.
For both problems, the input is a ground set $V$, a set system $\cE \subseteq 2^V$ over $V$, and an integer $t$.
\textsc{Set Cover} asks whether there is a
subcollection $S \subseteq \cE$ such that $|S| \le t$ and $\bigcup S=V$,
i.e., $S$ covers $V$, and
\textsc{Set Packing} asks whether is a subcollection $S \subseteq \cE$ of $t$ pairwise disjoint sets.
Their matroid variants \textsc{Rank $k$ Set Cover} and \textsc{Rank $k$ Set Packing} are defined as follows.
We are given as input a set $V$, a set family $\cE \subseteq 2^V$, a matroid $M = (V, \cI)$ of rank $k$, and an integer $t$.
In \textsc{Rank $k$ Set Cover}, the question is whether there is a subcollection $S \subseteq \cE$ with $|S| \le t$ such that  $\bigcup S$ has rank $k$.
\textsc{Rank~$k$ Set Packing} asks for a subcollection $S$ of pairwise disjoint $t$ sets such that $\bigcup S$ has size $k$ and rank $k$ (i.e., it is a basis of $M$).
Note that \textsc{Set Cover} is the special case of \textsc{Rank $k$ Set Cover} where $M$ is the free matroid, i.e., all subsets of $V$ are independent, and $k = |V|$.
Similarly, \textsc{Set Packing} is the special case of \textsc{Rank $k$ Set Packing} where $M$ is the uniform matroid of rank $k = |\bigcup S|$ for a solution $S$.
One may also consider the apparently more general variant of \textsc{Rank $k$ Set Packing} where one does not require that $\bigcup S$ is a basis for $M$, but only that it is independent. 
However, this reduces to \textsc{Rank $k$ Set Packing} by iterating over the acceptable cardinalities $|\bigcup S|$ and applying matroid truncation.

For $q \in O(1)$, 
the \textsc{$q$-Set Packing} problem is \textsc{Set Packing} in which every set has cardinality $q$.
The \textsc{$q$-Dimensional Matching} problem is a well-studied special case of \textsc{$q$-Set Packing}, where $V$ is partitioned into $q$ sets $V_1, \dots, V_q$, and every set in $\cE$ is from $V_1 \times \cdots \times V_q$.
The matroid analogs to \textsc{$q$-Dimensional Matching} and \textsc{$q$-Set Packing} are \textsc{$q$-Matroid Intersection} and \textsc{$q$-Matroid Parity}, respectively.
In \textsc{$q$-Matroid Intersection}, we are given as input $q$ matroids $M_1, \dots, M_q$ over the same ground set $V$ and an integer $k$, and the question is whether there is a subset $S \subseteq V$ of size $k$ that is independent in $M_i$ (equivalently, a basis for $M_i$ by applying matroid truncation) for every $i \in [q]$.
In \textsc{$q$-Matroid Parity}, we are given as input a matroid $M = (V, \cI)$ with $V$ partitioned into disjoint sets $\cE = \{ E_1, \dots, E_m \}$ each of size $q$, and an integer $k$, the question is whether a collection $S$ of $k$ sets from $V_1, \dots, V_m$, such that $\bigcup S$ has rank $qk$ in $M$.
The problems \textsc{$q$-Matroid Parity} and \textsc{$q$-Matroid Intersection} generalize \textsc{$q$-Set Packing} (when $M$ is the uniform matroid) and \textsc{$q$-Dimensional Matching} (when $M_i$ is the partition matroid over $V_i$), respectively.

We survey the known results for these problems.
The fastest known algorithm for \textsc{Set Cover}
in terms of $n=|V|$ is $O^*(2^n)$, which can be achieved either via
classical dynamic programming or by inclusion-exclusion. It is a major
open problem whether this can be improved; Cygan et al.~\cite{CyganDLMNOPSW16seth}
propose the \emph{Set Cover Conjecture} (SeCoCo) that effectively
conjectures that this is not possible, analogous to the more commonly
used strong exponential-time hypothesis (SETH).
More precisely, SeCoCo states that for every $\varepsilon > 0$
there exist $d \in \N$ such that \textsc{Set Cover}
on $n$ elements where all sets of size at most $d$
cannot be solved in $O^*(2^{(1-\varepsilon)n})$ time~\cite{CyganDLMNOPSW16seth}.
The currently known fastest algorithms for \textsc{$q$-Dimensional Matching} and \textsc{$q$-Set Packing} are by Bj\"orklund et al.~\cite{BjorklundHKK17narrow}.
(The reader is referred to \cite{BjorklundHKK17narrow} for a series of previous improvements on these problems e.g., \cite{Chen09,Koutis08ICALP,KoutisW09limits}.)
Their running time bounds are $O^*(2^{(q-2)k})$ and $O^*(2^{(q-\varepsilon_q)k})$, respectively, where $\varepsilon_q < 2$ is a constant depending on $q$, tending to zero as $q \to \infty$.
The fastest known algorithms for \textsc{$q$-Matroid Intersection} and \textsc{$q$-Matroid Parity} run in time $O^*(4^{qk})$ \cite{BrandP21}.
Very recently, Brand et al.~\cite{BrandKS23} gave an $O^*(4^k)$-time algorithm for $q \le 4$.

\subsection{Rank \texorpdfstring{$k$}{k} Set Cover and Rank \texorpdfstring{$k$}{k} Set Packing.}
We start with \textsc{Rank $k$ Set Cover}, reiterating Theorem~\ref{ithm:setcover}.
We will assume that the solution size is exactly $t$.
We will use the polynomial-space sieving algorithm (Theorem~\ref{theorem:sieve-basis}) if the underlying field has characteristic 2, and the sieving algorithm for strongly monotone circuits (Theorem~\ref{lemma:simple-matroid-sieving}) otherwise.
To that end, we construct a polynomial as follows.
Let $X = \{ x_{v, E} \mid v \in V, E \in \cE \}$ and $Y = \{ y_{i,E} \mid i \in [t], E \in \cE \}$ be a set of variables.
For \textsc{Rank $k$ Set Cover}, we define
\begin{align*}
  P(X, Y)
  &= \prod_{i \in [t]} \sum_{E \in \cE} y_{i,E} \prod_{v \in E} (1 + x_{v,E}) \\
  &= \sum_{f \colon [t] \to \cE} \prod_{i \in [t]} y_{i,f(i)} \prod_{v \in f(i)} (1 + x_{v,f(i)})
  = \sum_{f \colon [t] \to \cE} \sum_{\substack{ E_1, \ldots, E_t \\ E_i \subseteq f(i), i \in [t]}} \left( \prod_{i \in [t]} y_{i,f(i)} \right) \left( \prod_{i \in [t]} \prod_{v \in E_i} x_{v,f(i)} \right)
\end{align*}
For \textsc{Rank $k$ Set Packing}, we tweak the polynomial slightly:
\begin{align*}
  P(X, Y) = \prod_{i \in [t]} \sum_{E \in \cE} y_{i,E} \prod_{v \in E} x_{v,E}
  = \sum_{f \colon [t] \to \cE} \prod_{i \in [t]} y_{i,f(i)} \prod_{v \in f(i)} x_{v,f(i)}.
\end{align*}
Note that the function $f \colon [t] \to \cE$ plays the role of choosing $t$ sets from $\cE$.
For every $f \colon [t] \to \cE$, there is a distinct monomial and thus no further algebraic cancellation occurs.
Let $A \in \F^{k \times V}$ be the linear representation of $M$.
We may assume that $A$ has exactly $k$ rows by truncating $M$.
Let $P'(X)$ be the result of substituting every variable $y_{i,E}$ with a uniformly chosen random element from $\F$.
By the Schwartz-Zippel lemma, if there exists a collection $S \subseteq \mathcal{E}$ of $t$ sets such that $U = \bigcup S$, then with high probability the polynomial $P'(X)$ contains a monomial of the form $\prod_{v \in U} x_{v,\iota(v)}$, where $\iota\colon U\to S$ is any function satisfying $\iota(v) \in S$ for each $v\in U$.
Mapping every variable $x_{v,E}$ to the column vector $A[\cdot, v]$, we use the sieving algorithm.
If $\F$ has characteristic 2, then Theorem~\ref{theorem:sieve-basis} gives an $O^*(2^k)$-time algorithm.
Otherwise, we use Theorem~\ref{lemma:simple-matroid-sieving}.
Note that $P'(X)$ can be realized by a strongly monotone circuit.
Thus, we have:

\begin{theorem}[Restatement of Theorem~\ref{ithm:setcover}] \label{theoremm:setcover}
  \textsc{Rank~$k$ Set Cover} for
  matroids represented over a field $\F$ can be solved
  in $O^*(2^k)$ time and polynomial space if $\F$ has characteristic 2 and in $O^*(2^{\omega k / 2})$ time and $O^*(2^k)$ space in general. 
\end{theorem}

\begin{theorem} \label{theoremm:setpacking}
  \textsc{Rank~$k$ Set Packing} for
  matroids represented over a field $\F$ can be solved
  in $O^*(2^k)$ time and polynomial space if $\F$ has characteristic 2 and in $O^*(2^{\omega k / 2})$ time and $O^*(2^k)$ space in general. 
\end{theorem}

\subsection{\texorpdfstring{$q$}{q}-Matroid Parity and \texorpdfstring{$q$}{q}-Matroid Intersection}
Next, we discuss \textsc{$q$-Matroid Intersection} and \textsc{$q$-Matroid Parity}.
Recall that the problems are defined as follows:
In the \textsc{$q$-Matroid Intersection} problem, we are given $q$ matroids $M_1, \dots, M_q$ of rank $k$ over the same ground set $V$. The task is to decide whether there exists a subset $S \subseteq V$ that forms a basis in every matroid $M_i$.
In the \textsc{$q$-Matroid Parity} problem, we are given a matroid $M=(V,\cI)$ of rank $qk$ together with a partition $\cE=\{E_1,\dots,E_m\}$ of $V$ into disjoint sets, each of size $q$.
The problem is to decide whether there exists a collection $S\subseteq \cE$ of $k$ sets such that the union $\bigcup S$ has rank $qk$ in $M$.

We start with the \textsc{$q$-Matroid Parity} problem.
Let $X = \{ x_{E} \mid E \in \cE \}$ be a set of variables.
We define a polynomial:
\begin{align*}
  P(X) = \prod_{E \in \cE} \left( 1 +x_E \right)
  = \sum_{j \in [|\cE|]} \sum_{S \in \binom{\cE}{j}} \prod_{E \in S} x_E.
\end{align*}
We apply the sieving algorithm by associating every $x_E$ with $q$ columns $A[\cdot, E]$.
To sieve over general fields, observe that the polynomial $P(X)$ can be computed using a 1-skew strongly monotone circuit.
Using the basis sieving algorithm (Corollary~\ref{cor:basis} and Theorem~\ref{lemma:simple-matroid-sieving}), we obtain:

\begin{theorem} \label{theorem:general-matroid-parity}
  \textsc{$q$-Matroid Parity} for matroids represented over a field $\F$ can be solved in $O^*(2^{qk})$ time (and polynomial space if $\F$ has characteristic 2).
\end{theorem}

Since \textsc{$q$-Matroid Intersection} is a special case of \textsc{$q$-Matroid Parity}, we also obtain:

\begin{corollary} \label{cor:general-matroid-parity}
   \textsc{$q$-Matroid Intersection} for matroids represented over a field $\F$ can be solved in $O^*(2^{qk})$ time (and polynomial space if $\F$ has characteristic 2).
\end{corollary}

We obtain a greater speedup for \textsc{$q$-Matroid Intersection} by using the Cauchy-Binet formula. 
Suppose that $A_i \in \F^{k \times V}$ represents the matroid $M_i$.
Let $X=\{x_v \mid v \in V\}$ be a set of variables and let $A_1'$ be
the result of scaling every column $v$ of $A_1$ by $x_v$.
By the Cauchy-Binet formula,
\[
  P(X) := \det (A_1' A_2^T) = \sum_{B \in \binom{V}{k}} \det
  A_1[\cdot,B] \det A_2[\cdot,B] \prod_{v \in B} x_v.
\]
Thus $P(X)$ enumerates monomials $\prod_{v \in B} x_v$ for
common bases $B$ of $A_1$ and $A_2$, and we only have to sieve for
terms that in addition are bases of the remaining $q-2$ matroids. 
We construct a matroid of rank $(q-2)k$ by taking the direct sum of these $q-2$ matroids. Then, by applying Corollary~\ref{cor:basis} (with each variable $x_v$ corresponding to its copies in the direct sum), we obtain the following.

\begin{theorem}[Restatement of Theorem~\ref{ithm:matroid-intersection}]
  \textsc{$q$-Matroid Intersection} for linear matroids represented
  over a common field $\F$ of characteristic 2 can be solved in
  randomized time $O^*(2^{(q-2)k})$ and polynomial space.
\end{theorem}

For fields of characteristic other than 2, this does not represent a speedup over Corollary~\ref{cor:general-matroid-parity}
since the circuit computing $\det A_1' A_2^T$ is not strongly monotone. 
However, we do obtain a speedup for general $\F$ for the special case $q=3$.
Observe that every entry in $A_1' A_2^T$ has polynomial degree at most 1.
It is known that the determinant of a symbolic matrix can be computed with a skew circuit \cite{MahajanV97}.
Thus, there is a 1-skew circuit computing $\det (A_1' A_2^T)$.
Using the sieving algorithm of Theorem~\ref{lemma:extensor-matroid-sieving} for general arithmetic circuits, we obtain:

\begin{theorem} \label{thm:general-intersection}
  \textsc{$q$-Matroid Intersection} for linear matroids can be solved in $O^*(4^{(q - 2)k})$ time.
  In particular, the bound is $O^*(4^k)$ for $q = 3$.
\end{theorem}

It is an interesting open question whether \textsc{$q$-Matroid Parity} can be solved in $O^*(2^{(q - \varepsilon)k})$ for $\varepsilon > 0$ when $q \ge 3$ is constant.
Note that an enumerating polynomial for 2-matroid parity (let us call it \emph{matroid matching} for clarity) can be efficiently evaluated using the linear representation of Lov\'asz~\cite{Lovasz79}:
Suppose that $A$ represents a matroid $M = (V, \cI)$ with $V$ partitioned into pairs $P_i = \{ v_i, v_i' \}$.
If $x_i$ is a variable representing the pair $P_i$, then the Pfaffian $\Pf B$, where
\begin{align*}
   B = \sum_i x_i (A[\cdot, v_i] A^T[v_i', \cdot] - A[\cdot, v_i'] A^T[v_i, \cdot]),
\end{align*}
enumerates all matroid matching terms.
Lov\'asz \cite{Lovasz79} only showed that the rank of $B$ equals twice the maximum matroid matching size, but $\Pf B$ indeed enumerates all matroid matching terms.
We refer to the textbook of Murota~\cite[Section 7.3.4]{murota1999matrices} for this fact (the exposition concerns an alternative equivalent formulation of matroid matching proposed by Geelen and Iwata~\cite{GeelenI05}).
The trick employed by Bj\"{o}rklund et al.~\cite{BjorklundHKK17narrow} to speed up \textsc{$q$-Set Packing} of ``reducing'' (via colour-coding type arguments) to \textsc{$q$-Dimensional Matching}, however, seemingly does not work for the matroid analogs.
The simple idea of having the variable $x_i$ encode $q - 2$ columns in the matroid matching enumerating polynomial fails because the space spanned by vectors in the matroid matching is not necessarily orthogonal to the other of $q - 2$ columns.

\subsection{Odd Coverage}
Finally, let us discuss another corollary of Theorem~\ref{theorem:matroid-sieving} on a variant of \textsc{Set Cover}, called \textsc{Odd Coverage}.
The input is a set family $\cE$ over $V$ and integers $t, p$.
The question is whether there is a subcollection $S \subseteq \cE$ with $|S| = t$ such that there are at least $p$ elements $v \in V$ with $|\{ E \in S \mid v \in E \}| \bmod 2 = 1$ (i.e., $v$ is covered an odd number of times).
Over a set of variables $X = \{ x_v \mid v \in V \}$ and $Y = \{ y_E \mid E \in \cE \}$, define
\begin{align*}
  P(X, Y, z) = \prod_{E \in \cE} \left(1 + z y_E \prod_{v \in E} x_v \right).
\end{align*}
The coefficient of $z^t$ then enumerates the subcollections of size $t$.
Note that there is a solution if and only if there is a monomial (over $X$) whose odd support set is size at least $p$.
Thus, the odd sieving algorithm implies:

\begin{theorem}
  \textsc{Odd Coverage} can be solved in $O^*(2^p)$ time and polynomial space.
\end{theorem}

An $O^*(2^p)$-time (and exponential-space) algorithm for a special case is known, given by Saurabh and Zehavi~\cite{SaurabhZ18}.
They studied the following problem: given a graph $G = (V, E)$ and integers $t, p$, is there a set $S$ of exactly $t$ vertices such that there are at least $p$ edges with one endpoint in $S$ and the other in $V \setminus S$?
Note that this is a special case of \textsc{Odd Coverage} in which every element occurs in two sets.

\section{Balanced Solution and Diverse Collection}

As noted, given an efficient enumerating polynomial $P(X)$ for a
category of objects, and given a representable matroid $M$ over $X$,
we can use our methods out-of-the-box to sieve for objects in the
collection that are independent or spanning in $M$.
In this section, we survey two applications.
The first concerns the problem of finding a \emph{balance-fair} solution.
A balanced-fairness is, in a way, a stronger notion of colourfulness; every colour should appear not only once, but also almost equally frequently.
We note that with an efficient enumerating polynomial at hand, our sieving algorithm can find a balanced-fair solution. 
The second addresses another problem category, of finding a \emph{diverse}
collection of objects, with prescribed pairwise minimum distances.
Utilizing the odd sieving method (Theorem~\ref{theorem:matroid-sieving}), we show a general way to find a diverse collection.

\subsection{Balance-fair X paradigm}
\label{sec:balanced}

There is a recent trend in pursing fairness especially in artificial intelligence applications (see e.g., the work of Chierichetti et al.~\cite{Chierichetti0LV17}).
There are many notions of fairness known in the literature.
Here, we consider the problem of finding a \emph{balanced} solution.
We assume that every object is assigned a colour from a set $C$.
For $\alpha \le \beta \in \N$, a set $S$ of objects is said to be \emph{$(\alpha, \beta)$-balanced} if $\alpha \le |S_c| \le \beta$ for every colour $c \in C$, where $S_c \subseteq S$ denotes the objects in $S$ with colour $c$.
The problem of finding a balanced solution has been studied in the context of \textsc{Matroid Intersection}~\cite{Chierichetti0LV19}, \textsc{$k$-Matching}~\cite{BandyapadhyayFIS23}, and \textsc{$k$-Path}~\cite{BentertKN23}.
We define a general problem called \textsc{Balanced Solution} as follows.
The input is a set $E$, a collection of (possibly exponentially many) subsets $\cF \subseteq 2^E$ of $E$, a set of colours $C$, a colouring $\chi \colon E \to C$, and integers $k, \alpha, \beta$.
The question is whether there is a set $S \in \cF$ of size $k$ such that $\alpha \le |S \cap \chi^{-1}(c)| \le \beta$ for all $c \in C$.
We observe that the basis sieving (Theorem~\ref{theorem:sieve-basis}) solves this problem in time $O^*(2^k)$, if an enumerating polynomial for $\mathcal{F}$ can be evaluated in polynomial time over a field of characteristic~2.
To set up the matroid constraint, we use the observation of Bentert et al.~\cite{BentertKN23} that there is a linear matroid $M$ of rank $k$ with coloured objects as its ground set such that a set of $k$ objects is $(\alpha, \beta)$-balanced if and only if it is a basis for $M$.
In particular, a linear representation of $M$ over a field of characteristic 2 can be constructed in randomized polynomial time.
We thus obtain from the definitions:

\begin{theorem}
  \textsc{Balanced Solution} can be solved in $O^*(2^k)$ time if there is an enumerating polynomial for $\cF$ that can be evaluated in polynomial time over a field of characteristic~2.
\end{theorem}

In particular, this implies $O^*(2^k)$-time algorithms for balanced-fair variants of \textsc{Matroid Intersection}, \textsc{$k$-Matching}, and \textsc{$k$-Path} (see Section~\ref{sec:enum-poly} for the enumerating polynomials).
In particular, for \textsc{$k$-Path} we use the enumerating polynomial for $k$-walks, and give all copies $x_{v,i}$ for a vertex $v$ the same label in the matroid~$M$,
thereby ensuring that any surviving monomial represents a path.
This is an improvement over the existing algorithms, all of which run in $O^*(2^{ck})$ time for some $c > 1$.

\subsection{Diverse X paradigm} \label{sec:diverse}

In the so-called ``diverse X paradigm'' (X being the placeholder for an optimization problem), we seek -- rather than a single solution -- a diverse collection of solutions, where the diversity is measured in terms of the Hamming distance, i.e., the size of the symmetric difference.
Recently, there is an increasing number of publications studying the problem of finding diverse solutions from the parameterized complexity perspective \cite{BasteFJMOPR22,BasteJMPR19,FominGJPS24,FominGPPS24,HanakaKKO21}.

The \textsc{Diverse Collection} problem is defined as follows.
For a set $E$, let $\cF_i$ be a collection of (potentially exponentially many) subsets of $E$ for each $i \in [k]$.
Given $d_{i,j} \in \N$ for $i < j \in [k]$, the problem asks to determine the existence of subsets $S_i \in \cF_i$ for $i \in [k]$ such that $|S_i \Delta S_j| \ge d_{i,j}$ for each $i < j \in [k]$.
Here, $S_i \Delta S_j$ denotes the symmetric difference $(S_i \setminus S_j) \cup (S_j \setminus S_i)$.
We show that if all collections $\cF_i$ admit enumerating polynomials $P_i(X)$ that can be efficiently evaluated, then \textsc{Diverse Collection} can be solved in $O^*(2^D)$ time, where $D = \sum_{i < j \in [k]} d_{i,j}$.

Let $X' = \{ x_{e}^{\{ i, j \}} \mid i, j \in [k], e \in E \}$ and $Y = \{ y_{i,e} \mid i \in [k], e \in E \}$ be variables.
We define 
\begin{align*}
  P(X', Y) = \prod_{i \in [k]} P_i'(X', Y)
\end{align*}
where $P_i'(X', Y)$ is the result of plugging $x_e = y_{i,e} \prod_{j \in [k] \setminus \{ i \}} x_{e}^{\{ i, j \}}$ in the enumerating polynomial $P_i(X) = \sum_{S_i \in \cF_i} c(i,S_i) \prod_{e \in S_i} x_e$ for coefficients $c(i, S_i) \in \F$.
The variables $x_{e}^{\{ i, j \}}$ will play a key role in ensuring that $|S_i \Delta S_j| \ge d_{i,j}$.
Let us expand $P(X', Y)$ into a sum of monomials:
\begin{align*}
  P(X', Y)
  &= \sum_{\substack{S_1, \dots, S_k \\ S_i \in \mathcal{F}_i}} \left( \prod_{i \in [k]} c(i, S_i) \cdot \prod_{i \in [k], e \in S_i} y_{i,e} \cdot \prod_{i \in [k], e \in S_i} \prod_{j \in [k] \setminus \{ i \}} x_e^{\{ i, j \}} \right)
\end{align*}
With $S_i \in \cF_i$ fixed for each $i \in [k]$, we have
\begin{align*}
  \prod_{i \in [k], e \in S_i} \prod_{j \in [k] \setminus \{ i \}} x_e^{\{ i, j \}}
  = \prod_{\substack{i < j \in [k]}} \left( \prod_{e \in S_i} x_{e}^{\{i,j\}} \right) \left( \prod_{e \in S_j} x_{e}^{\{i,j\}} \right)
  = \prod_{i < j \in [k]} \left( \prod_{e \in S_i \Delta S_j} x_e^{\{i,j\}} \right) \left( \prod_{e \in S_i \cap S_j} (x_e^{\{i,j\}})^2 \right).
\end{align*}
We therefore have
\begin{align*}
  P(X', Y)
  &= \sum_{\substack{S_1, \dots, S_k \\ S_i \in \mathcal{F}_i}} \left( \prod_{i \in [k]} c(i, S_i) \cdot \prod_{i \in [k], e \in S_i} y_{i,e} \cdot  \prod_{i < j \in [k]} \left( \prod_{e \in S_i \Delta S_j} x_e^{\{i,j\}} \right) \left( \prod_{e \in S_i \cap S_j} (x_e^{\{i,j\}})^2 \right) \right).
\end{align*}
For every collection of $k$-tuples $(S_1, \dots, S_k)$ with $S_i \in \cF_i$, there is a distinct monomial in $P(X', Y)$.
We use the odd sieving algorithm of Theorem~\ref{theorem:matroid-sieving}.
More precisely, we add constraints such that for every $\{ i, j \} \subseteq [k]$, there are at least $d_{i,j}$ variables $x_{e}^{\{i,j\}}$ in the odd support set.
This ensures that each pairwise Hamming distance is at least $d_{i,j}$.
Note that these constraints can be realized using a partition matroid of rank $D$, where for each $i < j \in [k]$ there is a block $\{ x_e^{{i,j}} \mid e \in E \}$ with capacity $d_{i,j}$.
Thus, we obtain:

\begin{theorem} \label{theorem:diverse-collection}
  \textsc{Diverse Collection} can be solved in $O^*(2^{D})$ time if all collections $\cF_i$ admit enumerating polynomials that can be evaluated in polynomial time over a field of characteristic~2.
\end{theorem}

\begin{remark}
Our approach can be adapted to solve the weighted variant considered by Fomin et al.~\cite{FominGPPS24}.
For the weighted variant, every element $e$ has a positive weight $w_e \in \N$, and we require $S_i$ and $S_j$ to have $\sum_{e \in S_i \Delta S_j} w_e \ge d_{i, j}$, rather than $|S_i \Delta S_j| \ge d_{i, j}$.
To deal with weights, simply replace each variable $x_e^{\{i,j\}}$ with the product of $w_e$ variables $x_{e,1}^{\{i,j\}} x_{e,2}^{\{i,j\}} \cdots x_{e,w_e}^{\{i,j\}}$.
\end{remark}

\begin{remark}
A variant of \textsc{Diverse Collection} where we wish to maximise the sum of all pairwise Hamming distances (that is, $\sum_{i < j \in [k]} |S_i \Delta S_j| \ge D_+$) is also studied in the literature \cite{BasteFJMOPR22,BasteJMPR19,HanakaKKLO22,HanakaKKO21}.
A similar approach yields an FPT algorithm with running time $O^*(2^{D_+})$.
Using the same polynomial $P(X', Y)$, we require that there should be at least $D_+$ variables $x_e^{\{ i, j \}}$ in the odd support set.
Obviously, this can be done using a uniform matroid of rank $D_+$.
Thus, the sieving algorithm of Theorem~\ref{theorem:matroid-sieving} gives an $O^*(2^{D_+})$-time algorithm.
\end{remark}

We discuss several corollaries of Theorem~\ref{theorem:diverse-collection}.
First, we consider \textsc{Diverse Perfect Matchings}:
we are given an undirected graph $G$, an integer $k$, and $\binom{k}{2}$ integers $d_{i,j}$ for $i < j \in [k]$, and we want to find $k$ perfect matchings $M_1, \dots, M_k$ with $|M_i \Delta M_j| \ge d_{i,j}$ for every $i < j \in [k]$.
Let $d = d_{1,2}$ and $D = \sum_{i < j \in [k]} d_{i,j}$.
This problem is NP-hard even for $k = 2$ \cite{Holyer81a}.
Fomin et al.~\cite{FominGJPS24} gave an $O^*(4^{d})$-time algorithm for the special case $k = 2$.
Later, Fomin et al.~\cite{FominGPPS24} proved that \textsc{Diverse Perfect Matchings} is FPT for the case $d_{i,j} = d$ for all $i < j \in [k]$, giving an algorithm running in time $O^*(2^{2^{O(dk)}})$.
Since the Pfaffian is an enumerating polynomial for perfect matchings, we obtain:
\begin{corollary} \label{corollary:diverse-pm}
  \textsc{Diverse Perfect Matchings} can be solved in $O^*(2^{D})$ time.
\end{corollary}

Our approach also works for diverse matroid problems \textsc{Diverse Bases} and \textsc{Diverse Common Independent Sets}, which were introduced by Fomin et al.~\cite{FominGPPS24}.
In \textsc{Diverse Bases}, we are given a matroid $M$ and $k, d_{i,j} \in \N$ for $i < j \in [k]$, and the question is whether $M$ has bases $B_1, \dots, B_k$ such that $|B_i \Delta B_j| \ge d_{i, j}$ for all $i < j \in [k]$.
In \textsc{Diverse Common Independent Sets}, we are given two matroids and $k, d_{i,j} \in \N$ for $i < j \in [k]$, and the question is whether a collection of sets $I_1, \dots, I_k$ that are independent in both matroids such that $|I_i \Delta I_j| \ge d_{i,j}$ for all $i < j \in [k]$.
The previous known algorithms of Fomin et al.~\cite{FominGPPS24} solve \textsc{Diverse Bases} and \textsc{Diverse Common Independent Sets} in time $O^*(2^{O(k^2 d \log kd)})$ and $O^*(2^{O(k^3 d^2 \log kd)})$, respectively when $d_{i,j} = d$.

\begin{corollary}
  \textsc{Diverse Bases} on linear matroids represented over fields of characteristic~2 can be solved in $O^*(2^D)$ time.
\end{corollary}
\begin{corollary}
  \textsc{Diverse Common Independent Sets} on linear matroids represented over fields of characteristic~2 can be solved in $O^*(2^D)$ time.
\end{corollary}

Theorem~\ref{theorem:diverse-collection} also has an implication for the \textsc{$k$-Distinct Branching} problem.
Its input is a directed graph~$G$, two vertices $s$ and $t$, and an integer $k$.
The problem asks whether $G$ admits an out-branching $(V, B_s^+)$ rooted at $s$ and in-branching $(V, B_t^-)$ rooted at $t$ such that $|B_s^+ \Delta B_t^-| \ge k$.
The NP-hardness is even for $s = t$ and $k = 2n - 2$ \cite{Bang-Jensen91}.
Since Bang-Jensen and Yeo~\cite{Bang-JensenY08} asked whether \textsc{$k$-Distinct Branching} is FPT for $s = t$, this problem has been studied in parameterized complexity.
We briefly survey the history here.
Bang-Jensen et al.~\cite{Bang-JensenSS16} gave an FPT algorithm for strongly connected graphs.
Later, Gutin et al.~\cite{GutinRW18} showed that \textsc{$k$-Distinct Branching} on arbitrary directed graphs can be solved in $O^*(2^{O(k^2 \log^2 k)})$ time for $s = t$.
Very recently, Bang-Jensen et al.~\cite{Bang-JensenK021} designed an $O^*(2^{O(k \log k)})$-time algorithm.
They asked whether \textsc{$k$-Distinct Branchings} can be solved in $O^*(2^{O(k)})$ time.
As a corollary of Theorem~\ref{theorem:diverse-collection}, we answer this question in the affirmative.
Recall that the determinant of the symbolic Laplacian matrix yields an enumerating polynomial for out-branchings and for in-branchings by reversing arcs (see \cite{BjorklundKK17directed,gessel95algebraic}).
Thus, Theorem~\ref{theorem:diverse-collection} implies:

\begin{corollary}
  \textsc{$k$-Distinct Branchings} can be solved in $O^*(2^k)$ time.
\end{corollary}

\section{Path, cycle and linkage problems}
\label{sec:paths-linkages}

One of the main application areas of algebraic algorithms in
parameterized complexity is for path and cycle problems. Indeed, one
of the earliest examples of an algebraic FPT algorithm was for
\textsc{$k$-Path}, finding a path on $k$ vertices in a possibly
directed graph, ultimately improved to time $O^*(2^k)$~\cite{Koutis08ICALP,Williams09IPL,KoutisW09limits}.
Another breakthrough result in the area is Björklund's algorithm for
\textsc{Hamiltonicity}, finding a Hamiltonian path in an undirected
graph, in time $O^*(1.66^n)$~\cite{Bjorklund14detsum},
and more generally solving \textsc{$k$-Path} in undirected graphs in
time $O^*(1.66^k)$~\cite{BjorklundHKK17narrow}.
In fact, even the apparently simple question of \textsc{$k$-Path},
\textsc{$k$-Cycle} and \textsc{Hamiltonicity} problems remains a
highly active area of research.
This is particularly true in \emph{directed} graphs; 
however, in this section we restrict ourselves to undirected graphs. 
We also restrict ourselves solely to matroids represented over fields
of characteristic 2, since we need the power of the odd support
sieving method (Theorem~\ref{theorem:matroid-sieving}). 

Another, subtly different problem is to find a cycle of length
\emph{at least} $k$, which we refer to as the \textsc{Long Cycle}
problem. Unlike the corresponding ``\textsc{Long Path}'' problem,
being able to find a $k$-cycle in time $O^*(c^k)$ does not guarantee 
being able to solve \textsc{Long Cycle} in the same time.
On directed graphs, the first algorithm for \textsc{Long Cycle}
with running time $O^*(2^{O(k)})$ was given by Fomin et al.~\cite{FominLPS16JACM}
using \emph{representative families} 
(cf.~Sections~\ref{sec:intro-noDP} and~\ref{sec:noDP});
the current record is $O^*(4^k)$ by Zehavi~\cite{Zehavi16ldc}.
For undirected graphs, the currently fastest algorithm for
\textsc{Long Cycle} is by reduction to the more general
\textsc{Long $(s,t)$-Path} problem. Note that, again unlike unrooted
\textsc{Long Path}, asking for an $(s,t)$-path of length at least $k$
is a sensible question that does not trivially reduce to rooted
\textsc{$k$-Path} (i.e., to finding an $(s,t)$-path of length exactly $k$).
In turn, the fastest algorithm for \textsc{Long $(s,t)$-Path}
is by Fomin et al.~\cite{FominGKSS24} in time $O^*(2^k)$; see below.

In a different direction, in the problem \textsc{$T$-Cycle}
(a.k.a.\ \textsc{$K$-Cycle}), the input is an undirected graph $G$ and
a set of vertices $T \subseteq V(G)$, and the question is whether
there is a simple cycle in $G$ that visits every vertex in~$T$.
As mentioned in the introduction, this problem was known to be FPT
using an algorithm working over heavy graph structural methods~\cite{Kawarabayashi08},
and it was a major surprise when Björklund, Husfeldt and Taslaman~\cite{BjorklundHT12soda}
showed an $O^*(2^{|T|})$-time algorithm based on polynomial cancellations. 
Specifically, they defined a polynomial, roughly corresponding to
walks without U-turns, and showed in an intricate argument 
that an $O^*(2^{|T|})$-time sieving step over this polynomial
tests for $T$-cycles in $G$. 
Wahlström~\cite{Wahlstrom13STACS} adapted Björklund's
\emph{determinant sums} method~\cite{Bjorklund14detsum} to the
\textsc{$T$-Cycle} problem and thereby showed that it even allows for
a \emph{polynomial compression}, i.e., a reduction in polynomial time
to an object of size $|T|^{O(1)}$ from which the existence of a
$T$-cycle can be decided. 

Recently, Fomin et al.~\cite{FominGKSS24} considered problems
pushing the envelope on the method of Björklund, Husfeldt and
Taslaman~\cite{BjorklundHT12soda}, showing more involved 
cancellation-based algorithms for more general path and cycle
problems, and also extending the scope to \emph{linkages}.
Let $G=(V,E)$ be a graph, and $S, T \subseteq V$ be vertex sets.
An \emph{$(S,T)$-linkage} in $G$ is a set $\cP$ of pairwise vertex-disjoint
$(S,T)$-paths. The \emph{order} of the linkage is $p=|\cP|$. 
We say the linkage is \emph{perfect} if $|\cP|=|S|=|T|$.
Let the \textsc{Colourful $(S,T)$-Linkage} problem refer to the
following question. Let $G=(V,E)$, $S, T \subseteq V$, and an integer
$k$ be given. Furthermore, let $c \colon V \to [n]$ be a not
necessarily proper vertex colouring, also given as input.
Then the question is: Does $G$ contain a perfect $(S,T)$-linkage using
vertices of at least $k$ colours? (More generally, one may ask of an
$(S,T)$-linkage of order $p$, but this is essentially equivalent
as we can create new sets $S'$ and $T'$ of $p$ vertices each,
and connect them to $S$ and $T$.)
Fomin et al.\ showed, using complex polynomial
cancellation arguments, that \textsc{Colourful $(S,T)$-Linkage}
can be solved in time $O^*(2^{k+p})$ where $p=|S|=|T|$~\cite{FominGKSS24}.
We show the following improvement.

\begin{theorem}
  \label{thm:k-st-linkage}
  \textsc{Colourful $(S,T)$-Linkage} for undirected graphs can be
  solved in randomized time $O^*(2^k)$ and polynomial space. 
\end{theorem}

As Fomin et al.\ note, even the problem \textsc{Colourful $(s,t)$-Path}
(being the case where $|S|=|T|=1$) has a multitude of applications.
Among others, their result implies solving \textsc{Long $(s,t)$-Path} and
\textsc{Long Cycle} in time $O^*(2^k)$ -- i.e., in an undirected
graph, find an $(s,t)$-path, respectively a cycle, of length \emph{at least} 
$k$ in time $O^*(2^k)$, and \textsc{$T$-Cycle} in time $O^*(2^{|T|})$.
All of these improve on or match the previous state of the art.

Fomin et al.\ also consider the more general setting of \emph{frameworks},
as defined by Lov\'asz (and previously known as
\emph{pregeometric graphs} or \emph{matroid graphs})~\cite{FominGKSS24,Lovasz1977,LovaszGeomBook2019}.
Let $G=(V,E)$ be an undirected graph and $M=(V,\cI)$ a matroid over
the vertex set of $G$. Let $S, T \subseteq V$ and let $k$ be an integer. 
Fomin et al.\ show that if $M$ is represented over a finite field of order $q$, 
then an $(S,T)$-linkage of rank at least $k$ in $M$ can be found in
time $O^*(2^{p+O(k^2 \log (k+q))})$~\cite{FominGKSS24}.
We note that if $M$ is represented over a field of characteristic 2,
then we get a significant speedup over their algorithm.

\begin{restatable}{theorem}{linkageframework}
  \label{thm:st-linkage-framework}\RestateRemark
  Given an undirected graph $G=(V,E)$, a matroid $M$ over $V$ represented 
  over a field of characteristic 2, sets $S, T \subseteq V$ and an integer $k$, 
  in randomized time $O^*(2^k)$ and polynomial space we can find a
  perfect $(S,T)$-linkage in $G$ which has rank at least~$k$~in~$M$. 
\end{restatable}

Theorem~\ref{thm:k-st-linkage} follows from
Theorem~\ref{thm:st-linkage-framework} by letting $M$ be a partition matroid.
Concretely, let $M=(V, \cI)$ be the linear matroid with a
representation where each vertex $v \in V$ is associated with the
$c(v)$-th $n$-dimensional unit vector $e_{c(v)}$.
Then a linkage has rank at least $k$ if and only if it visits vertices
of at least $k$ different colours.

We note that \emph{directed} variants of the above results are
excluded, as it is NP-hard to find a directed $(s,t)$-path with even two
distinct colours (see Fomin et al.~\cite{FominGKSS24}).

Finally, Fomin et al.~\cite{FominGKSS24} ask as an open question
whether \textsc{Long $(s,t)$-Path} and \textsc{Long Cycle} can
be solved in $O((2-\varepsilon)^k)$ for any $\varepsilon > 0$.
We show this in the affirmative, giving an algorithm for both problems
that matches the running time for \textsc{Undirected Hamiltonicity}.
Our algorithm is a mild reinterpretation of the \emph{narrow sieves}
algorithm for \textsc{$k$-Path}~\cite{BjorklundHKK17narrow}, rephrased
in terms of an external matroid labelling the vertices of $G$.

\begin{theorem}
  \label{thm:long-st-path}
  \textsc{Long Cycle} and \textsc{Long $(s,t)$-Path} can be solved in
  randomized time $O^*(1.66^k)$ and polynomial space.
\end{theorem}

All the above theorems follow from the same underlying enumerating
polynomial result. At the heart of the \textsc{Hamiltonicity}
algorithm of Björklund~\cite{Bjorklund14detsum},
and the polynomial compression for \textsc{$T$-Cycle} of Wahlström~\cite{Wahlstrom13STACS},
is the result that given a graph $G=(V,E)$ and $s, t \in V$, 
there is a particular \emph{almost} symmetric matrix $A_{st}$
such that $\det A_{st}$ effectively enumerates $(s,t)$-paths, with
some additional ``padding'' terms (see below).
We note that this statement can be generalized to linkages:
Given $G=(V,E)$ and $S, T \subseteq V$ there is a matrix $A_{ST}$
such that $\det A_{ST}$ enumerates padded perfect $(S,T)$-linkages. 
Furthermore, the ``padding'' is compatible with the odd sieving
approach of Theorem~\ref{theorem:matroid-sieving}. We review this
construction next.

\subsection{The linkage-generating determinant}

We now present the algebraic statements that underpin the algorithms
in this section. Like the rest of the paper, these algorithms are based on
algebraic sieving over a suitable enumerating polynomial.
Here, we present this polynomial, in the form of a
\emph{linkage-enumerating determinant}.

\subsubsection{Path enumeration}

We begin with the simpler case of enumerating $(s,t)$-paths. This result
is from Wahlström~\cite{Wahlstrom13STACS}, repeated for completeness,
but is also implicitly present in Björklund~\cite{Bjorklund14detsum}.
We note $(s,t)$-path enumeration is still far from a trivial conclusion, since we
want to enumerate only \emph{paths} without also enumerating $(s,t)$-walks.
Indeed, there is a catch, since otherwise we could solve
\textsc{Hamiltonicity} in polynomial time by searching for an
$(s,t)$-path term of degree $n$. Specifically, we generate
\emph{padded} $(s,t)$-paths, which is a union of $(s,t)$-paths and 2-cycles;
details follow.

Let $G=(V,E)$ be an undirected graph and $s, t \in V$ be vertices. We
show that a modified Tutte matrix of $G$ can be used to produce a
polynomial that effectively enumerates $(s,t)$-paths in $G$. This was
previously used in the polynomial compression for the $T$-cycle
problem~\cite{Wahlstrom13STACS}. 

Let $X=\{x_e \mid e \in E\}$ be a set of edge variables. 
Let $P$ be an $(s,t)$-path in $G$ and define
\[
  X(P) = \prod_{e \in E(P)} x_e.
\]
A \emph{2-cycle term} over $(G,X)$ is a term $x_e^2$ for some $e \in E$;
note that as a polynomial, if $e=uv$ then this term corresponds to the
closed walk $uvu$ in $G$. A \emph{padded $(s,t)$-path term} for an
$(s,t)$-path $P$ is a term
\[
  X(P) \cdot \prod_{e \in M} x_e^2,
\]
where $M$ is a (not necessarily perfect) matching of $G-V(P)$.
We assume by edge subdivision that $st \notin E$.

\begin{lemma} \label{lemma:padded-st-paths}
  Assume (e.g.\ via edge subdivision) that $st \notin E$. 
  There is a matrix $A_{st}$ whose entries are linear polynomials
  over a field of characteristic 2 such
  that $\det A_{s,t}$ enumerates padded $(s,t)$-path terms.
\end{lemma}
\begin{proof}
  Let $A$ be the Tutte matrix of $G$ over a sufficiently large field
  of characteristic 2. Define $A_{st}$ starting from $A$ modified 
  by letting $A_{st}[v,v]=1$ for every
  $v \in V \setminus \{s,t\}$, $A_{st}[s,t]=0$, $A_{st}[t,s]=1$ and $A_{st}[t,v]=0$
  for every $v \in V-s$. We claim that $\det A_{st}$ enumerates padded
  $(s,t)$-path terms as described. This follows from arguments in
  Wahlstr\"om~\cite{Wahlstrom13STACS}. Viewing the rows and columns of
  $A_{st}$ as vertices of $G$, each term of $\det A_{st}$ can be
  viewed as an \emph{oriented cycle cover} of $G$, i.e., a partition
  of $V$ into oriented cycles (which may include cycles of length 1
  where a diagonal entry of $A_{st}$ is used). Due to the
  modifications made to $A_{st}$ above, $t$ has $s$ as its unique out-neighbour 
  in every oriented cycle cover, and for every vertex $v \in V \setminus \{s,t\}$
  the loop term on $v$ can be used in the cycle cover.
  Furthermore, every other edge of the graph is bidirected
  (i.e., symmetric). Hence, if a cycle cover $\cC$ contains any cycle
  $C$ of at least three edges which does not use the arc $ts$,
  then the orientation of $C$ can be reversed to produce a distinct
  oriented cycle cover $\cC'$, corresponding to a distinct term of the
  determinant. Let a \emph{reversible cycle} in
  an oriented cycle cover $\cC$ be a cycle $C$ in $\cC$ which contains
  at least three edges and does not use the arc $ts$. To argue that
  all oriented cycle covers with reversible cycles cancel over a field
  of characteristic 2, we define the following pairing. Fix an
  arbitrary ordering $<$ on $V$. For each oriented cycle cover $\cC$ with
  at least one reversible cycle, select such a cycle $C \in \cC$
  by the earliest incidence of a vertex of $C$ according to $<$,
  and let $\cC'$ be the result of reversing $C$ in $\cC$. Since
  the selection of $C$ is independent of orientation, this map defines
  a pairing between $\cC$ and $\cC'$. By the symmetry of $A_{st}$,
  $\cC$ and $\cC'$ contribute precisely the same term to $\det A_{st}$.
  Generalising the argument, every oriented cycle cover $\cC$ with at
  least one reversible cycle cancels in $\det A_{st}$ in
  characteristic 2. 

  For any oriented cycle cover $\cC$ that is not cancelled by this
  argument, we note that the monomial contributed by $\cC$ to $\det A_{st}$
  is unique (recall that there is one distinct variable $x_e$ for
  every edge $e$ of $G$). Furthermore, let $e=uv$ be an edge such that
  $x_e$ occurs in a monomial $m$ of $\det A_{st}$ corresponding to an
  oriented cycle cover $\cC$. If $e$ occurs in a 2-cycle in $\cC$,
  then $m$ contains $x_e^2$; otherwise $e$ occurs in the $(s,t)$-cycle $C$
  with a passage such as $uvw$, and $x_e$ has degree $1$~in~$m$. 
\end{proof}

Note the slightly subtle interaction between 2-cycle-terms in $\det A_{st}$
and applications of Theorem~\ref{theorem:matroid-sieving}. Since variables
in 2-cycle-terms have even degree, they are not relevant for the matroid basis
sieving of the algorithm, which will therefore effectively sieve
directly over $(s,t)$-paths in $G$. However, the 2-cycle terms prevent us
from (for example) finding a Hamiltonian $(s,t)$-path in polynomial time
by sieving for terms of degree $n$. (However, using a weight-tracing
variable it is possible to find a \emph{shortest} solution, and to
check for the existence of an \emph{odd} or \emph{even} solution.)

\subsubsection{Linkage enumeration}

Through the same principle, we can construct a matrix whose
determinant enumerates perfect $(S,T)$-linkages. Let $G=(V,E)$
be an undirected graph and let $S, T \subseteq V$ where $|S|=|T|$.
As above, define a set of edge variables $X=\{x_e \mid e \in E\}$.
For an $(S,T)$-linkage $\cP$, define
\[
  X(\cP) = \prod_{e \in E(\cP)} x_e.
\]
A \emph{padded $(S,T)$-linkage term} for an $(S,T)$-linkage $\cP$ is
defined as a term 
\[
  X(\cP) \cdot \prod_{e \in M} x_e^2,
\]
where again $M$ is a (not necessarily perfect) matching in $G-V(\cP)$.
Since we are interested in perfect $(S,T)$-linkages we make some
simplifications. If there is a vertex $v \in S \cap T$, simply
delete $v$ from $G$, $S$ and $T$ since the only possible path on $v$
in a perfect $(S,T)$-linkage is the length-0 path~$v$. Hence we assume
$S \cap T = \emptyset$.
We also assume by edge subdivision that $S \cup T$ is an independent set:
Note that no $(S,T)$-linkage needs to use an edge of $G[S]$ or $G[T]$,
and a perfect $(S,T)$-linkage cannot use such an edge.
Furthermore, any edge between $S$ and $T$ can be safely subdivided
without altering the structure of linkages (and, e.g., give the subdividing vertex
the zero vector in the matroid representation).

\begin{lemma} \label{lemma:padded-st-linkages}
  Assume that $S \cap T = \emptyset$ and that $S \cup T$ is an
  independent set. 
  There is a matrix $A_{ST}$ over a field of characteristic 2 such
  that $\det A_{ST}$ enumerates padded perfect $(S, T)$-linkage terms.
\end{lemma}
\begin{proof}
  We follow the proof of Lemma~\ref{lemma:padded-st-paths}, suitably modified.
  Let $A$ be the Tutte matrix of $G$ over a sufficiently large field of characteristic 2.
  Let $S=\{s_1,\ldots,s_p\}$ and $T=\{t_1,\ldots,t_p\}$ with arbitrary ordering.
  We obtain $A_{ST}$ from $A$ by
  letting $A_{ST}[v, v] = 1$, $A_{ST}[s,s]$, $A_{ST}[t,t]=0$, $A_{ST}[v, s] = 0$ and $A_{ST}[t, v] = 0$ for every $s \in S$, 
  $t \in T$ and $v \in V \setminus (S \cup T)$. Furthermore,
  we let $A_{ST}[t_i,s_i]=1$ for $i \in [p]$ and $A_{ST}[t_i,s_j]=0$ otherwise. 
  Essentially, one can think of $A_{ST}$ as the Tutte-like matrix on
  the directed graph $G'$ where every $v \in V \setminus (S \cup T)$
  has a self-loop and the incoming arcs of $S$ and outgoing arcs of $T$ are replaced by
  the induced matching $\{t_is_i \mid i \in [p]\}$. 

  We claim that $\det A_{ST}$ enumerates padded perfect $(S, T)$-linkage terms as described.
  This mimics the argument of Lemma~\ref{lemma:padded-st-paths}: the terms of $\det A_{ST}$ have 
  a one-to-one correspondence with oriented cycle covers of $G'$.
  Note that all other edges of $G'$ not incident with $S \cup T$ are bidirected (i.e., symmetric).
  Hence, if a cycle cover $\cC$ contains any cycle $C$ of length at least 3 disjoint from $S \cup T$,
  then the orientation of $C$ can be reversed to produce a distinct oriented cycle cover $\cC'$.
  We call such a cycle \emph{reversible}.
  To argue that all oriented cycle covers with reversible cycles cancel over a field
  of characteristic 2, we define the following pairing. Fix an
  arbitrary ordering $<$ on $V$. For each oriented cycle cover $\cC$ with
  at least one reversible cycle, select such a cycle $C \in \cC$
  by the earliest incidence of a vertex of $C$ according to $<$,
  and let $\cC'$ be the result of reversing $C$ in $\cC$. Since
  the selection of $C$ is independent of orientation, this map defines
  a pairing between $\cC$ and $\cC'$. By the symmetry of $A_{ST}$,
  $\cC$ and $\cC'$ contribute precisely the same term to $\det A_{ST}$.
  Generalising the argument, every oriented cycle cover $\cC$ with at
  least one reversible cycle cancels in $\det A_{ST}$ in characteristic 2. 

  It remains to show that any oriented cycle cover where every cycle either
  intersects $S \cup T$ or has length at most 2 corresponds to a monomial
  that does not cancel in $\det A_{ST}$, and that such terms are precisely
  padded perfect $(S,T)$-linkages. Let $\cC$ be such an oriented cycle cover.
  We note that the monomial contributed by $\cC$ is a unique ``fingerprint'' of $\cC$
  as an \emph{undirected} cycle cover, since all edges correspond to
  distinct variables. Hence if $\cC$ is cancelled, it has to be
  against a distinct oriented cycle cover $\cC'$ over the same
  underlying set of undirected edges. However, reversing a cycle $C$ of
  length at most 2 yields precisely the same oriented cycle $C$ again,
  and any cycle $C$ intersecting $S \cup T$ is non-reversible.
  The latter follows since the only edges leaving $T$ or entering $S$ in $G'$ are
  directed edges from $T$ to $S$, so reversing $C$ leads to attempting to use a
  non-existing edge from $S$ to $T$.
  Hence any oriented cycle cover $\cC$ that consists of cycles intersecting $S \cup T$, 
  2-cycles and 1-cycles survives cancellation. 

  We next show that the surviving oriented cycle cover terms
  correspond directly to padded perfect $(S,T)$-linkages.
  For any perfect $(S,T)$-linkage $\cP$ padded with a matching $M$, 
  we can construct a non-cancelled oriented cycle cover:
  connect the paths of $\cP$ up using the edges $t_is_i$, $i \in [p]$.
  This defines a vertex-disjoint cycle packing on $V(\cP)$ which covers all of $S \cup T$.
  The number of cycles in the cycle cover depends on how the paths in
  $\cP$ connect their endpoints, but the number of cycles is immaterial
  to the correctness; it is enough that $\cP$ produces a unique
  non-padded term. Together with $M$ and 1-cycles we get an oriented
  cycle cover with no reversible~cycles. 

  Finally, let $\cC$ be a surviving oriented cycle cover,
  let $\cC' \subseteq \cC$ be the set of cycles of length at least 3
  (which includes every cycle on $S \cup T$ by construction),
  and let $\cP$ be the set of paths produced by deleting any arcs
  $ts$, $t \in T$, $s \in S$ from the cycles of $\cC'$.
  We claim that $\cP$ is a perfect $(S,T)$-linkage. 
  Indeed, since $\cC$ is a cycle cover every vertex of $S \cup T$
  occurs in a cycle, and there are no cycles on $S \cup T$ of length 1 or 2.
  Hence $S \cup T$ occur in $\cP$. Furthermore, they clearly occur as
  endpoints, and oriented such that every path in $\cP$ leads from $S$
  to $T$. The cycles of $\cC \setminus \cC'$ correspond to the padding 
  of the term produced. 
  Thus, $\det A_{ST}$ enumerates all padded $(S, T)$-linkage terms.
\end{proof}

\subsection{Rank \texorpdfstring{$k$ $(S,T)$}{k (S,T)}-linkage}

We now formally note Theorem~\ref{thm:st-linkage-framework}
(from which Theorem~\ref{thm:k-st-linkage} follows). 
We begin by bridging the gap between edge variables (from
Lemma~\ref{lemma:padded-st-linkages}) and vertex variables (from the
labels of matroid~$M$). 

\begin{lemma} \label{lemma:linkage-vertex-vars}
  Let $G=(V,E)$ be an undirected graph and $S, T \subseteq V$
  disjoint vertex sets so that $G[S \cup T]$ is edgeless.
  Let $X_V=\{x_v \mid v \in V\}$ and $X_E=\{x_e \mid e \in E\}$. 
  There is a polynomial $P(X_V,X_E)$ which can be evaluated in
  polynomial time over any field of characteristic 2
  such that the following~hold:
  \begin{enumerate}
  \item For every perfect $(S,T)$-linkage $\cP$ there is a monomial in $P(X_V,X_E)$
    whose odd support corresponds to $V(\cP) \cup E(\cP)$
  \item For every monomial $m$ in $P(X_V,X_E)$ with odd support
    $U \subseteq X_V$ and $F \subseteq X_E$, $F$ is the edge set of a
    perfect $(S,T)$-linkage $\cP$ where $U \subseteq V(\cP)$
  \end{enumerate}
\end{lemma}
\begin{proof}
  Let $A_{ST}$ be the matrix constructed in Lemma~\ref{lemma:padded-st-linkages}
  over a new set of variables $X_E'=\{x_e' \mid e \in E\}$.
  Thus, $\det A_{ST}$ enumerates padded perfect $(S.T)$-linkages
  over the variable set $X_E'$. We evaluate $\det A_{ST}$ with an
  assignment where
  \[
    x_{uv}' \gets x_{uv}(x_u+x_v),
  \]
  and define
  \[
    P(X_V,X_E) = \det A_{ST} \cdot \prod_{s \in S} x_s.
  \]
  We claim that this produces monomials precisely as described. 

  First, let $\cP$ be a perfect $(S,T)$-linkage, and let
  $m=X(\cP)=\prod_{e \in E(\cP)} x_e'$, with no padding (i.e., with
  1-cycles on all other vertices). Then $m$ is a monomial produced by
  $\det A_{ST}$. We consider the expansion of $m$ into monomials 
  over $X_V \cup X_E$ resulting from the evaluation.
  We claim that the term
  \[
    \prod_{v \in V(\cP)} x_v \cdot \prod_{e \in E(\cP)} x_e
  \]
  is contributed multilinearly precisely once in $P(X_V,X_E)$. 
  Indeed, for every edge $e=uv \in E(\cP)$, effectively the expansion
  has to select either the contribution $x_u$ or $x_v$. Since 
  $P(X_V,X_E)$ is ``pre-padded'' by $\prod_{s \in S} x_s$, 
  every edge $sv$ leaving $s \in S$ in $\cP$ must be oriented to
  produce $x_v$ instead, or otherwise $x_s$ gets even degree. 
  It follows that the only production that covers every variable
  $x_v$ for $v \in V(\cP)$ is when every edge $uv$, oriented in $\cP$
  from $S$ to $T$ as $(u,v)$, contributes its head variable $x_v$. 
  
  Conversely, assume that $P(X_V,X_E)$ has a monomial $m$ where the
  odd support consists of $U \subseteq X_V$ and $F \subseteq X_E$.
  Then there is a perfect $(S,T)$-linkage $\cP$ such that
  $F=\{x_e \mid e \in E(\cP)\}$. We claim that every variable
  $x_v \in U$ comes from an edge variable $x_e'$ where $x_e \in F$.
  Indeed, the only other production of $x_v$ would be from a padding
  2-cycle $uvu$, which contributes
  \[
    (x_{uv}(x_u+x_v))^2 = x_{uv}^2(x_u^2+x_v^2)
  \]
  since we are working over a field of characteristic 2. Since no
  padding cycles intersect $S \cup T$ and since padding 2-cycles
  evidently do not contribute to the odd support of $m$,
  the conclusion follows. 
\end{proof}

We can now finish the result. We recall the statement of the theorem.

\linkageframework*

\begin{proof} 
  Let $I=(G, M, S, T, k)$ be the input. As noted before Lemma~\ref{lemma:padded-st-linkages}
  we can safely modify $I$ so that $S \cap T = \emptyset$
  and $G[S \cup T]$ is edgeless. 
  Let $X_V=\{x_v \mid v \in V\}$ and $X_E=\{x_e \mid e \in E\}$
  and let $P(X_V,X_E)$ be the polynomial of Lemma~\ref{lemma:linkage-vertex-vars}.
  Let $A_M$ be the representation of $M$ truncated to rank $k$
  and dimension $k \times |V|$. 
  Recall that this can be constructed efficiently, possibly by moving
  to an extension field $\F$ (see Section~\ref{sec:prel-matroids}). 
  Furthermore, we assume $|\F| = \Omega(n)$ for the sake of vanishing
  error probability; again, this can be arranged by moving to an
  extension field. Now, we use Theorem~\ref{theorem:matroid-sieving} to
  sieve over $P(X_V,X_E)$ for a monomial whose odd support in $X_V$
  spans $A_M$. By Lemma~\ref{lemma:linkage-vertex-vars}, if
  there is such a monomial $m$ then the vertex set of the monomial is
  contained in a perfect $(S,T)$-linkage $\cP$, hence there is a
  perfect $(S,T)$-linkage $\cP$ such that $V(\cP)$ spans $A_M$.
  Conversely, if there is a perfect $(S,T)$-linkage $\cP$ such that
  $V(\cP)$ spans $A_M$, then there is also a monomial $m$ of
  $P(X_V,X_E)$ such that the odd support of $m$ spans $A_N$.
  The running time and failure probability comes from
  Theorem~\ref{theorem:matroid-sieving} and $|\F|$. 
\end{proof}

We note a handful of consequences (although the variations on finding 
\emph{shortest} solutions need a little bit more introspection of the proof).

\begin{corollary}
  \label{cor:framework-linkage}
  The following problems can be solved in randomized time $O^*(2^k)$
  and polynomial space. 
  \begin{enumerate}
  \item Finding a perfect $(S,T)$-linkage of total length at least $k$
  \item In a vertex-coloured graph, finding a perfect $(S,T)$-linkage
    which uses at least $k$ different colours
  \item Given a set of terminals $K \subseteq V(G)$ with $|K|=k$, 
    finding a perfect $(S,T)$-linkage that visits every vertex of $K$
  \item Given a matroid $M$ over $V \cup E$ of total rank $k$,
    represented over a field of characteristic 2, 
    finding a perfect $(S,T)$-linkage $\cP$ such that $E(\cP) \cup
    V(\cP)$ is independent in $M$
  \end{enumerate}
  Furthermore, for each of these settings we can find a
  \emph{shortest} solution, or a shortest solution of odd,
  respectively even total length. 
\end{corollary}
\begin{proof}
  The first three applications follow as in the discussion at the
  start of this section, by assigning appropriate vectors to the
  vertices of $G$. To additionally find a shortest, respectively
  shortest odd/even solution, attach a weight-tracing variable $z$ to
  every edge variable $x_e$ and look for a non-zero term in the
  sieving whose degree in $z$ is minimum (respectively minimum subject
  to having odd/even degree). For every perfect
  $(S,T)$-linkage $\cP$, there is a multitude of padded productions,
  but there is a unique monomial $m$ where every vertex $v
  \notin V(\cP)$ is padded using a 1-cycle (such that only edge
  variables corresponding to $E(\cP)$ occur in $m$).
  Finding this minimum degree therefore corresponds to finding the
  shortest length $|\cP|$ for a solution $\cP$. Finally, note that
  padding terms always come in pairs, hence padding $m$ does not
  change its parity in $z$.

  For the final case, first let $R=S \cap T$. We delete $R$ from $S$,
  $T$ and $G$, contract $R$ in $M$, and set $k \gets k-|R|$. We then proceed as follows.
  Attach a weight-tracing variable $z$ to every edge variable $x_e$.
  Guess the value of $k_e=|E(\cP)|$ and note that $|V(\cP)|=|E(\cP)|+|S|$ 
  for every perfect $(S,T)$-linkage; indeed, since $S \cap T = \emptyset$,
  every path contains an edge, hence every path $P \in \cP$
  is a tree with $|V(P)|=|E(P)|+1$. Thus set $k'=2k_e+|S|$,
  restricting the guess for $k_e$ to values such that $k' \leq k$,
  and truncate $M$ to rank $k'$. Use interpolation to extract terms of
  $P(X_V,X_E)$ of degree $k_e$ in $X_E$ and use Theorem~\ref{theorem:sieve-basis}
  to check for multilinear monomials that span the truncation of $M$. 
  As above, if there is a solution $\cP$, with the parameter
  $k_e=|E(\cP)|$ as guessed, then there is also a monomial $m$ in
  $P(X_V,X_E)$ of degree precisely $k_e$ in $X_E$ such that $m$
  is multilinear and its support corresponds precisely to $E(\cP) \cup
  V(\cP)$. Furthermore, $m$ is of degree precisely $k'$,
  hence $m$ precisely spans the truncation of $M$. 
  Any term in $P(X_V,X_E)$ of total degree $k_e$ in $X_E$ that is not of this form
  will either contribute fewer than $k_e+|S|$ variables from $V(\cP)$
  or will fail to be multilinear, and hence will fail to pass
  Theorem~\ref{theorem:sieve-basis}. 
\end{proof}

\subsection{Faster Long \texorpdfstring{$(s,t)$}{(s,t)}-Path and Long Cycle}

Fomin et al.~\cite{FominGKSS24} ask whether \textsc{Long
  $(s,t)$-Path} or \textsc{Long Cycle} -- i.e., the problem of finding,
respectively, an $(s,t)$-path or a cycle of length at least $k$ -- can be
solved in time $O^*((2-\varepsilon)^k)$ for any $\varepsilon > 0$,
given that there is an algorithm solving \textsc{$k$-Cycle} in time
$O^*(1.66^k)$ by Bj\"orklund et al.~\cite{BjorklundHKK17narrow}.
We answer in the affirmative, showing that the algorithm of Bj\"orklund
et al.\ can be modified to solve \textsc{Long $(s,t)$-Path} in time
$O^*(1.66^k)$ by working over the cycle-enumerating determinant
of Lemma~\ref{lemma:padded-st-paths}. 
A corresponding algorithm for \textsc{Long Cycle} follows, by
iterating over all choices of $(s,t)$ as an edge of the cycle. 
We prove the following.

\begin{theorem} \label{theorem:long-path}
  Let $G=(V,E)$ be an undirected graph and $s, t \in V$.
  There is a randomized algorithm that finds an $(s,t)$-path in $G$ of length at
  least $k$ in time $O^*((4(\sqrt{2}-1))^k)=O^*(1.66^k)$ and polynomial space.
\end{theorem}

The result takes the rest of the subsection. Like Bj\"orklund et
al.~\cite{BjorklundHKK17narrow}, the algorithm is based around
randomly partitioning the vertex set of $G$ as $V=V_1 \cup V_2$,
then use algebraic sieving to look for an $(s,t)$-path $P$ that splits
``agreeably'' between $V_1$ and $V_2$, in time better than
$O^*(2^{|P|})$. More specifically, we pick integers $k_x$ and $k_2$
and define a matroid $M=M(k_2,k_x)$ of rank $r=\eta k$ for some $\eta < 3/4$, 
and prove the following: Let $P$ be an $(s,t)$-path that (1) intersects $V_2$ in
precisely $k_2$ vertices, (2) contains precisely $k_x$ edges that
cross between $V_1$ and $V_2$, and (3) has an edge set that spans $M$.
Then $|V(P)| \geq k$. We can then look for such a path
by working over Lemma~\ref{lemma:padded-st-paths}. The details follow.

For a partition $V=V_1 \cup V_2$, let $E_1=E(G[V_1])$,
$E_2=E(G[V_2])$ and $E_X=E \setminus (E_1 \cup E_2)$ so that
$E=E_1 \cup E_X \cup E_2$ partitions $E$. 
Given a partition $V=(V_1,V_2)$ and integers $r_1$ and $r_2$, 
define a matroid $M(r_1,r_2)$ as follows.
Let $M_1$ be a uniform matroid over $E_1$ of rank $r_1$.
Let $M_2$ be a transversal matroid on ground set $F=E_X \cup E_2$,
defined via the bipartite graph $H=(F \cup V_2, E_H)$
where each $e \in F$ is connected to every vertex $v \in e \cap V_2$ in $V_2$. 
Furthermore, truncate $M_2$ to have rank $r_2$. 
That is, a set $S \subseteq E_X \cup E_2$ is independent in $M_2$
if and only if $|S| \leq r_2$ and $S$ has a set of distinct
representatives in $V_2$. Let $M(r_1,r_2)$ over ground set $E$
be the disjoint union of $M_1$ and $M_2$.

\begin{lemma} \label{lemma:long-cycle-rank}
  Let a partition $V=V_1 \cup V_2$ be given with corresponding edge
  partition $E=E_1 \cup E_X \cup E_2$. Furthermore let $s, t \in V$,
  and $M=M(r_1,r_2)$ for some $r_1$, $r_2$.  
  Let $P$ be an $(s,t)$-path and let $C=P+st$ be the corresponding cycle.
  Assume there is a set $F \subseteq E(C)$ such that $E_2 \cap E(C) \subseteq F$
  and $F$ is independent in $M$, and let $|F \cap E_X|=r_x$. 
  Then $|V(C)| \geq |F|+r_x$. 
\end{lemma}
\begin{proof}
  Decompose $C$ cyclically into edges in $E_1$ and \emph{$V_1$-paths},
  i.e., paths whose endpoints lie in $V_1$ and whose internal vertices lie
  in $V_2$, where we require each $V_1$-path to have at least one
  internal vertex. Let $P'$ be a $V_1$-path. We claim that the initial and final
  edges of $P'$ cannot both be in $F$. Indeed, all internal edges of
  $P'$ (except the initial and final edges) lie in $F$, and the full
  set of edges $E(P')$ intersect only $|E(P')|-1$ distinct vertices in
  $V_2$, i.e., $E(P')$ is dependent in $M$. Since this argument applies
  to every $V_1$-path in $C$ separately, and since the $V_1$-paths partition
  the edges of $E(C) \cap (E_X \cup E_2)$, we conclude
  \[
    |(E(C) \cap E_X) \setminus F| \geq r_x.
  \]
  Hence $|V(C)|=|E(C)| \geq |F|+r_x$.
\end{proof}

We show that this implies an algorithm for detecting long $(s,t)$-paths.

\begin{lemma} \label{lemma:alg-longcycle}
  Let $V=V_1 \cup V_2$ be a partition and $k$, $k_2$ and $k_x$ be integers. 
  Let $\mu=\max(k_2, k-k_x/2)$. There is a randomized, polynomial-space
  algorithm with running time $O^*(2^\mu)$ that detects the existence
  of an $(s,t)$-path $P$ such that $|V(P)| \geq k$, $|V(P) \cap V_2|=k_2$
  and $|E(P+st) \cap E_X|=k_x$.
\end{lemma}
\begin{proof}
  Assume $k>1$ (or else solve the problem in polynomial time),
  and reject if $k_x$ is odd. Remove any edge $st$ from the graph.
  We will use $st$ as a ``virtual'' edge to complete any
  $(s,t)$-path $P$ into a cycle and use Lemma~\ref{lemma:long-cycle-rank} 
  to look for cycles $C=P+st$ of length at least $k$.
  Formally, let $G=(V,E)$ be the input graph with no edge $st$
  present, and let $G'=G+st$ be the graph with $st$ introduced. 
  
  Let $\ell_1=\max(0, k-k_x/2-k_2)$ and note $\mu=k_2+\ell_1$. 
  Construct the matroid $M=M(\ell_1,k_2)$ over $G'$ on the ground set
  $E(G')=E(G)+st$. We will test whether $G'$ contains a cycle $C=P+st$
  and an edge set $F \subseteq E(C)$ such that the following hold.
  \begin{enumerate}
  \item $F$ is independent in $M$;
  \item $E(C) \cap E_2 \subseteq F$;
  \item $|F| \geq k-k_x/2$;
  \item $|F \cap E_X| \geq k_x/2$.
  \end{enumerate}
  By Lemma~\ref{lemma:long-cycle-rank}, if such a cycle exists, then
  $|V(C)| \geq k$. We show that the converse is true, i.e., if $G'$
  contains a cycle $C=P+st$ with $|V(C)| \geq k$ then there is a set
  $F \subseteq E(C)$ meeting the above conditions, and we show how to 
  use odd support sieving over the construction of
  Lemma~\ref{lemma:padded-st-paths} to detect such a pair $(C,F)$. 

  We first define the polynomial $P(X)$ that we are working over.
  Let $X=\{x_e \mid e \in E(G')\}$ be the variable set.
  Let $A_{st}$ be the matrix of Lemma~\ref{lemma:padded-st-paths}
  constructed from $G$, so that $\det A_{st}$ enumerates padded $(s,t)$-paths.
  Create additional variables $z_x$ and $z_2$, and scale the entries
  of $A_{st}$ so that variables $x_e$ for $e \in E_2$ are scaled by a factor of $z_2$,
  and variables $x_e$ for $e \in E_X$ are scaled by a factor of $z_x$.
  Similarly define $z_{st}=1$ if $s, t \in V_1$:
  $z_{st}=z_2$ if $s, t \in V_2$; and $z_{st}=z_x$ otherwise.
  Finally, we let $P(X)$ be the coefficient of $z_x^{k_x} z_2^{k_2-k_x/2}$ 
  in $x_{st}z_{st} \det A_{st}$. Associate every variable $x_e$,
  $e \in E(G')$ with the vector $M(e)$ in the representation of $M$.
  We claim that there is an $(s,t)$-path $P$ with $|V(P)| \geq k$,
  $|V(P) \cap V_2|=k_2$ and $|E(P+st) \cap E_X|=k_x$
  if and only if $P(X)$ contains a term whose odd support spans $M$.

  First, let $C=P+st$ be a simple cycle meeting the conditions.
  Orient $C$ arbitrarily cyclically and let $F \subseteq E(C)$
  consist of every edge oriented towards a vertex of $V_2$
  together with $\ell_1$ further edges of $E(C) \cap E_1$;
  note $|E(C) \cap E_1| \geq \ell_1$.
  Also, $F$ contains precisely $k_x/2$ crossing edges.
  Furthermore, clearly $F$ is a basis for $M$. 
  Then $x_{st}z_{st} \det A_{st}$ contains a monomial $m=x_{st}z_{st} \prod_{e \in E(P)} x_ez_e$
  (for the suitable value of $z_e \in \{1, z_2, z_x\}$),
  by taking the term from the path $P$ and using only loops for padding.
  The degree of $z_x$ in $m$ is precisely $k_x$
  and the degree of $z_2$ is precisely $k_2-k_x/2$.
  Hence $m$ also occurs in $P(X)$. This proves one direction of the equivalence.

  Conversely, let $m$ be a term in $P(X)$ whose odd support spans $M$,
  and let $F$ be a subset of the odd support of $m$ such that $F$ is a
  basis for $M$. Let $C$ be a cycle such that $m$ corresponds to a
  padding of $C$. By the definition of $P(X)$, $m$ contains precisely $k_x$ crossing
  edges and $k_2-k_x/2$ edges of $E_2$, counting both $C$ and any
  2-cycles in the padding. Now, every vertex of $V_2 \cap V(C)$
  contributes two endpoints in $E(C)$, every edge of $E_2 \cap E(C)$
  represents two such endpoints, and every edge of $E_X \cap E(C)$
  represents one such endpoint. Hence
  $|V_2 \cap V(C)|=|E_2 \cap E(C)|+|E_X \cap E(C)|/2$. 
  Since not all edges counted in the degrees of $z_x$ and $z_2$ must
  come from $C$ itself, this is upper bounded by $(k_2-k_x/2) + k_x/2=k_2$, 
  with equality only if no padding 2-cycle uses an edge of $E_X \cup E_2$.
  Furthermore, since $F$ spans $M$, $F$ represents precisely $k_2$ 
  edges incident with distinct vertices of $V_2$, and since we sieve
  in the odd support $F$ cannot use edges from any padding 2-cycles.
  We conclude that $C$ contains precisely $k_x$ crossing edges and
  is incident with exactly $k_2$ vertices of $V_2$.
  Finally, $|V(C)| \geq r(M)+k_x/2 \geq k$ by Lemma~\ref{lemma:long-cycle-rank}.
  The running time and failure probability follow from Theorem~\ref{theorem:matroid-sieving}.
\end{proof}

It now only remains to combine Lemma~\ref{lemma:alg-longcycle} with a
carefully chosen random partition strategy for $V=V_1 \cup V_2$.

\begin{proof}[Proof of Theorem~\ref{theorem:long-path}]
  Let $P$ be an $(s,t)$-path and $C=P+st$ a cycle, and let $|V(C)|=ck$, $c \geq 1$. 
  Sample a partition $V=V_1 \cup V_2$ by placing every vertex $v$ into $V_2$
  independently at random with some probability $p$.
  For fix choices of $p$, $k_2$ and $\ell_1$ we estimate the probability
  that $C$ contains precisely $k_2$ vertices of $V_2$ and
  precisely $\ell_1$ edges of $E_1$.
  First, the probability that $|V(C) \cap V_2|=k_2$ is
  \[
    p_1(p,k_2) := \binom{ck}{k_2} p^{k_2}(1-p)^{ck-k_2}.
  \]
  In particular, all $\binom{ck}{k_2}$ colourings of $V(C)$ 
  with $k_2$ members of $V_2$ are equally likely.
  Now, let us count the number among those colourings where
  there are precisely $k_x$ transitions between $V_1$ and $V_2$. 
  To eliminate edge cases, assume $k_2 < |V(P)|$, $k_x < 2k_2$
  and $k_x < 2(ck-k_2)$ and that $k_x$ is even, so that $k_x$ is achievable.
  To describe the outcomes, consider an initial shorter cycle
  of $k_2$ elements, all of which are coloured $V_2$, and consider
  the different ways to place $ck-k_2$ vertices coloured $V_1$ between
  these so that there are precisely $k_x$ transitions between $V_1$
  and $V_2$. Counting cyclically, this implies that there are
  precisely $k_x/2$ blocks of vertices coloured $V_1$.
  There are precisely
  \[
    \binom{ck-k_2-1}{k_x/2-1}
  \]
  ordered sequences of $k_x/2$ positive numbers that sum to $ck-k_2$.
  Indeed, these can be thought of as placing all $ck-k_2$ elements in
  a sequence (coding the number $ck-k_2$ in unary) and selecting $k_x/2-1$ 
  out of the $ck-k_2-1$ gaps between elements to insert a break between
  blocks. For every such ordered sequence, we similarly select
  \[
    \binom{k_2}{k_x/2}
  \]
  positions in the cycle of $V_2$-vertices into which to insert the blocks.
  This undercounts slightly -- e.g., for a given vertex $v \in V(C)$,
  this accurately counts the number of assignments where $v \in V_2$,
  missing assignments where $v \in V_1$ -- but it is tight up to a
  polynomial factor. Hence, given an outcome with $|V(C) \cap V_2|=k_2$ the probability
  of precisely $k_x$ crossing edges is at least
  \[
    p_2(k_2,k_x) := \binom{ck}{k_2}^{-1} \binom{ck-k_2-1}{k_x/2-1} \binom{k_2}{k_x/2}.
  \]
  Thus the total probability of meeting both conditions is at least
  \[
    p_3(p, k_2, k_x) = p_1(p,k_2) p_2(k_2,k_x)
    = p^{k_2}(1-p)^{ck-k_2} \binom{ck-k_2-1}{k_x/2-1} \binom{k_2}{k_x/2}.
  \]
  Given such an outcome, we can then detect a cycle by Lemma~\ref{lemma:alg-longcycle}
  in time $O^*(2^{\mu})$ where $\mu=\max(k_2,k-k_x/2)$.
  By repeating the algorithm $\Theta^*(p_3(p,k_2,k_x))$ times,
  we get a high probability of success, with a total running time of
  \[
    O^*(2^\mu/p_3(p, k_2, k_x)) = n^{O(1)}
    \frac{2^\mu}{p^{k_2}(1-p)^{ck-k_2}  \binom{ck-k_2-1}{k_x/2-1} \binom{k_2}{k_x/2}}.
  \]
  The choice of $p$, $k_2$ and $k_x$ will depend on $c$, but since
  there are only $n-k+1$ possible values of $|V(C)| \geq k$ we may
  repeat the algorithm for every such value. We follow approximately
  the analysis used by Bj\"orklund et al.~\cite{BjorklundHKK17narrow}.
  Due to $\mu$, the algorithm has two modes, depending on the value of $c$.
  The expected value of $k_x/2$ depends on $c$ and $p$, and is maximised
  at $p=1/2$ with expected value $ck/4$. When $c$ is close to 1 setting $p=1/2$
  yields $E[k_2]=ck/2 < k-E[k_x/2]=(1-c/4)k$, hence the algorithm is
  dominated by $\mu=k-k_x/2$, and the best strategy is to maximise $k_x$.
  Here, the analysis of Bj\"orklund et al.\ applies. At some crossover point
  (e.g., $c = 4/3$ if we use the na\"ive values $p=1/2$, $k_2=ck/2$,
  $k_x/2=ck/4$) the algorithm at $p=1/2$ becomes dominated by $\mu=k_2$,
  and the best strategy is to pick $p$ so that $E[k_2]=E[k-k_x/2]$.
  Let us consider the second case first. We refrain from optimizing
  the running time for these values (since this regime does not
  represent the limiting behaviour of the algorithm) and 
  use $k_2=cpk$ and $k_x/2=cp(1-p)k$. Then we set $p$ so that
  \[
    cpk = k-cp(1-p)k \Rightarrow p = 1-\sqrt{1-\frac{1}{c}}.
  \]
  It can easily be checked that with $k_2=pck$ and $k_x=2p(1-p)ck$,
  we get $1/p_3(p,k_2,k_x)=O^*(1)$ -- e.g., the first part is $2^{-H(p)ck}$
  and up to polynomial factors the binomial terms are $2^{H(p)c(1-p)k}$
  and $2^{H(1-p)cpk}=2^{H(p)cpk}$, respectively, where $H(p)=-p\log p-(1-p)\log (1-p)$ 
  is the binary entropy function. Thus the running time of the
  algorithm in this regime is $O^*(2^{c(1-\sqrt{1-1/c})k})$,
  where the exponent decreases with increasing $c$ (approaching $k/2$)
  and at $c=4/3$ it becomes $O^*(2^{2k/3})=O^*(1.59^k)$.
  Now we focus on the regime $c<4/3$, in which case we set $p=1/2$
  and $k_2=ck/2$ (to maximise the expected number of crossing edges).
  We set $k_x=2cp(1-p)k + 2\beta c k=(1/2 + 2\beta)ck$ where $\beta > 0$
  is a parameter to optimize. We are in the case $\mu = k-k_x/2$. Consider
  the effect of increasing $k_x$ by 2. Noting that
  \[
    \binom{n}{k+1} = \binom{n}{k} \cdot \frac{n-k}{k+1},
  \]
  the total running time is multiplied by a factor
  \[
    \frac{1/2}{
      \frac{(1/4-\beta)ck}{(1/4+\beta)ck} \cdot
      \frac{(1/4-\beta)ck}{(1/4+\beta)ck}}.
  \]
  For the best possible value of $\beta$, this will equal $1 \pm o(1)$,  
  since otherwise we can improve the running time by raising or
  lowering $k_x$. Thus
  \[
    1/4 + \beta = (1/4-\beta)\sqrt{2} \Rightarrow
    \beta = \frac{\sqrt{2}-1}{4(1+\sqrt{2})}  =   \frac{(\sqrt{2}-1)^2}{4} 
    = \frac{3}{4} - \frac{1}{\sqrt{2}}
  \]
  by multiplying with $\sqrt{2}-1$ in the next-to-last step.
  We now revisit the total running time. By Stirling's approximation,
  \[
    \binom{n}{\alpha n} 
    = \Theta^*\left(\left( \frac{1}{\alpha^{\alpha} (1-\alpha)^{1-\alpha}} \right)^n\right)
  \]
  (see~\cite{BjorklundHKK17narrow} for a derivation).
  Plugging the values ($p=1/2$, $k_2=ck/2$, $k_x/2=(1/4+\beta)ck=(1-1/\sqrt{2})ck$)
  into the running time and simplifying we get (up to a polynomial factor)
  \[
    \frac{2^{ck} 2^{k-(1-1/\sqrt{2})ck}}
    {\binom{ck/2}{(1-1/\sqrt{2})ck}^2}
    =
    2^k2^{(1/\sqrt{2})ck}(2-\sqrt{2})^{(2-\sqrt{2})ck} (\sqrt{2}-1)^{(\sqrt{2}-1)ck}
    =  2^k2^{ck} (\sqrt{2}-1)^{ck},
  \]
  where the last step follows by factoring $2-\sqrt{2} = \sqrt{2}(\sqrt{2}-1)$ and
  simplifying the result.
  This equals $O^*((4(\sqrt{2}-1))^k)=O^*(1.66^k)$ for the basic case $ck=k$,
  and is a decreasing function in $c$. Therefore, this analysis
  applies for values of $c$ up to the crossover point where $\mu=k_2$. 
  Switching from the running time $2^\mu=2^{k-k_x/2}$ to $2^\mu=2^{k_2}$
  for the above values of $k_2$ and $k_x$ 
  represents multiplying the running time by
  \[
    2^{k_2-k+k_x/2}=2^{ck/2-k+(1-1/\sqrt{2})ck}
  \]
  hence the running time after the crossover point, with the above
  parameters, is some function $O^*(\xi^{ck})$, $\xi>1$. We evaluate
  the formula at $c=4/3$ and find it reaches $O^*(1.62^k)$ at this point. 
  Hence no further case distinctions are needed. In summary, our
  algorithm has the following steps:
  \begin{enumerate}
  \item Repeat the below with every target value $ck \in \{k,\ldots,n\}$.
  \item If $c \leq 4/3$, set $p=1/2$, $k_2=ck/2$ and $k_x=2(1-1/\sqrt{2})ck$.
  \item If $c>4/3$, set $p = 1-\sqrt{1-1/c}$, $k_2=pck$ and $k_x=2p(1-p)ck$.
  \item Repeat $\Omega(n/p_3(p, k_2, k_x))$ times: Compute a partition $V=V_1 \cup V_2$ 
    by placing every vertex $v \in V$ into $V_2$ independently at
    random with probability $p$. Use Lemma~\ref{lemma:alg-longcycle}
    with partition $(V_1,V_2)$ and arguments $k$, $k_2$, $k_x$
    to detect a cycle in $G$ with the given parameters. If successful,
    return YES.
  \item If every attempt fails, return NO.
  \end{enumerate}  
  This concludes the proof of Theorem~\ref{theorem:long-path}.
\end{proof}

\section{Subgraph problems}
\label{sec:subgraphs}

Another major area of applications of algebraic methods in FPT
algorithms is problems of detecting a particular kind of subgraph. 
This is of course very broad (and in the form just described would
arguably cover all of Section~\ref{sec:paths-linkages} as well),
so we focus on two topics. Essentially, these topics correspond to two
families of enumerating polynomials: \emph{branching walks} and
\emph{homomorphic images}. 

The first topic, which is already very broad, is detecting whether a
given graph $G$ contains a connected subgraph meeting a particular
condition. The parameter here may either be the size of the subgraph
or a parameter related to the condition itself. We review some
examples. The first application of branching walks that we are aware
of was for the well-known \textsc{Steiner Tree} problem.
Here the input is a graph $G=(V,E)$ and a set of terminals $T \subseteq V$, $|T|=k$,
and the task is to find a smallest connected subgraph of $G$ that
spans $T$ (i.e., a subtree of $G$ whose vertex set contains $T$).
Nederlof~\cite{Nederlof13algorithmica} showed a polynomial-space,
$O^*(2^k)$-time algorithm for this problem, using an inclusion-exclusion
sieving algorithm. We may also consider generalisations of the problem.
In \textsc{Group Steiner Tree}, the terminals come in $k$ groups and
the task is to find the smallest subtree that contains a terminal of
each group. In \textsc{Directed Steiner Out-Tree}, the graph is
directed, and the task is to find the smallest out-tree (i.e., a
directed subtree of $G$ where every arc leads away from the root)
which spans the terminals. Misra et al.~\cite{MisraPRSS12} showed that both of these
variants can be solved in the same running time $O^*(2^k)$.

Another problem, perhaps of an apparently very different nature, is
\textsc{Graph Motif}. The precise definition is slightly involved,
but in its base variant the input contains a graph $G=(V,E)$,
a vertex colouring $c \colon V \to [n]$, an integer $k$,
and a capacity $d_q$ for every colour $q \in c(V)$.
The task is to find a subtree $T$ of $G$ on $k$ vertices such that
every colour $q$ occurs in at most $d_q$ vertices of $T$. 
As surveyed in the introduction, this problem led to the development
of the \emph{constrained multilinear detection} method by Björklund et
al., leading to a solution in time $O^*(2^k)$ for \textsc{Graph Motif}
as well as several weighted and approximate variants~\cite{BjorklundKK16}.

In our first application in this section, we note a 
generalisation of the above results.

\begin{theorem}
  \label{thm:subtree-framework}
  Let $G=(V,E)$ be an undirected graph and $M$ be a matroid over $V$. 
  Let $k, w \in \N$. If $M$ is represented over a field of characteristic 2,
  then in randomized time $O^*(2^k)$ and polynomial space we can
  detect the existence of a connected subgraph $H$ of $G$ such that
  $V(H)$
  has rank at least $k$ in $M$ and
  $|V(H)| \leq w$.
  If $M$ is represented over any other field, then the algorithm needs
  $O^*(2^{\omega k})$ time and $O^*(4^k)$ space. 
\end{theorem}

\begin{remark}
  In the conference version of this paper, we also claimed results
  for \textsc{Rank $k$ Connected Subgraph} with a matroid defined on
  the edge set of $G$ instead of the vertex set. In preparing 
  this journal version, we realized that the proof given previously
  was incomplete. An edge version of Theorem~\ref{thm:subtree-framework} is possible,
  but cannot use the simple branching walk polynomial $P_k(X,Y)$
  of Björklund et al.~\cite{BjorklundKK16}, but must use a ``decorated'' version
  (e.g., in the style of the results of Section~\ref{subsec:subgraphiso}). 
  Given that most (if not all) main applications of Theorem~\ref{thm:subtree-framework}
  are for matroids over the vertex set, we choose to omit the details
  of such an edge variant.
\end{remark}

The second result regards subgraph isomorphism. Given graphs $G$ and $H$,
the problem of checking whether $H$ is a subgraph of $G$ parameterized
by $|V(H)|$ can be either FPT, as when $H$ is a path,
or W[1]-hard, as when $H$ is a clique. More generally,
\textsc{Subgraph Isomorphism} is FPT by $|V(H)|$ if $H$ comes from a
family of graphs with bounded treewidth (originally shown using the
colour-coding technique of Alon et al.~\cite{AlonYZ95}),
and there is good evidence that no more general such class
exists~\cite{BringmannS21,Marx10twd}.
In general, the parameterized complexity of subgraph isomorphism
problems has been extensively and meticulously investigated~\cite{JansenM15,MarxP14everything}.

In fact, one of the fastest methods for \textsc{Subgraph Isomorphism}
works via an arithmetic circuit for evaluating the \emph{homomorphism polynomial}~\cite{FominLRSR12},
allowing for the randomized detection of a subgraph $H$ with
$|V(H)|=k$ and of treewidth $w$ in time $O(2^kn^{w+1})$. 
Furthermore, the exponent $w+1$ here is optimal, up to plausible conjectures~\cite{BringmannS21}.
We observe that this running time is compatible with an additional constraint that 
the copy of $H$ found in $G$ should be independent in a given linear matroid.
Since we in this application care about the concrete polynomial factor,
unlike in the rest of the paper, we make an additional assumption that 
field operations can be performed in $k^{O(1)}$ time. As briefly discussed
in Section~\ref{subsec:conventions}, this covers many but not all matroids in common use
in parameterized complexity. 

\begin{theorem}
  \label{thm:subgraph-hom}
  Let $G$ and $H$ be undirected graphs, $k=|V(H)|$ and $n=|V(G)$.
  Let a tree decomposition of $H$ of width $w$ be given. Also let $M$
  be a matroid over $V(G)$. If $M$ is represented over a field of
  characteristic 2, then in randomized time $O(2^k \cdot k^{O(1)} \cdot n^{w+1})$
  and polynomial space we can detect whether there is a subgraph of $G$ 
  isomorphic to $H$ whose vertex set is independent in $M$.
  Similarly, given $M$ over 
  $V(G) \cup E(G)$ in time
  $O(2^{k+E(H)} \cdot k ^{O(1)} \cdot n^{w+1})$
  we can detect a subgraph of $G$ isomorphic to $H$ whose
  edge and vertex set, are independent in $M$.
  Here, we assume that field operations over $\F$ take at most $k^{O(1)}$ time.
  Over a general field, the algorithm needs $O^*(4^r)$ space and 
  the running time becomes $O(2^{\omega r} k^{O(1)} n^{w+1})$,
  or $O(4^r k^{O(1)} n^{w+1})$ if a path decomposition of $H$ of width $w$ is provided
  instead of a tree decomposition. Here, $r=k$ if $M$ is over $V(G)$
  and $r=k+|E(G)|$ if $M$ is over $V(G) \cup E(G)$.
\end{theorem}

\subsection{Finding high-rank connected subgraphs}

For our first application, generalizing \textsc{Steiner Tree} and
\textsc{Graph Motif}, we apply the concept of branching walks.
Informally, branching walks in a graph $G$ are a relaxation of
subtrees of $G$, similar to how walks are a generalisation of paths.
More formally, a branching walk in $G$ can be described as a tree
$T$ and a homomorphism mapping $T$ into $G$. Let us recall the
definitions. 

\begin{definition}
  Let $G$ and $H$ be undirected graphs. A \emph{homomorphism} from $G$
  to $H$ is a mapping $\varphi \colon V(G) \to V(H)$ such that
  for every edge $uv \in E(G)$, $\varphi(u)\varphi(v) \in E(H)$. 
\end{definition}

Branching walks were defined by Nederlof~\cite{Nederlof13algorithmica}.
We use the more careful definition of Björklund et al.~\cite{BjorklundKK16}.

\begin{definition}
  Let $G$ be a graph. A \emph{branching walk} $W=(T,\varphi)$
  is an ordered, rooted tree $T$ and a homomorphism
  $\varphi$ from $T$ to $G$.
  We assume w.l.o.g.\ that $V(G)=\{1,\ldots,n\}$ and $V(T)=\{1,\ldots,|V(T)|\}$,
  where $V(T)$ is ordered according to the preorder traversal of $T$.
  We say that $W$ \emph{starts from} the vertex $\varphi(1)$ in $G$.
  The \emph{size} of $W$ is $|V(T)|$ and its \emph{span} is $\varphi(V(T))$.   
  $W$ \emph{visits} a vertex $v \in V(G)$ if $v \in \varphi(V(T))$.
  $W$ is \emph{simple} if $\varphi$ is injective. 
  Finally, $W$ is \emph{properly ordered}
  if for any two sibling nodes $a, b \in V(T)$ with $a<b$
  we have $\varphi(a)<\varphi(b)$. 
\end{definition}

The ordering here is a technical device to make the map non-ambiguous.
Björklund et al.\ define a generating polynomial (or in our terms,
an enumerating polynomial) for properly ordered branching walks.
We recall their construction next.
Fix a host graph $G=(V,E)$ and a size $k$ for the branching walk.
Introduce two sets of variables $X=\{x_v \mid v \in V(G)\}$
and $Y=\{y_{(u,v)}, y_{(v,u)} \mid uv \in E(G)\}$.
For a properly ordered branching walk $W = (T, \varphi)$ in $G$, define the
corresponding monomial
\[
  m(W,X,Y) = x_{\varphi(1)} \prod_{ab \in E(T): a<b} y_{(\varphi(a),\varphi(b))} x_{\varphi(b)}.
\]
As Björklund et al.\ show, $m(W,X,Y)$ is multilinear if and only if
$W$ is simple, and $W$ can be reconstructed from the factors of $m(W,X,Y)$. 
Given a target size $k$ for $W$ and a starting vertex $s \in V(G)$, define
\[
  P_{k,s}(X,Y) = \sum_W m(W,X,Y) \quad \text{ and } \quad P_k(X,Y) = \sum_{s \in V(G)} P_{k,s}(X,Y),
\]
where the sum goes over all properly ordered branching walks of size
$k$ in $G$ that start from $s$. Björklund et al.\ show that $P_{k,s}(X,Y)$
can be evaluated in time polynomial in $n+k$ (in fact, in $O(k^2 m)$ field operations,
where $m=|E(G)|$)~\cite{BjorklundKK16}. 

We now visit our target result, Theorem~\ref{thm:subtree-framework}.
We observe the key property of branching walks that make them
algorithmically useful: A \emph{minimal}
branching walk spanning a given vertex set is always a subtree of $G$.

\begin{lemma} \label{lemma:branching-walks-simple}
  Let $G=(V,E)$ be a graph and let $U \subseteq V$ be a set of vertices such that $G[U]$ is connected.
  Then there is a properly ordered branching walk $W$ in $G$ with span $U$
  and size $|U|$ such that the corresponding monomial $m(W,X,Y)$ is contributed only once in $P_{|U|}(X,Y)$.
  Furthermore, any branching walk with span $U$ and size $|U|$ is simple.
\end{lemma}
\begin{proof}
  Clearly, every branching walk $W$ with span $U$ needs size at least $|U|$,
  and any branching walk whose size equals the cardinality of its span is simple.
  For existence, let $T$ be an arbitrary spanning tree of $G[U]$
  where the nodes of $T$ are ordered in preorder traversal, such
  that at every vertex the lowest-index unvisited child is visited
  first. Let $W=(T,\varphi)$ where $\varphi$ is the inverse of the
  resulting vertex ordering of $T$. Then $W$ is a properly ordered
  branching walk with size $|U|$ and span $U$. Finally,
  Björklund et al.~\cite{BjorklundKK16} show that any simple properly
  ordered branching walk can be reconstructed from its monomial
  fingerprint $m(W,X,Y)$. It follows that $m(W,X,Y)$ has coefficient
  precisely 1 in $P_{|U|}(X,Y)$. 
\end{proof}

Given this, and given the ability to evaluate $P_k(X,Y)$,
Theorem~\ref{thm:subtree-framework} follows easily.

\begin{proof}[Proof of Theorem~\ref{thm:subtree-framework}]
  Let $(G, M, k, w)$ be the input, where $M$ is represented by a
  matrix $A$ over a field of characteristic 2. We assume by truncation
  that $A$ has dimension $k \times V(G)$ and rank $k$,
  and let $w=\min(w,|E(G)|)$. 
  Furthermore, ensure that $A$ is over a field $\F$ of size
  $\omega(n^2)$, e.g., $\F=GF(2^{c \log n})$ for $c > 2$.
  For each $\ell=k, \ldots, w$ let $P_\ell(X,Y)$ be the branching walk
  polynomial for branching walks of size $\ell$ over variable sets $X$ and $Y$
  defined above. Use Theorem~\ref{theorem:matroid-sieving}
  with the vectors of $A$ associated with $X$ 
  and assume that for some $\ell$, Theorem~\ref{theorem:matroid-sieving} reports that $P_\ell(X,Y)$
  contains a monomial $m$ whose odd support spans $A$.
  Then $m=m(W,X,Y)$ for a branching walk $W$. Let $W=(T,\varphi)$
  and let $S \subseteq \varphi(V(T)) \cup \varphi(E(T))$ correspond 
  to the subset of the odd support of $m$ that spans $A$, $|S|=k$. 
  Since $W$ is a branching walk, $S$ is the vertex set 
  of a connected subgraph of $G$ on at most $\ell$ vertices. 
  Hence $(G,M,k,w)$ is a YES-instance.

  On the other hand, assume that $(G,M,k,w)$ is a YES-instance
  and let $H$ be a subgraph of minimum cardinality that spans $A$.
  Let $\ell=|V(H)|$. 
  By Lemma~\ref{lemma:branching-walks-simple}, 
  there is a simple branching walk $W$ with span $V(H)$.
  Thus $P_\ell(X,Y)$ contains a monomial $m=m(W,X,Y)$ which is multilinear in $X$
  and which spans $A$. 
  Note that $P_\ell(X,Y)$ is a homogeneous polynomial of degree
  $2\ell - 1 < n^2$, so the probability of a false negative for
  $P_\ell(X,Y)$ is $o(1)$.
  Hence with probability $1-o(1)$, the algorithm reports that the
  input is a YES-instance.
  The running time follows from Theorem~\ref{theorem:matroid-sieving}.

  In the case that $M$ is represented over some other field $\F$, the same
  analysis applies (including assuming $|\F| = \omega(n^2)$), but the 
  running time and space complexity follow from
  Theorem~\ref{lemma:extensor-matroid-sieving} instead.
  More precisely, instead of directly evaluating $P_\ell(X,Y)$,
  we evaluate $P_\ell(X',Y)$ for a new set of variables $X'=\{x' \mid x \in X\}$,
  at a value of $x'=1+x$ for every $x \in X$. We argue that 
  this works as a form of ``spanning set sieve'' for $P_\ell(X,Y)$
  over arbitrary characteristic. Indeed, if an application
  of Theorem~\ref{lemma:extensor-matroid-sieving} reports that
  there is a multilinear monomial $m$ in $P_\ell(X',Y)$ such that the support
  of $X'$ in $m$ spans $M$, then $m$ is produced from some monomial
  $m(W,X,Y)$ in $P_\ell(X,Y)$, and the input is a YES-instance.
  On the other hand, if the input is a YES-instance, then 
  $P_\ell(W,X,Y)$ contains a multilinear monomial $m(W,X,Y)$
  which spans $M$; let $S \subseteq X$ be the support of $m$ in $X$
  of some basis $B$ of $M$. Then $P_\ell(X',Y)$ contains the monomial
  $m'=\prod_{v \in B} x_v \prod_{y \in \supp(m) \cap} y$ with coefficient 1:
  it is clear that $m'$ is produced precisely once from $m$. 
  Since $\supp(m) \cap Y = \supp(m') \cap Y$, and $m$ can be recovered from
  its support in $Y$, $m'$ cannot be produced from any other
  monomial in $P_\ell(X,Y)$. 
\end{proof}

We note some applications of this result. First, consider the basic
\textsc{Steiner Tree} problem, and let $G=(V,E)$ and $T \subseteq V$
be an input, $T=\{t_1,\ldots,st_k\}$. Define a $k$-dimensional matroid
$M$ over $V$ by letting vertex $t_i$ be associated with vector $e_i$,
and every other vertex associated with the $k$-dimensional zero
vector. Then a connected subgraph $H$ of $G$ spans $M$ if and only if
$T \subseteq V(H)$. We can cover \textsc{Group Steiner Tree} with a
similar construction. Let the input be $(G=(V,E), \cT)$, with
terminal grouping $\cT=\{T_1,\ldots,T_k\}$, $T_i \subseteq V$ for each $i$.
We assume the terminal sets are pairwise disjoint by adding pendants:
for every $T_i \in \cT$ and every vertex $t \in T_i$, add a pendant
$t^i$ to $t$ and replace $T_i$ by the set $\{t^i \mid t \in T_i\}$. 
This raises the size of a minimum solution by precisely $k$ vertices.
We can now apply label $e_i$ to every vertex in $T_i$ and the zero vector
as label to every other vertex and proceed as above.

Next, let us review how to use matroid constructions to solve the various
optimization variants of \textsc{Graph Motif} surveyed by Björklund
et al.~\cite{BjorklundKK16}. Let $(G=(V,E), c, k, (d_q)_{q \in c(V)})$
be a \textsc{Graph Motif} instance. Additionally, we consider the
following operations, mimicking the \textsc{Edit Distance} problem.
Let $H$ be a connected subgraph of $G$ with $k$ vertices.
Let $C=C(H)$ be the multiset of vertex colours in $H$, i.e.,
$C(H)=\{c(v) \mid v \in V(H)\}$ with element multiplicities preserved.
To \emph{substitute} a colour $q \in C$ for another colour $q' \in c(V)$,
we remove one copy of $q$ from $C$ and add a copy of $q'$.
To \emph{insert} a colour $q$, we add a copy of $q$ to $C$.
To \emph{delete} a colour $q$, we remove a copy of $q$ from $C$.
Furthermore, let a multiset $Q$ of colours be given. We wish to decide
whether there is a connected subgraph $H$ of $G$ with $|V(H)|=k$
such that $C(H)$ can be transformed into $Q$
by making at most $k_s$ substitutions, at most $k_i$ insertions,
and at most $k_d$ deletions.
Individually, these operations correspond well to standard matroid
transformations. Let $M$ be the partition matroid over $V(D)$ where
every colour class $c^{-1}(q)$ has capacity $d_q$.
Then using $M$ in Theorem~\ref{thm:subtree-framework} directly solves
\textsc{Graph Motif}. Further allowing substitutions, insertions
and/or deletions can be handled by combinations of extensions and
truncations over $M$. We consider the following general case.

\begin{lemma}
  Let $G=(V,E)$ be a graph and $c$ a vertex colouring of $G$.
  Let $k_s, k_d, k_i \in \N$ and a multiset $Q$ be given.
  There is a matroid $M$ over $V$, representable over a field of characteristic 2,
  such that any set of $k$ vertices from $V$ forms a basis of $M$
  if and only if the multiset $C(H)$ can be transformed into $Q$
  by making at most $k_s$ substitutions, $k_d$ deletions and $k_i$
  insertions. 
\end{lemma}
\begin{proof}
  We note that since $C(H)$ and $Q$ are multisets, without element
  order, finding the minimum cost for a transformation is much simpler
  than in \textsc{Edit Distance}. Let $C=C(H)$ and let $C_0=C \cap Q$
  be the multiset intersection. Let $a=\min(|C|-|C_0|,|Q|-|C_0|,k_s)$.
  
  \begin{claim}
    $C(H)$ can be transformed into $Q$ with the operation limits
    prescribed if and only if $|C_0|+a \geq \max(|Q|-k_i,k-k_d)$.
  \end{claim}
  \begin{subproof}
    The transformation can use $a$ substitutions, and thereafter it
    has either exhausted $C$, $Q$ or the budget $k_s$. In either case, 
    it thereafter needs to use $|C|-|C_0|-a$ deletions and
    $|Q|-|C_0|-a$ insertions, i.e.,
    $k_d \geq |C|-(|C_0|+a)$ and $k_i \geq |Q|-(|C_0|+a)$. 
    The result follows.
  \end{subproof}

  Hence, let $r=|C_0|+a$ be the number of vertices from $V(H)$ that
  can be matched against $Q$ by using at most $k_s$ substitutions.
  We wish to accept a vertex set $V(H)$ if and only if
  $r \geq k_0 := \max(|Q|-k_i,k-k_d)$. We construct a matroid for this purpose.
  If $k_0>k$, reject the parameters. Otherwise, execute the following
  construction sequence.
  \begin{enumerate}
  \item Let $M_1$ be the partition matroid over $V(G)$ where for every
    colour $q$, the set $c^{-1}(q)$ has capacity in $M_1$ corresponding
    to its count in $Q$. 
  \item Let $M_2$ be $M_1$ with the rank extended by $k_s$.
  \item Let $M_3$ be $M_2$ truncated to rank $k_0$.
  \item Let $M$ be $M_3$ extended by additional rank $k-k_0$.    
  \end{enumerate}
  Indeed, let $S$ be a basis of $M$. Then there is a set $S' \subseteq S$
  with $|S'|=k_0$ that is a basis of $M_3$, implying that using at
  most $k_s$ substitutions, $S'$ can be matched into $Q$, and $|S'|=k_0$.
  Conversely, if $S$ is a set of $k$ vertices, let $C'=C(S) \cap Q$,
  and assume that $|C'|+\min(k-|C'|, |Q|-|C'|,k_s) \geq k_0$.
  Then there is a set $S' \subseteq S$ consisting of $k_0$
  vertices that can be matched into $Q$ using at most $k_s$
  substitutions, i.e., $S'$ is independent in $M_3$,
  and $S$ is a basis of $M$. 
\end{proof}

Since there are only $O(k^3)$ valid options for the integers $k_s$,
$k_d$ and $k_i$, by repeating this construction we can clearly sieve
for a connected subgraph $H$ of size $k$ with a minimum \emph{cost}
of transformation, using costs as given in Björklund et al.~\cite{BjorklundKK16}.
Another option, of attaching weight-tracing variables keeping track of the
number of substitutions and deletions, similarly to the algorithm in~\cite{BjorklundKK16},
would of course also be possible, but our purpose here was to
illustrate the matroid construction.

Further variations are clearly also possible, e.g., as in the notion
of \emph{balanced solutions} considered in Section~\ref{sec:balanced}
one may look for subgraphs with both upper and lower bounds 
$d_q \leq |V(H) \cap c^{-1}(q)| \leq e_q$ for every colour class~$q$. 

Interestingly, both \textsc{Steiner Tree} and \textsc{Graph Motif}
are SeCoCo-hard, i.e., under the set cover conjecture they cannot be
solved in time $O^*(2^{(1-\varepsilon)k})$ for any $\varepsilon > 0$~\cite{CyganDLMNOPSW16seth}.
Hence improving the algorithm of Theorem~\ref{thm:subtree-framework} is certainly SeCoCo-hard as well. 

\subsection{Independent subgraph isomorphism}
\label{subsec:subgraphiso}
Next, we review the \textsc{Subgraph Isomorphism} problem.
Let $G$ and $H$ be undirected graphs, and introduce a variable set $X=\{x_v \mid v \in V(G)\}$. 
Let $k=|V(H)|$ and $n=|V(G)|$. The \emph{homomorphism polynomial}
is the polynomial
\[
  \sum_{\varphi \colon V(H) \to V(G)}  \prod_{v \in V(H)} x_{\varphi(v)}
\]
where the sum goes over homomorphisms $\varphi$.
If a treewidth decomposition of width $w$ is given for $H$,
then the homomorphism polynomial can be evaluated
in time $f(k) \cdot n^{w+1}$ for a modest function $f(k)$~\cite{FominLRSR12}.
In particular, we follow the exposition of Brand~\cite{Brand19thesis}
who shows that in time $O(c^k+n^{w+1})$ for $c<2$ we can both compute a tree
decomposition of width $w$ for $H$ and construct an algebraic circuit
of total size $O(k \cdot n^{w+1})$ which evaluates the homomorphism
polynomial.

Since the homomorphism polynomial has no negative terms working over,
e.g., the reals there is no concern for cancellations. However, since
we want to work over fields of characteristic 2, we introduce
additional terms in the way of \emph{algebraic fingerprinting}
(cf.~\cite{KoutisW16CACM}), to prevent cancellations.
In fact, we introduce two sets of additional variables for algorithmic
convenience. Let $X'=\{x'_{i,v} \mid i \in V(H), v \in V(G)\}$
and $Y=\{y_e \mid e \in V(G)\}$. Then we define
the \emph{decorated homomorphism polynomial}
\[
  P_{H \to G}(X,X',Y) = \sum_{\varphi \colon V(H) \to V(G)}
  \prod_{i \in V(H)} x_{\varphi(i)} x'_{i,\varphi(i)} 
  \prod_{ij \in E(H)} y_{\varphi(i)\varphi(j)},
\]
where the variables $y_{\varphi(i)\varphi(j)}=y_{\varphi(j)\varphi(i)}$
are taken without order on its subscript terms,
and where $\varphi$ ranges over all homomorphisms from $H$ to $G$. 
Then clearly, every homomorphism $\varphi$ has a unique algebraic
``fingerprint'' monomial $m(\varphi)$,  since it is encoded in the
$X'$-factors~of~$m(\varphi)$.

It is easy to modify the construction of Brand~\cite[Section~4.6.1]{Brand19thesis}
to construct a circuit for (or directly compute) $P_{H \to G}(X,X',Y)$
at a slightly larger polynomial cost in $k$.

\begin{proposition}
  Let a tree decomposition of with $w$ for $H$ be given. Then in time
  $k^{O(1)}n^{w+1}$ we can construct an algebraic circuit of size
  $k^{O(1)}n^{w+1}$ that computes $P_{H \to G}$. Furthermore, if the
  decomposition is a path decomposition, then the circuit can be made
  skew.   
\end{proposition}

We can now show Theorem~\ref{thm:subgraph-hom}. This follows the
obvious path, with some extra care taken to ensure that our polynomial
term remains $n^{w+1}$.

\begin{proof}[Proof of Theorem~\ref{thm:subgraph-hom}]
  Let $G$ and $H$ be given, as well as a tree decomposition of $H$ of
  width $w$. Note that a homomorphism $\varphi \colon H \to G$
  represents a subgraph of $G$ isomorphic to $H$ if and only if
  $\varphi$ is injective on $V(H)$, which holds if and only if
  $\prod_{i \in V(H)} x_{\varphi(i)}$ is multilinear. 
  
  Let $M$ be a linear matroid, and let $r=k$ if $M$ is over
  $V(G)$ or 
  $r=k+|E(H)|$ if $M$ is over $V(G) \cup E(G)$. We assume that $M$ is represented by a matrix
  $A$ with $r$ rows and rank $r$. We also assume $A$ is over a
  sufficiently large field $\F$ (where in fact some $|\F|=\Omega(k^2)$ suffices
  since only the degree of $P_{H \to G}$ is important for correctness).
  We now employ the sieving of Theorem~\ref{theorem:sieve-basis}.
  Note that $P_{H \to G}$ is homogeneous of degree $2k+|E(H)|$
  (and homogeneous in $X$ and $Y$ separately). Using the precise
  running time bound from Theorem~\ref{theorem:sieve-basis},
  we note that field operations over $\F$ can be performed
  in $\log^{O(1)} |\F|=\tilde O(1)$ time, independent of $n$.
  A larger field $\F$ could be given in the input, but in this case
  operations over $\F$ take only $k^{O(1)}$ time by assumption.
  If there is a subgraph $H'$ of $G$ isomorphic to $H$ and independent
  in $M$, then $P_{H \to G}$ will contain a term $m$ corresponding to that
  map $\varphi$, thus $m$ is multilinear in $X \cup Y$ and
  Theorem~\ref{theorem:sieve-basis} applies and will detect $H'$ with
  high probability. Conversely, assume that Theorem~\ref{theorem:sieve-basis}
  detects a monomial $m$ such that $m$ (in $X$, respectively in $X \cup Y$)
  is multilinear and contains a basis for $M$. Then $m$ must be
  multilinear in $X$, since $m$ spans $M$. Since the monomials of
  $P_{H \to G}$ are in 1-to-1 correspondence with homomorphisms
  $\varphi \colon H \to G$, $m$ must represent a homomorphism $\varphi$
  which is injective over $V(H)$. Thus $H'$ is isomorphic to $H$ as
  noted above, and is independent in $M$. If $M$ is over $V(G) \cup E(G)$,
  then the same argument and algorithm applies except that we are
  sieving over both the variable sets $X$ and $Y$. 

  The running time over a general field follows from
  Theorem~\ref{lemma:extensor-matroid-sieving}. In particular, Brand~\cite{Brand19thesis}
  notes that the homomorphism polynomial circuit can be made skew if
  constructed over a path decomposition, and not otherwise.
\end{proof}

The result with a matroid over $V(G) \cup E(G)$ could also be achieved
by simply subdividing every edge of $H$ and $G$, but this would blow
up $|V(G)|$ and hence the polynomial factor of the algorithm.

As with our other matroid applications, there is a range of
consequences, including (e.g.) finding a colourful copy of $H$
in a vertex-coloured graph; finding a copy of $H$ in $G$ subject to
capacity constraints on vertex classes; finding a copy $H'$ of $H$ in
$G$ such that $G-E(H')$ is connected; and all the other applications
of matroid constraints covered in this paper.

\section{Speeding-up Dynamic Programming}
\label{sec:noDP}

The notion of \emph{representative sets} for linear matroids plays an essential role in the design of FPT algorithms \cite{FominLPS16JACM,FominLPS17}, as well as kernelization \cite{KratschW20JACM}.
For a matroid $M = (V, \cI)$, a set $X \subseteq V$ is said to \emph{extend} a set $Y$, if $X$ and $Y$ are disjoint and $X \cup Y$ is independent in $M$.
The \emph{representative set lemma}, due to Lov\'asz \cite{Lovasz1977} and Marx \cite{Marx09-matroid}, states the following:
Let $M = (V, \cI)$ be a linear matroid of rank $k$, and $\mathcal Y \subseteq 2^V$ be a collection of subsets of $V$.
Then, there is a subcollection $\mathcal Y' \subseteq \mathcal Y$ (which can be computed ``efficiently'') of size at most $2^k$ that \emph{represents} $\mathcal Y$, i.e., for every $X \subseteq V$, there is a set in $\mathcal Y$ extending $X$ if and only if such a set exists in $\mathcal Y'$.
There are plethora of dynamic programming FPT algorithms in the literature, where the table size is reduced from $n^k$ to $2^k$, using the representative set lemma.
In this section, we exemplify how to use determinantal sieving in place of such dynamic programming approaches in three applications, \textsc{Minimum Equivalent Subgraph}, \textsc{Eulerian Deletion}, and \textsc{Conflict-free Solution}.
We improve the running time over existing algorithms, while saving space usage to polynomial.

\subsection{Minimum Equivalent Graph}

\textsc{Minimum Equivalent Graph} is defined as follows.
We are given a directed graph $G = (V, E)$ and an integer $k$, and the question is whether there is a subgraph $G' = (V, E')$ with at most $k$ edges with the same reachability pattern, i.e., for every $u, v \in V$, there is a $uv$-path in $G$ if and only if there is in $G'$.
Fomin et al.~\cite{FominLPS16JACM} show that \textsc{Minimum
  Equivalent Graph} reduces to the following question:
Are there a pair $B_1$, $B_2$ where $B_1$ is an in-branching
and $B_2$ is an out-branching in $G$, with a common root $v$,
such that they have at least $\ell$ edges in common? 
We phrase this as a matroid-theoretical problem. 

Let \textsc{Matroid Intersection Overlap}\footnote{Called
  \textsc{Matroid Intersection Intersection} in an earlier version of
  the paper} refer to the
following problem. The input is four matroids $M_1$, \ldots, $M_4$
over the same ground set $U$, each of rank $k$, and an integer $\ell>0$.
The question is whether there are bases $B_A \in M_1 \cap M_2$
and $B_B \in M_3 \cap M_4$ such that $|B_A \cap B_B| \geq \ell$.
This captures the above question, since rooted in- and
out-branchings can be constructed as the intersection of a graphic
matroid and a suitable partition matroid. 

\begin{lemma}
  For matroids represented over a common field $\F$ of characteristic 2, 
  \textsc{Matroid Intersection Overlap} can be solved in
  $O^*(2^{2k})$ time. 
\end{lemma}
\begin{proof}
  We define a new ground set $U^*=U_1 \cup U_2$, where
  $U_1$ is a copy of $U$ and $U_2=U \times U$.
  We also define new matroids $M_1'$, \ldots, $M_4'$ as follows.
  Let $A_i$, $i=1, \ldots, 4$ be the representation of $M_i$.
  Then $M_i'$ is represented by a matrix $A_i'$ where
  for $x \in U_1$ we have $A_i'[\cdot,x]=A_i[\cdot,x]$,
  and for $(x,y) \in U_2$ we have
  $A_i'[\cdot,(x,y)]=A_i[x]$ for $i=1, 2$ and
  $A_i'[\cdot,(x,y)]=A_i[y]$ for $i=3,4$.
  
  \begin{claim}
    The rank of $M_i'$ for each $i=1, \ldots, 4$ is $k$.
  \end{claim}
  \begin{subproof}
    Since each column of $A_i'$ is a copy of a column from $A_i$,
    clearly the rank of $A_i'$ is at most the rank of $A_i$.
    Conversely, since every column of $A_i$ occurs as a column of $A_i'$,
    the rank of $A_i'$ is at least the rank of $A_i$. 
  \end{subproof}
  
  We show that this reduces \textsc{Matroid Intersection Overlap}
  to a kind of ``weighted'' instance of \textsc{4-Matroid Intersection}
  over matroids of rank $k$.
  
  \begin{claim}
    The input instance is positive if and only if there is a common
    basis of $M_1'$ through $M_4'$ that contains at least $\ell$
    elements from $U_1$. 
  \end{claim}
  \begin{subproof}
    The idea is the following. Let $(B_A, B_B)$ be a solution to the problem.
    Split $(B_A, B_B)$ as $B_0=B_A \cap B_B$, $B_1 = B_A \setminus B_B$
    and $B_2=B_B \setminus B_A$, $|B_0|=r$ for some $r \geq \ell$.
    Write $B_0=\{u_1,\ldots,u_r\}$, 
    $B_1=\{x_1,\ldots,x_{k-r}\}$ and $B_2=\{y_1,\ldots,y_{k-r}\}$.
    Define the new set
    \[
      S=B_0 \cup \{(x_i,y_i) \mid i \in [k-r]\}
    \]
    where $B_0 \subseteq U_1$ and $S \setminus B_0 \subseteq U_2$. 
    Then $S$ is a common basis of $M_1'$ through $M_4'$. 
    Indeed, for $M_1'$ and $M_2'$ the matrix
    $A_i'[\cdot,S]$ induced by $S$ is a copy of $A_i[\cdot,B_A]$
    and for $M_3'$ and $M_4'$ the matrix
    $A_i'[\cdot,S]$ is a copy of $A_i[\cdot,B_B]$.
    These are bases by assumption. Furthermore $S$ contains
    $r \geq \ell$ elements from $U_1$.

    Conversely, assume that $S$ is a common basis of $M_1'$ through
    $M_4'$ with $|S \cap U_1|=r$ for some $r \geq \ell$.
    Extract the sets
    \[
      B_A = (S \cap U_1) \cup \{x \mid (x,y) \in S \cap U_2\}
      \quad \text{and} \quad
      B_B = (S \cap U_1) \cup \{y \mid (x,y) \in S \cap U_2\}.
    \]
    We claim that $B_A$ is a basis for $M_1$ and $M_2$,
    and $B_B$ a basis for $M_3$ and $M_4$. Indeed, we have
    $|B_A|=|B_B|=|S|$ since otherwise one matrix $A_i'[\cdot,S]$
    would contain duplicated columns. 
    Thus the matrices $A_i'[\cdot,S]$
    and $A_i[\cdot,B_A]$ (respectively $A_i[\cdot,B_B]$)
    are identical up to column ordering.
  \end{subproof}
  
  Hence we are left to solve the question whether there exists
  a common basis for $M_1'$ through $M_4'$ with at least $\ell$
  elements from $U_1$. For this, we proceed as in the algorithm for
  \textsc{4-Matroid Intersection}. By the Cauchy-Binet formula,
  there is a polynomial $P(X)$ over $X=\{x_u \mid u \in U\}$
  which enumerates common bases of $M_1'$ and $M_2'$.
  Introduce a new variable $z$ and 
  evaluate $P(X)$ at a value where $x_u$ is multiplied by $z$ for $u \in U_1$. 
  Let $P_r'(X)$ be the coefficient of $z^r$ in $P(X)$ under this evaluation. 
  The question now reduces to asking if there is a monomial $m$ in
  $P_r'(X)$ for any $r \geq \ell$ such that $m$ is independent in
  $M_3'$ and $M_4'$, which can be solved using Corollary~\ref{cor:basis}
  in~time~$O^*(2^{2k})$. 
\end{proof}

Since rooted in-branchings and out-branchings can be represented via matroid intersection over matroids of rank $n-1$, an $O^*(2^{2n})$-time algorithm for \textsc{Minimum Equivalent Graph} follows.

\subsection{Eulerian Deletion}

An undirected (or directed) graph is said to be \emph{Eulerian} if it admits a closed walk that visits every edge (or arc, respectively.) exactly once.
It is known that an undirected graph is Eulerian if and only if it is connected and \emph{even}, i.e., every vertex has even degree and that a directed graph is Eulerian if and only if weakly connected and \emph{balanced}, i.e., every vertex has the same number of in-neighbors as out-neighbors~\cite{book-digraphs}.
\textsc{Undirected Eulerian Edge Deletion} (\textsc{Directed Eulerian Edge Deletion}) is the problem of determining whether the input undirected (directed) graph has an edge (arc) set $S$ of size at most $k$ such that $G \setminus S$ is Eulerian.
Cai and Yang~\cite{CaiY11} initiated the parameterized analysis of these problems among other related problems.
The parameterized complexity of \textsc{Undirected Eulerian Edge Deletion} and \textsc{Directed Eulerian Edge Deletion} was left open by Cai and Yang~\cite{CaiY11}.
Cygan et al.~\cite{CyganMPPS14euler} designed the first FPT algorithms with running time $O^*(2^{O(k \log k)})$ based on the colour-coding technique.
Later, Goyal et al.~\cite{GoyalMPPS18even} gave improved algorithms running in time $O^*(2^{(2 + \omega)k})$ using a representative set approach.

We briefly describe their approach on \textsc{Undirected Eulerian Edge Deletion} (the directed version is similar).
Let $T$ denote the set of vertices of odd degree in $G$. (Note that $T$ must have even cardinality.)
An edge set $S$ is called a \emph{$T$-join} if $T$ is exactly the set of vertices of odd degree in the graph $(V, S)$.
In other words, a $T$-join is an edge set that is the disjoint union of a set of $T$-paths $\mathcal{P}$ that induce a matching on the vertices in $T$ and a set of cycles (see e.g., \cite{Schrijver}). 
We note that a $T$-join is (inclusion-wise) minimal if and only if it is acyclic. 
It follows that a minimal $T$-join decomposes (though not necessarily uniquely) into $|T|/2$ paths connecting disjoint pairs of vertices in $T$.
However, a $T$-join that decomposes into paths can still be non-minimal if, for instance, two paths intersect at two vertices and thereby form a cycle.
We will say that a $T$-join is \emph{semi-minimal} if it is the edge-disjoint union of $|T|/2$ walks between disjoint pairs of $T$.
Cygan et al.~\cite{CyganMPPS14euler} observed that an edge set $S$ with $|S| \le k$ is a solution for \textsc{Undirected Eulerian Edge Deletion} if and only if $S$ is a $T$-join and $G \setminus S$ is connected.
The algorithm of Goyal et al.~\cite{GoyalMPPS18even} employs a dynamic programming approach; there is an entry for every subset of $T' \subseteq T$ (which may have size up to $2k$), which stores $T$-walks between disjoint pairs of $T'$.
The number of such walks is unbounded in $k$, but the number of walks stored in the table can be reduced using representative sets of co-graphic matroids.

We give an $O^*(2^k)$-time (and polynomial-space) algorithm.
The improvements are twofold.
First, we avoid computing the representative families.
Second, we avoid dynamic programming over the subsets of $T$.
Let $X = \{ x_e \mid e \in E \}$ be a set of edge variables
and $Y = \{y_{tt'} \mid t, t' \in T\}$ a set of variables which will encode
the decomposition of a minimal $T$-join into $T$-paths. We use the convention
that $y_{tt'}$ and $y_{t't}$ denote the same variable. 

\begin{lemma} \label{lemma:undirected-eulerian}
  There is a polynomial $P(X,Y)$ that can be efficiently evaluated such that its terms that are multilinear of degree $k$ in $X$  enumerate all minimal $T$-joins $S$ of size $k$ and (not necessarily all) semi-minimal $T$-joins of size $k$. 
\end{lemma}
\begin{proof}
  Let $A$ be a symmetric matrix indexed by $V$, where $A[u, v] = x_{uv}$ if $uv \in E$ and $A[u, v] = 0$ otherwise.
  For every $v \in V$, let $e_v$ be the $|V|$-dimensional vector where $e_v[v] = 1$ and $e_v[v']= 0$ for each $v' \in V \setminus \{ v \}$.
  We define a skew-symmetric matrix $A'$ indexed by $T$, where for every~$u, v \in$~$T$,
  \begin{align*}
    A'[u, v] = \sum_{\ell \in [k]} e_u^T A^{\ell} e_v y_{uv},
  \end{align*}
  which enumerates all $(u,v)$-walks of length up to $k$ (with an extra term $y_{t,t'}$).
  Note that this is the \emph{unlabelled} walk polynomial, as opposed to the labelled walk polynomial defined in Section~\ref{sec:prel}.
  We claim that the degree-$k$ terms of $\Pf A'$ yield the desired polynomial.
  Recall that the Pfaffian enumerates all perfect matching on the complete graph on $T$.
  Thus, every multilinear term in the monomial expansion corresponds to a set of $T$-walks that connect disjoint pairs of $T$ with no edge occurring twice or more.
  Each multilinear term thus corresponds to a semi-minimal $T$-join in $G$.
  In the other direction, let $S$ be a minimal $T$-join, decomposed into paths as $S=E(\cP_1) \cup \ldots \cup E(\cP_t)$.
  Note that since $S$ is acyclic, every path $\cP_i$ is uniquely determined by its endpoints. 
  Let $F \subseteq \binom{T}{2}$ be the matching on $T$ induced by the decomposition,
  i.e., for every $i \in [t]$ there is an edge $e_i \in F$ on the endpoints of $\cP_i$.
  Then the monomial
  \[
    \prod_{i=1}^t y_{e_i} \prod_{e \in E(\cP_i)} x_e = \prod_{st \in F} y_{st} \prod_{e \in S} x_e
  \]
  is produced only exactly once in $\Pf A'$. 
\end{proof}

We solve \textsc{Undirected Eulerian Deletion} as follows.
Assume that there is no solution of size at most $k - 1$.
To find a solution of size exactly $k$, let $P_k(X,Y)$ be the part of $P(X,Y)$ that has degree $k$ in $X$, 
where $P(X,Y)$ is the polynomial defined in Lemma~\ref{lemma:undirected-eulerian}.
Note that $P_k(X,Y)$ can be evaluated via polynomial interpolation. 

We use the basis sieving algorithm (Theorem~\ref{theorem:sieve-basis}) over the cographic matroid $M_k$ truncated to rank~$k$.
Note that a multilinear term corresponding to a semi-minimal $T$-join vanishes during the sieving step, since we assumed that there is no solution of size $k - 1$ or smaller.
Thus, we sieve for minimal $T$-join $S$ such that $S$ is a basis for $M_k$, i.e., $G \setminus S$ is connected.
By iterating over all values up to $k$, we can find a minimum solution (or decide that no solution of size $k$ or smaller exists) in time $O^*(2^k)$.

\begin{theorem}
  \textsc{Undirected Eulerian Deletion} can be solved in $O^*(2^{k})$ time.
\end{theorem}

Next, we briefly discuss the \textsc{Directed Eulerian Deletion}.
Goyal et al.~\cite{GoyalMPPS18even} showed that an arc set $S$ with $|S| \le k$ is a minimal solution for \textsc{Directed Eulerian Deletion} if and only if $S$ is the union of $\ell$ arc-disjoint paths $\cP = \{ P_1, \dots, P_{\ell} \}$ such that (i) $G \setminus S$ is weakly connected and (ii) there are exactly $\deg^+(v) - \deg^-(v)$ paths starting at every $v$ with $\deg^+(v) > \deg^-(v)$ and $\deg^-(v) - \deg^+(v)$ paths ending at every $v$ with $\deg^-(v) > \deg^+(v)$, where $\ell = \frac{1}{2} \sum_{v \in V} |\deg^+(v) - \deg^-(v)|$.
We show the directed analogue of Lemma~\ref{lemma:undirected-eulerian}.
For bookkeeping, we modify the graph as follows. For every $v \in V$ with $\deg^+(v)>\deg^-(v)$, create
$\iota(v):=\deg^+(v)-\deg^-(v)$ new vertices $v_i^+$, $i \in [\iota(v)]$, and add an arc $v_i^+v$ for each of them.
Let $T^+$ be the union of all such vertices $v_i^+$. 
Similarly, for every $v \in V$ with $\deg^-(v)>\deg^+(v)$, create
$\iota(v) := \deg^-(v)-\deg^+(v)$ new vertices $v_i^-$, $i \in [\iota(v)]$, and add an arc $vv_i^-$ for each of them.
Let $T^-$ be the union of all such vertices $v_i^-$. 
Let $G'$ be the modified graph. We can now identify edge sets $S$ in $G$ that meet the balance requirement 
of a solution with $(T^+,T^-)$-flows in $G'$. Let $E_T$ be the edges incident to $T^+ \cup T^-$. 
For simplicity, we refer to a \emph{$(T^+,T^-)$-flow in $G$} as an edge set $S$ in $G$
such that $S \cup E_T$ is a $(T^+,T^-)$-flow in $G'$. Analogous to the undirected case,
a \emph{minimal $(T^+,T^-)$-flow} is a $(T^+, T^-)$-flow in $G$ which is inclusion-wise minimal,
and a \emph{semi-minimal $(T^+,T^-)$-flow} is a $(T^+,T^-)$-flow $S$ in $G$ such that $S \cup E_T$
decomposes into $(T^+,T^-)$-walks.

Let $X=\{x_e \mid e \in E\}$ be a set of edge variables; note that no edge variables are created
for the edge of $E_T$. Furthermore, introduce a second set of variables
$Y=\{y_{e,pq} \mid e \in E, p \in T^+, q \in T^-\}$ to keep track of the decomposition of a $(T^+,T^-)$-flow $S$ into paths.

\begin{lemma}
  \label{lemma:directed-eulerian}
  There is a polynomial $P(X,Y)$ that can be efficiently evaluated such that its terms which are
  multilinear of degree $k$ in $X$ enumerate all minimal $(T^+,T^-)$-flows $S$ in $G$ of size $k$,
  in addition to possibly some that are semi-minimal but not minimal.
\end{lemma}
\begin{proof}
  We define a $T^+ \times T^-$ matrix $A'$, where for $u_i^+ \in T^+$ and $v_j^- \in T^-$
  we let entry $A'[u_i^+, v_j^-]$ contain a polynomial enumerating all $(u_i^+,v_j^-)$-walks in $G'$ of length at most $k$, as in Lemma~\ref{lemma:undirected-eulerian}, except every variable $x_e$, $e \in E$ is multiplied by $y_{e,pq}$. Note, again, that we use the value $1$ rather than a variable $x_e$ for edges $e \in E_T$. 
  Let $P(X,Y) = \det A$.
  Then, each term in $P(X,Y)$ corresponds to a set of $\ell$ walks from $T^+$ to $T^-$ in $G'$, where the initial and final edges are shared 
  between all terms and thus can be ignored. Thus every monomial corresponds to a semi-minimal $(T^+,T^-)$-flow.
  Furthermore, for every minimal $(T^+,T^-)$-flow $S$ and every decomposition of $S$ into paths $\cP_1$, \ldots, $\cP_\ell$, 
  the resulting monomial is unique, since the $Y$-variables encode the decomposition of $S$, and given such a decomposition every path $\cP_i$ contains a unique spanning walk. 
  As in Lemma~\ref{lemma:undirected-eulerian}, the part of $P(X,Y)$ that is multilinear in $X$ of degree $k$
  then corresponds to semi-minimal $(T^+,T^-)$-flows in $G$ of size $k$, including all minimal flows, and can be evaluated using $\det A$ via interpolation. 
\end{proof}

As for \textsc{Undirected Eulerian Deletion}, using the sieving algorithm of Theorem~\ref{theorem:sieve-basis} over the cographic matroid of the underlying undirected graph, we obtain:

\begin{theorem}
  \textsc{Directed Eulerian Deletion} can be solved in $O^*(2^k)$ time.
\end{theorem}

\subsection{Conflict-free Solution}
\label{sec:conflict}

There has been a line of research studying ``conflict-free'' variants of classical problems \cite{AgrawalJKS20,DarmannPS09,DarmannPSW11,JacobMR21}.
Consider a problem in which we search for a solution $S$, which is a subset of the ground set $E$.
In the conflict-free version, the solution should form an independent (i.e., pairwise non-adjacent) set in $H$, where $H$ is an additionally given graph whose vertices are $E$ and whose edges are ``conflicts''.
Formally, let us define the problem \textsc{Conflict-free solution} as follows.
The input is a collection $\mathcal{F}$ of (possibly exponentially many) subsets of $E$, a conflict graph $H$ on $E$, and an integer $t$.
The problem asks there is a set $S \in \cF$ of size $t$ that forms an independent set in $H$.
We give another immediate consequence of Theorem~\ref{theorem:matroid-sieving} on \textsc{Conflict-free Solution}.
As a by-product, we improve on existing algorithms.
Since \textsc{Conflict-free Solution} is W[1]-hard in general, we restrict the input as follows.

\newcommand{\IS}{\text{IS}}

Let $P_H(X)$ over $\{ x_v \mid v \in V \}$ be an enumerating polynomial for independent sets in $H$, i.e., $P_H(X) = \sum_{I} c_I \prod_{v \in I} x_v$, where $I$ ranges over all independent sets of $H$ and $c_I \in \F$ is a constant.
Since it is NP-hard to determine the existence of an independent set of size $k$, the polynomial $P_H(X)$ cannot be efficiently evaluated, unless P = NP.
However, when $H$ is from a restricted graph, it can be.
Let $\mathcal{G}_{\IS}$ be a class of such graphs, i.e., let $\mathcal{G}_{\IS}$ contain all graphs $H$ such that $P_H(X)$ can be evaluated in
time $O(|V(H)|^c)$ for some constant $c$ using some fixed algorithm. 
One example of such a class $\mathcal{G}_{\IS}$ is then, for instance, chordal graphs (graphs that do not contain any cycle of length four or greater as an induced subgraph) as shown by Achlioptas and Zampetakis \cite{AchlioptasZ21}.
Moreover, we will say that a set family $\cF$ over $E$ is \emph{$k$-representable} if there is a matrix $A \in \F^{k \times \ell}$ over a field $\F$ of characteristic 2
such that every $e \in E$ is associated with a pairwise disjoint subset $\Gamma_e \subseteq [\ell]$, and for any $S \subseteq E$, $S \in \cF$ if and only if $A[\cdot, \bigcup_{e \in S} \Gamma_e]$ has full row rank (i.e., contains a non-singular submatrix).
By Theorem~\ref{theorem:matroid-sieving}, we have

\begin{theorem} \label{theorem:conflict-free}
  If $\cF$ is $k$-representable and $H \in \mathcal{G}_{\IS}$, then 
  \textsc{Conflict-free Solution} can be solved in $O^*(2^{k})$ time.
\end{theorem}

Theorem~\ref{theorem:conflict-free} gives the following improvements over existing dynamic programming algorithms:
\begin{itemize}
  \item
  \textsc{Conflict-free Matching}: Given a graph $G = (V, E)$, a conflict graph $H = (E, E')$, and an integer $k$, the task is to decide whether $G$ has a conflict-free matching of size $k$.
  Agrawal et al.~\cite{AgrawalJKS20} showed that \textsc{Conflict-free Matching} can be solved in $O^*(2^{(2\omega + 2)k})$ time when $H$ is chordal.
  We can solve this problem in $O^*(4^k)$ time because the collection of matchings of size $k$ is $2k$-representable:
  Let $A$ be the representation of the uniform matroid over $V$ (with every vertex $v$ copied $\deg(v)$ times) of rank $2k$.
  For each edge $uv$, let $\Gamma_{uv}$ contain one copy of the column for $u$ and one copy of the column for $v$.
  Then a set of $k$ edges spans $A$ if and only if it is pairwise disjoint, i.e., forms a matching.
  \item
  \textsc{Conflict-free Set Cover}: Given a collection $\cE$ of sets over $V$ (with $|V| = n$), a conflict graph $H = (\cE, E')$, and an integer $t$, the task is to decide whether $G$ has a conflict-free set cover $S \subseteq \cE$ (i.e., $\bigcup S = V$) of size at most $t$.
  Jacob et al.~\cite{JacobMR21} gave an $O^*(3^n)$-time algorithm for \textsc{Conflict-free Set Cover} when the conflict graph $H$ is chordal.
  Since set covers are $n$-representable (mapping every set $E$ to a list of elements $(v,E)$, $v \in E$ as in Section~\ref{sec:matroid-applications}),
  by Theorem~\ref{theorem:conflict-free}, it can be solved in $O^*(2^n)$ time.
\end{itemize}

Incidentally, Jacob et al.~\cite{JacobMR21} showed that \textsc{Conflict-free Set Cover} is W[1]-hard parameterized by $n$, even if the conflict graph is bipartite.
This implies that the class of bipartite graphs is not contained in $\mathcal{G}_{\IS}$, although a maximum independent set can be found in polynomial time on bipartite graphs.
It is perhaps no coincidence that the problem of counting independent sets is \#P-hard for bipartite graphs \cite{ProvanB83}.

\section{Conclusions} \label{sec:conc}

We have presented \emph{determinantal sieving}, a new powerful method
for algebraic exact and FPT algorithms that
extends the power of multilinear sieving with the ability to sieve for
terms in a polynomial that in addition to being multilinear are also
independent in an auxiliary linear matroid.
This yields significantly improved and generalized results for a range
of FPT problems, including \textsc{$q$-Matroid Intersection} in time $O^*(2^{(q-2)k})$
over a field of characteristic 2, improving on a previous result, of $O^*(4^{qk})$~\cite{BrandP21},
as well as algorithms solving problems over \emph{frameworks}
in the same running time as was previously known for the basic
existence problem (e.g., \textsc{Subgraph Isomorphism}).
Additionally, we showed that over fields of characteristic 2,
we can exploit cancellations in monomial expansion to
sieve for terms in a polynomial whose \emph{odd support} contains a
basis for the auxiliary matroid. 
This has further applications for a multitude of problems,
such as finding diverse solution collections and
for parameterized path, cycle and linkage problems.
Over fields of characteristic 2, 
all our algorithms are randomized and use polynomial space. 

Let us mention a few issues that we have not focused on in this paper.

\paragraph{Weighted problems.} As in most algebraic algorithms, we can
handle solution weights with a pseudopolynomial running time. That is,
given an algebraic sieving algorithm for some problem over a ground set $V$
with a running time of $O^*(f(k))$, and given a set of item weights
$\omega \colon V \to \N$ and a weight target $W$, we can usually find
solutions of weight $W$ in time $O^*(f(k) \cdot W)$
by multiplying every variable $x_v$, $v \in V$
by $z^{\omega(v)}$ for a new variable $z$ and using interpolation.
However, the gold standard would be to reduce the weight dependency
to $O^*(f(k) \cdot \log W)$ for the task of finding a min-weight solution,
and this appears incompatible with algebraic algorithms.
Even for the most classical problem \textsc{TSP}, whose unweighted
variant \textsc{Hamiltonicity} is solvable in $O^*(1.66^n)$ time~\cite{Bjorklund14detsum},
the $O^*(2^n)$-barrier has been broken recently, and even then,
only partially and conditionally~\cite{Nederlof20TSP}.

\paragraph{Counting.} Since our most efficient algorithms work over
fields of characteristic 2, they do not intrinsically allow us to
count the number of solutions. Indeed, for many settings relevant
to us, such as \textsc{$k$-Path} and bipartite or general matchings,
the corresponding counting problems are known to be hard (\#W[1]-hard
and \#P-hard, respectively; see Curticapean~\cite{Curticapean18counting} for a survey). 
On the other hand, being able to detect the existence of a witness
does have some applications for \emph{approximate} counting.
In particular, having access to a decision oracle for
\emph{colourful} witnesses, given a colouring of the ground set,
implies approximate counting algorithms~\cite{BhattacharyaBGM22stacs,DellLM22counting,DellL21}.
Improved, algebraically based FPT approximate counting algorithms are
also known for particular problems, such as \textsc{\#$k$-Path}~\cite{LokshtanovBSZ21countkpath}.

\paragraph{Derandomization.} 
The task of derandomizing our results ranges from doable with known
methods to completely infeasible, given the details of the
application.

In a typical application of our method, 
combining a polynomial $P(X)$ and a linear matroid $M$ over $X$,
we have two sources of randomness: Finding a representation of $M$
and the Schwartz-Zippel step of checking whether $P(X)$ is non-zero
(including avoiding cancellations due to interference between 
multilinear monomials in $P(X)$, representing different bases of $M$).
For many matroids and matroid constructions, a representation can be
found efficiently, including uniform matroids, partition matroids, and
graphic and co-graphic matroids, as well as matroids constructed from
these using operations of dualization, contraction, disjoint union and
truncation~\cite{LokshtanovMPS18,OxleyBook,Marx09-matroid}. 
Beyond this, there appears to be a barrier. A deterministic
representation of transversal matroids would presumably also lead to a
deterministic solution to \textsc{Exact Matching}, which is long open. 
Since gammoids generalize transversal matroids, and transversal
matroids can be constructed via a sequence of matroid union steps
over very simple matroids~\cite{OxleyBook}, gammoids and non-disjoint union
also appear difficult. However, some progress has been made on
constructing representations in superpolynomial time that depends on
the rank~\cite{LokshtanovMP0Z18repr,MisraPR020repr}.

For the Schwartz-Zippel PIT step, the obstacles are oddly similar. 
A polynomial $P(X)$ (over $\Q$ or $\R$) is \emph{combinatorial} if all
coefficients are non-negative. In such a case, PIT can be derandomized
via the exterior algebra; see Brand~\cite{Brand19thesis}.
This aligns with the extensor-based determinantal sieving used
in this paper, hence some of our results can be derandomized, albeit at the expense of increased running time.
More precisely, given a skew arithmetic circuit computing a combinatorial polynomial, 
we can sieve for a $k$-multilinear term whose support forms a 
matroid basis in $O^*(4^k)$ time, using the idea of lift mapping.
However, for results that depend on odd sieving, or where $P(X)$
corresponds to a determinant or Pfaffian computation, 
derandomization once again appears infeasible.

\subsection{Open questions}

Let us highlight some open questions.

One very interesting question is \textsc{Directed Hamiltonicity}.
We note two very different methods for checking Hamiltonicity
in bipartite digraphs in time $O^*(c^n)$ for $c<2$. The first is by 
Cygan, Kratsch and Nederlof~\cite{CyganKN18hamilton}, who established
a rank bound of $2^{n/2-1}$ on the \emph{perfect matching connectivity matrix}. 
This leads to a SETH-optimal algorithm for \textsc{Hamiltonicity}
parameterized by pathwidth, and an algorithm for
\textsc{Hamiltonicity} in bipartite digraphs in time $O^*(1.888^n)$.
However, the fastest algorithm for \textsc{Hamiltonicity} in bipartite
digraphs follows a polynomial sieving approach by Björklund, Kaski
and Koutis~\cite{BjorklundKK17directed}. In our terminology, we would
describe their algorithm as, given a bipartite digraph $G=(U \cup V,E)$,
enumerating subgraphs of $G$ that have in- and out-degree 1 in $V$
and whose underlying undirected graph is a spanning tree of $G$
plus one edge. It then remains to sieve via inclusion-exclusion for
those graphs which have non-zero in- and out-degree for every vertex
in $U$ as well, which they show can be done in time $O^*(3^{|U|})$ 
rather than $4^{|U|}$ due to the structure of the problem space.
Still, it remains unknown whether \textsc{Hamiltonicity} in general
digraphs can be solved in $O^*(c^n)$ for any $c<2$.

We would be very interested in a derivation of the $2^{n/2}$
rank bound for the perfect matching connectivity matrix
in a less problem-specific manner. Such a result, one would hope,
could uncover new tools and methods that could lead to improved
algebraic algorithms for a wider range of applications.
We also note, to the best of our
knowledge, that the optimal running time for \textsc{Hamiltonicity}
parameterized by treewidth remains open. Finally, can
\textsc{$k$-Path} be solved in time $O^*(c^k)$ for some $c<2$
on bipartite digraphs? 

Another major problem concerns \textsc{$k$ Disjoint Paths}
and its harder variant \textsc{Min-Sum $k$ Disjoint Paths}.
Given the success of algebraic algorithms for related problems,
it would be very interesting to find an algebraic algorithm for either problem
for general $k$. Björklund and Husfeldt~\cite{BjorklundH19twodp}
show an algebraic algorithm for \textsc{Min-Sum $k$ Disjoint Paths} for $k=2$, 
by showing a way to compute the permanent over $\mathbb{Z}_4[X]/(X^m)$, the ring
of bounded-degree polynomials over $\mathbb{Z}_4$. For the nearly ten years
since this result's original publication, we do not know of any developments
even for $k=3$. 

As a more down-to-earth problem, what is the best running time for
\textsc{$q$-Matroid Parity} and (possibly) \textsc{$q$-Set Packing}?
Recall that Björklund et al.~\cite{BjorklundHKK17narrow} showed that
\textsc{$q$-Dimensional Matching} can be solved in $O^*(2^{(q-2)k})$ time
and that \textsc{$q$-Set Packing} can be solved in time $O^*(2^{(q-\varepsilon_q)k})$
for some $\varepsilon_q > 0$, essentially by a randomized reduction to
\textsc{$q$-Dimensional Matching}.
Can \textsc{$q$-Matroid Parity}, the generalisation of \textsc{$q$-Set Packing},
be solved in $O^*(2^{(q-\varepsilon_q)k})$ time for some $\varepsilon_q > 0$?
The difference is most stark for $q=3$, where
\textsc{3-Matroid Intersection} is solvable in time $O^*(2^k)$, 
\textsc{3-Set Packing} in time $O^*(3.328^k)$
and \textsc{3-Matroid Parity} only in time $O^*(8^k)$. 

Among other individual problems of interest are to find improvements
and the best possible running times for problems such as 
\textsc{Long Directed Cycle}~\cite{FominLPS16JACM,Zehavi16ldc}
\textsc{Connected $f$-Factor}~\cite{FominGPS19sidma},
and more generally parameterized connectivity problems, cf.~\cite{AgrawalMPS22rSNDP,Feldmann0L22}.

\printbibliography

\end{document}